\documentclass[aps,
               prl,
               showpacs,
               amssymb,
               superscriptaddress,
               twocolumn,
               longbibliography]{revtex4-1}
\usepackage[ansinew]{inputenc}
\usepackage{bbold}
\usepackage{bbm}
\usepackage{bm}
\usepackage{amsbsy}
\usepackage{amsthm}
\usepackage{amssymb}
\usepackage{amsfonts}
\usepackage{amsmath, mathtools}
\usepackage{dsfont}
\usepackage{graphicx}
\usepackage{epsfig}
\usepackage{epstopdf}
\usepackage{dsfont}
\usepackage{mathrsfs}
\usepackage{multibib}
\usepackage{multirow}
\usepackage[dvipsnames]{xcolor}
\usepackage[colorlinks]{hyperref}
\makeatletter
\newcommand\org@hypertarget{}
\let\org@hypertarget\hypertarget
\renewcommand\hypertarget[2]{%
  \Hy@raisedlink{\org@hypertarget{#1}{}}#2%
  }
\makeatother
\usepackage[figure,table]{hypcap}
\usepackage{MnSymbol}
\usepackage{enumerate}
\usepackage{float}
\usepackage{comment}
\usepackage{braket}

\usepackage{subfigure}

\hypersetup{
	bookmarksnumbered,
	pdfstartview={FitH},
	citecolor={darkgreen},
	linkcolor={darkred},
	urlcolor={darkblue},
	pdfpagemode={UseOutlines}}
\definecolor{darkgreen}{RGB}{50,190,50}
\definecolor{darkblue}{RGB}{0,0,190}
\definecolor{darkred}{RGB}{238,0,0}
\usepackage{soul}

\newcommand{\nl}{\ensuremath{\hspace*{-0.5pt}}}
\newcommand{\nr}{\ensuremath{\hspace*{0.5pt}}}

\newcommand{\subtiny}[3]{\ensuremath{_{\hspace{#1 pt}\protect\raisebox{#2 pt}{\tiny{$ #3$}}}}}
\newcommand{\supscr}[3]{\ensuremath{^{\hspace{#1 pt}\protect\raisebox{#2 pt}{\scriptsize{$ #3$}}}}}
\newcommand{\suptiny}[3]{\ensuremath{^{\hspace{#1 pt}\protect\raisebox{#2 pt}{\tiny{$ #3$}}}}}

\DeclarePairedDelimiter\floor{\lfloor}{\rfloor}

\setstcolor{red}

\newcommand{\rank}{\text{rank}}
\newcommand{\identity}{\mathbbm{1}}

\newtheorem{theorem}{Theorem}
\newtheorem{corollary}{Corollary}
\newtheorem{lemma}{Lemma}
\newtheorem{proposition}{Proposition}
\newtheorem{observation}{Observation}

\theoremstyle{definition}

\newtheorem{definition}{Definition}
\newtheorem{problem}{Problem}

\newtheorem{remark}{Remark}

\newcommand{\tr}{\textnormal{Tr}}
\newcommand{\djj}{d\kern-0.4em\char"16\kern-0.1em}

\newcommand{\newreptheorem}[2]{\newtheorem*{rep@#1}{\rep@title}\newenvironment{rep#1}[1]{\def\rep@title{#2 \ref*{##1}}\begin{rep@#1}}{\end{rep@#1}}}

\newreptheorem{lemma}{Lemma}
\newreptheorem{remark}{Remark}

\renewcommand{\thesection}{\Roman{section}}
\renewcommand{\thesubsection}{\Roman{section}.\Alph{subsection}}
\renewcommand{\thesubsubsection}{\Roman{section}.\Alph{subsection}.\arabic{subsubsection}}
\makeatletter
\renewcommand{\p@subsection}{}
\renewcommand{\p@subsubsection}{}
\makeatother

\usepackage[framemethod=tikz]{mdframed}
\definecolor{mycolor}{rgb}{0.122, 0.435, 0.698}
\newmdenv[innerlinewidth=0.5pt, roundcorner=4pt,linecolor=mycolor,innerleftmargin=6pt,
innerrightmargin=6pt,innertopmargin=6pt,innerbottommargin=6pt]{mybox}
\usepackage{tabularx}
\usepackage{paracol}
\usepackage{tcolorbox}
\tcbuselibrary{theorems}

\newtcbtheorem{Definitions}{Definition}%
{colback=mycolor!5,colframe=mycolor,fonttitle=\bfseries,
left=4pt,
right=4pt,
top=5pt,
bottom=5pt}{defi}

\newtcbtheorem{Lemmas}{Lemma}%
{colback=green!5,colframe=green!50!blue,fonttitle=\bfseries,
left=4pt,
right=4pt,
top=5pt,
bottom=5pt
}{lemma}

\newtcbtheorem{BBoxes}{Box}%
{colback=mycolor!5,colframe=mycolor,fonttitle=\bfseries,
left=4pt,
right=4pt,
top=5pt,
bottom=5pt,float=ht!}{boxbox}

\newtcbtheorem{WideBoxes}{Box}%
{width=\textwidth,colback=mycolor!5,colframe=mycolor,fonttitle=\bfseries,
left=4pt,
right=4pt,
top=5pt,
bottom=5pt,float*=ht!}{boxbox}

\newtcolorbox[blend into=figures]{boxdefi}[3][]
{ float*=ht,width=\textwidth,lower separated=false, center upper,
title={#2},label= def:#3,#1}

\newtcolorbox[blend into=tables]{smallboxtable}[3][]
{colback=mycolor!5, colframe=mycolor, float*=ht, width=\textwidth, lower separated=false, blend before title=colon hang,
title={#2}, label= table:#3 ,#1}

\begin{document}
\title{
High-dimensional entanglement witnessed by correlations in arbitrary bases
}
\author{Nicky Kai Hong Li}
\affiliation{Atominstitut, Technische Universit\"{a}t Wien, Stadionallee 2, 1020 Vienna, Austria}
\affiliation{Vienna Center for Quantum Science and Technology, TU Wien, 1020 Vienna, Austria}
\affiliation{Institute for Quantum Optics and Quantum Information (IQOQI), Austrian Academy of Sciences, Boltzmanngasse 3, 1090 Vienna, Austria}
\author{Marcus Huber}
\affiliation{Atominstitut, Technische Universit\"{a}t Wien, Stadionallee 2, 1020 Vienna, Austria}
\affiliation{Institute for Quantum Optics and Quantum Information (IQOQI), Austrian Academy of Sciences, Boltzmanngasse 3, 1090 Vienna, Austria}
\author{Nicolai Friis}
\affiliation{Atominstitut, Technische Universit\"{a}t Wien, Stadionallee 2, 1020 Vienna, Austria}
\affiliation{Institute for Quantum Optics and Quantum Information (IQOQI), Austrian Academy of Sciences, Boltzmanngasse 3, 1090 Vienna, Austria}

\date{\today}

\begin{abstract}
Certifying entanglement is an important step in the development of many quantum technologies, especially for higher-dimensional systems, where entanglement promises increased capabilities for quantum communication and computation. A key feature distinguishing entanglement from classical correlations is the occurrence of correlations for complementary measurement bases. In particular, mutually unbiased bases (MUBs) are a paradigmatic example that is well-understood and routinely employed for entanglement certification. However, implementing unbiased measurements exactly is challenging and not generically possible for all physical platforms. Here, we extend the entanglement-certification toolbox from correlations in MUBs to arbitrary bases. This practically significant simplification paves the way for efficient characterizations of high-dimensional entanglement in a wide range of physical systems. Furthermore, we introduce a simple three-MUBs construction for all dimensions without using the Wootters-Fields construction, potentially simplifying experimental requirements when measurements in more than two MUBs are needed, especially in high-dimensional settings.
\end{abstract}
\maketitle

\section{Introduction}
Entanglement is an important signature of ``quantumness" and a central resource in quantum information processing. In particular, it is a crucial ingredient to achieving quantum advantages in many communication~\cite{BennettWiesner1992,BennettBrassardCrepeauJozsaPeresWootters1993,BuhrmanCleveMassarDeWolf2010,ScaraniEtAl2009,XuMaZhangLoPan2020}, metrological~\cite{DegenReinhardCappellaro2017,GiovannettiLloydMaccone2011}, and computational tasks~\cite{JozsaLinden2003}. Consequently, continuous efforts are being made to develop mathematical tools for the detection and quantification of entanglement in experiments~\cite{FriisVitaglianoMalikHuber2019}.

The certification of high-dimensional entanglement is of particular relevance to setups that use multilevel quantum systems to store and process information~\cite{BavarescoEtAl2018,HrmoEtAl2023}. Entanglement in higher dimensions can be more robust to noise and can allow for higher data throughput when used for teleportation, such that the performance of many quantum information-processing tasks improves with the dimension of the accessible entanglement resources. For example, using higher-dimensional entanglement can improve secure key rates in quantum key distribution~\cite{DodaHuberMurtaPivoluskaPleschVlachou2021} and benefit a wide range of quantum technologies such as entanglement-enhanced imaging~\cite{NdaganoDefienneLyonsStarshynovVillaTisaFaccio2020,DefienneNdaganoLyonsFaccio2021,CameronCourmeVernierePandyaFaccioDefienne2024}.

Moreover, entanglement certification can serve as a benchmark 
for quantum computers and simulators: If a device is supposed to output states with high-dimensional entanglement, then certifying the latter in the actual outputs can indicate how well the device is functioning~\cite{LanyonEtAl2009,HrmoEtAl2023}. 
At the same time, entanglement certification can increase one's confidence that a quantum advantage can be achieved since the hardness of classically simulating a many-body system increases with the amount of entanglement~\cite{AmicoFazioOsterlohVedral2008,VerstraeteMurgCirac2008,EisertCramerPlenio2010}.

A central intuition behind entanglement detection is that entanglement leads to correlations between outcomes of local measurements in two or more complementary bases. 
In this context, complementarity is typically approached via the extremal case of mutually unbiased bases (MUBs). For a quantum system prepared in any basis state of any one of these bases, all measurement outcomes for any of the other MUBs are equally likely, i.e., knowing the measurement outcome in one basis tells us nothing about the outcomes in the complementary bases. Consequently, MUBs have been at the centre of many existing entanglement-detection methods~\cite{SpenglerHuberBrierleyAdaktylosHiesmayr2012,ErkerKrennHuber2017,BavarescoEtAl2018,Herrera-ValenciaSrivastavPivoluskaHuberFriisMcCutcheonMalik2020,MorelliHuberTavakoli2023}. Correlations measured in MUBs can in turn be used to bound well-defined figures of merit that quantify how strongly entangled the underlying state is. An example of such a quantity is the Schmidt number~\cite{TerhalHorodecki2000}{\textemdash}a generalization of the Schmidt rank for mixed states{\textemdash}that is used here to quantify entanglement dimensionality.

However, while some systems allow one to freely select the measurement bases (e.g., via spatial light modulators and single-mode fibres for spatial degrees-of-freedom of photons~\cite{BavarescoEtAl2018,Herrera-ValenciaSrivastavPivoluskaHuberFriisMcCutcheonMalik2020}), this is generically not the case for all setups. 
The inability to measure in the desired MUBs often goes hand in hand with (but does not mathematically imply) the impossibility of carrying out tomographically complete sets of measurements. 
For previous approaches to certifying high-dimensional entanglement, such limited control has been prohibitively restrictive.  

Nevertheless, classical correlations cannot simultaneously be arbitrarily strong for any set of bases that are complementary in the sense of corresponding to non-commuting observables. 
Indeed, the detection of bipartite entanglement from measurements in arbitrary bases can, in principle, be achieved via entropic uncertainty relations~\cite{BertaChristandlColbeckRenesRenner2010,ColesBertaTomamichelWehner2017}, which provide a lower bound on the entanglement cost~\cite{DevetakWinter2005}, but not on the Schmidt number.

In this work, we fill this gap by proposing a family of Schmidt-number witnesses based on correlations in at least two coordinated local orthonormal bases which can be chosen arbitrarily. We provide analytic upper bounds of the corresponding witness operators evaluated on any bipartite state with Schmidt number at most~$k$ (Theorem~\ref{thm:SchmidtRkBound} \& Lemma~\ref{lemma:SchmidtRkLoose}). The main advantage of our method is that the upper bounds depend only on the absolute values of the basis-vector overlaps and are independent of the relative phases between the measurement bases, which are often not directly measurable in experiments. Therefore, our witness significantly simplifies the requirements for certifying high-dimensional entanglement in experiments across a wider range of platforms.
We also show that the bounds are the tightest when the underlying measurement bases are mutually unbiased (Corollary~\ref{cor:MUBoptimal}), confirming the intuition that MUBs are optimal for entanglement detection within this framework.\\[-3mm]

In addition, we analytically lower bound the \textit{entanglement fidelity} (or the \textit{singlet fraction} \cite{BennettDiVincenzoSmolinWootters1996, HorodeckiMPR1999}){\textemdash}the maximum fidelity of a bipartite state with any two-qudit maximally entangled state{\textemdash}using our witness (Theorem~\ref{thm:SchmidtRkBound} \& Lemma~\ref{lemma:SchmidtRkLoose}). This provides an alternative way to quantify high-dimensional entanglement. Next, we demonstrate the effectiveness of our Schmidt-number witnesses with two examples: the two-qudit isotropic state and the noisy two-qudit purified thermal states, and evaluate the difference between the actual entanglement fidelity and our bound for these examples (Sec.\;\ref{sec:examples} \& Appendix \ref{appendix:Examples}). To complete our analysis, we compare the white-noise tolerances of our witnesses with those proposed in Ref.~\cite{BavarescoEtAl2018} (Appendix \ref{app:JessicaWitness}). We also discuss the possibility of using random measurement bases in high dimensions (Sec.\;\ref{sec:randomBases} \& Appendix \ref{app:COM_Levy}) or approximately MUBs (AMUBs) (Appendix \ref{app:AMUBs}) to witness Schmidt numbers and propose a simple (and, to the best of our knowledge, new) construction of three MUBs for all dimensions without using the Wootters-Fields construction \cite{WoottersFields1989} (Sec.\;\ref{sec:3MUBsAnyD}). Finally, we compare various existing methods for certifying high-dimensional entanglement with our method in Table~\ref{table:CertifyHighDimEntMethods}.\\[-8mm]

\section{Results}
\subsection{Background \& notation}
To detect bipartite entanglement, parties A and B measure their shared state $\rho\subtiny{0}{0}{A\nl B}$ in~$m$ local bases with global projectors $|e^z_a\rangle\!\langle e^z_a|\otimes|\tilde{e}^z_a{}^*\rangle\!\langle\tilde{e}^z_a{}^*|$ where $\{\ket{e^z_a}\}_{a=0}^{d-1}$ is the $z$-th orthonormal basis of the $m$ bases, $\ket{\phi^*}$ denotes the complex conjugate of the state $\ket{\phi}$ with respect to the computational basis $\{\ket{i}\}_{i=0}^{d-1}$, and $\ket{\tilde{e}^z_a{}^*} \coloneqq U\ket{e^z_a{}^*}$ with $U\in U(d)$ fixed for all $a$ and $z$\@. Note that we do not require the $m$ measurement bases to be MUBs as in Ref.\;\cite{MorelliHuberTavakoli2023}. The entanglement witness is then defined to be the sum of the probabilities of all matching outcomes in all matching pairs of bases, i.e.,
\begin{equation}
    \mathcal{S}_{d}\suptiny{1}{0}{(m)}(\rho\subtiny{0}{0}{A\nl B}) = \sum_{z=1}^m\sum_{a=0}^{d-1} \langle e^z_a,\tilde{e}^z_a{}^*|\nr\rho\subtiny{0}{0}{A\nl B}\nr|e^z_a,\tilde{e}^z_a{}^*\rangle.
    \label{eq:witness}
\end{equation}
In Theorem~\ref{thm:SchmidtRkBound}, we show how the upper bound of Eq.\;\eqref{eq:witness} depends on the Schmidt number $k(\rho\subtiny{0}{0}{A\nl B})$~\cite{TerhalHorodecki2000} of the state $\rho\subtiny{0}{0}{A\nl B}$, which is defined as
\begin{equation}
    k(\rho\subtiny{0}{0}{A\nl B}) \coloneqq \inf_{\mathcal{D}(\rho\subtiny{0}{0}{A\nl B})}\left\{\max_{\{(p_i,\ket{\psi_i})\}_i} \rank(\tr\subtiny{0}{0}{B}|\psi_i\rangle\!\langle\psi_i|)\right\},
\end{equation}
where $\mathcal{D}(\rho)$ is the set of all pure-state decompositions, $\{(p_i,\ket{\psi_i})\}_i$, of $\rho=\sum_i p_i|\psi_i\rangle\!\langle\psi_i|$ and $\{p_i\}_i$ is a probability distribution. In addition, we show that the maximum fidelity of $\rho\subtiny{0}{0}{A\nl B}$ with any maximally entangled state,
\begin{align}
    \text{$\mathcal{F}$}(\rho\subtiny{0}{0}{A\nl B})\coloneqq \max_{U\subtiny{0}{0}{A}}\langle\Phi^+_d|(U\subtiny{0}{0}{A}\otimes\identity\subtiny{0}{0}{B})\rho\subtiny{0}{0}{A\nl B}(U\subtiny{0}{0}{A}\otimes\identity\subtiny{0}{0}{B})^\dagger|\Phi^+_d\rangle,
\end{align}
where the maximization is over all unitaries $U\subtiny{0}{0}{A}$ acting on subsystem $A$ and $\ket{\Phi^+_d}=\frac{1}{\sqrt{d}}\sum_{i=0}^{d-1}\ket{ii}$, can be lower bounded using the quantity $\mathcal{S}_{d}\suptiny{1}{0}{(m)}(\rho\subtiny{0}{0}{A\nl B})$\@. From now on, we call $\mathcal{F}(\rho\subtiny{0}{0}{A\nl B})$ the \textit{entanglement fidelity} (also known as the \textit{singlet fraction} in the case of qubits \cite{BennettDiVincenzoSmolinWootters1996, HorodeckiMPR1999}) of $\rho\subtiny{0}{0}{A\nl B}$\@.\\[-3mm]

Let us define the maximum and minimum overlaps between two bases $z$ and $z'$ as $c^{z,z'}_\text{max}=\max_{a,a'} |\langle e^z_a|e^{z'}_{a'}\rangle|^2$ and $c^{z,z'}_\text{min}=\min_{a,a'} |\langle e^z_a|e^{z'}_{a'}\rangle|^2$, respectively. We then define $\mathcal{C}=\{|\langle e^z_a|e^{z'}_{a'}\rangle|^2\}_{a,a',z\neq z'}$ ($\overline{\mathcal{C}}=\{(c^{z,z'}_\text{max},c^{z,z'}_\text{min})\}_{z\neq z'}$) to be the set that contains all (pairs of maximum and minimum) overlaps between any two different measurement bases.\\[-5mm]

\subsection{Schmidt-number witness \& entanglement-fidelity bound}
We now present our main results that use the expectation value $\mathcal{S}_{d}\suptiny{1}{0}{(m)}$ to infer lower bounds on the Schmidt number and entanglement fidelity of the state $\rho\subtiny{0}{0}{A\nl B}$\@.
\begin{theorem}\label{thm:SchmidtRkBound}
For any bipartite state $\rho\subtiny{0}{0}{A\nl B}$ of equal local dimension $d$ and Schmidt number at most~$k$, it holds that
\begin{equation}
    \mathcal{S}_{d}\suptiny{1}{0}{(m)}(\rho\subtiny{0}{0}{A\nl B})\leq \frac{k(m-\mathcal{T}(\mathcal{C}))}{d}+ \mathcal{T}(\mathcal{C}) \eqcolon \mathcal{B}_k, \label{eq:SchmidtRkBnd}
\end{equation}
where the upper bound $\mathcal{B}_k$ depends on the integers $d$ and $k$, the number of measurement bases $m$, and a quantity $\mathcal{T}(\mathcal{C})$ which depends on the set of bases overlaps. More specifically, $\mathcal{T}(\mathcal{C}) \coloneqq \min\{\lambda(\mathcal{C}), m\}$, $\lambda(\mathcal{C}) \coloneqq \frac{1}{2}\left(1+\sqrt{1+2d\sum_{z\neq z'}G^{z,z'}}\right)\geq 1$, and $G^{z,z'} \coloneqq 1-(d+1)c^{z,z'}_\text{min} + \frac{1}{d}\sum_{a,a'}|\langle e^z_a|e^{z'}_{a'}\rangle|^4$\@. Furthermore, the entanglement fidelity of $\rho\subtiny{0}{0}{A\nl B}$ can be lower bounded as follows:
\begin{equation}\label{ineq:Fbound}
    \mathcal{F}(\rho\subtiny{0}{0}{A\nl B})\geq \max\left\{0, \frac{\mathcal{S}_{d}\suptiny{1}{0}{(m)}(\rho\subtiny{0}{0}{A\nl B})-\mathcal{T}(\mathcal{C})}{m-\mathcal{T}(\mathcal{C})}\right\} \eqcolon \mathcal{F}_m.
\end{equation}
\end{theorem}

We sketch the proof of Theorem~\ref{thm:SchmidtRkBound} here and refer to Sec.\;\ref{appendix:Thm1} for the formal proof. 
\begin{proof}[Proof sketch]
   We define the witness operator $W=\sum_{z}\sum_{a} |e^z_a\rangle\!\langle e^z_a|\otimes|\tilde{e}^z_a{}^*\rangle\!\langle\tilde{e}^z_a{}^*|$ so that $\tr(W\rho\subtiny{0}{0}{A\nl B})=\mathcal{S}_{d}\suptiny{1}{0}{(m)}(\rho\subtiny{0}{0}{A\nl B})$\@. We notice that $W\ket{\widetilde{\Phi}^+_d}= m\ket{\widetilde{\Phi}^+_d}$ where $\ket{\widetilde{\Phi}^+_d} \coloneqq (\identity_d\otimes U)\ket{\Phi^+_d}$\@. Since $W$ is positive semi-definite, it has a spectral decomposition. We then prove that all eigenvalues of $W$ in the subspace orthogonal to $\ket{\widetilde{\Phi}^+_d}$ are upper bounded by $\lambda(\mathcal{C})$ by showing that $W^2 \leq W + d\sum_{z\neq z'} (c^{z,z'}_\text{min} \ket{\widetilde{\Phi}^+_d}\!\!\bra{\widetilde{\Phi}^+_d} + \frac{1}{2}G^{z,z'}\identity_{d^2})$\@. Since the maximum eigenvalue of $W$ is $m$, we get $W \leq (m-\mathcal{T}(\mathcal{C}))\ket{\widetilde{\Phi}^+_d}\!\!\bra{\widetilde{\Phi}^+_d} + \mathcal{T}(\mathcal{C})\identity_{d^2}$\@. Using $\bra{\widetilde{\Phi}^+_d}\rho\subtiny{0}{0}{A\nl B}\ket{\widetilde{\Phi}^+_d}\leq\mathcal{F}(\rho\subtiny{0}{0}{A\nl B})\leq \frac{k}{d}$ for all $\rho\subtiny{0}{0}{A\nl B}$ with Schmidt number at most~$k$~\cite{TerhalHorodecki2000}, we obtain Ineqs.\;\eqref{eq:SchmidtRkBnd} and \eqref{ineq:Fbound}.
\end{proof}

Theorem~\ref{thm:SchmidtRkBound} implies that if the measured quantity $\mathcal{S}_{d}\suptiny{1}{0}{(m)}(\rho\subtiny{0}{0}{A\nl B})$ exceeds $\mathcal{B}_k$ for $1\leq k\leq d-1$, then the Schmidt number of $\rho\subtiny{0}{0}{A\nl B}$ must be at least $k+1$\@. 
For a given set of local measurement bases $\{\{\ket{e^z_a}\}_{a=0}^{d-1}\}_{z=1}^m$, parties A and B can maximize the certified Schmidt number by choosing $U$ (their relative reference frame) that maximizes $\mathcal{S}_{d}\suptiny{1}{0}{(m)}(\rho\subtiny{0}{0}{A\nl B})$ since the upper bound in Ineq.\;\eqref{eq:SchmidtRkBnd} is independent of $U$.
Notice that if all the measurement bases are MUBs, i.e., $|\langle e^z_a|e^{z'}_{a'}\rangle|^2=\frac{1}{d}\;\forall\;a,a',z\neq z'$, then $\mathcal{T}(\mathcal{C})=1$ and the bound in Ineq.\;\eqref{eq:SchmidtRkBnd} coincides with the one in Ref.\;\cite{MorelliHuberTavakoli2023}. Since $\mathcal{T}(\mathcal{C})\geq 1$ for any $m\geq2$, we immediately arrive at Corollary~\ref{cor:MUBoptimal}.

\begin{corollary}\label{cor:MUBoptimal}
    The bound $\mathcal{B}_k$ in Ineq.\;\eqref{eq:SchmidtRkBnd} is the tightest for all $k<d$ when the $m$ measurement bases are MUBs.
\end{corollary}

In case we only have access to the set of maximum and minimum overlaps $\overline{\mathcal{C}}$ instead of all the overlaps $\mathcal{C}$, we can still bound the Schmidt number and the entanglement fidelity by loosening the bounds in Theorem~\ref{thm:SchmidtRkBound}. By maximizing the term $\sum_{a,a'}|\langle e^z_a|e^{z'}_{a'}\rangle|^4$ in Eq.\;\eqref{eq:SchmidtRkBnd} such that the overlaps are compatible with $\overline{\mathcal{C}}$, we obtain the following lemma which is proven in Sec.\;\ref{appendix:Lemma1}.

\begin{lemma}\label{lemma:SchmidtRkLoose}
For any bipartite state $\rho\subtiny{0}{0}{A\nl B}$ of equal local dimension $d$ and Schmidt number at most~$k$, it holds that
\begin{equation}
    \mathcal{S}_{d}\suptiny{1}{0}{(m)}(\rho\subtiny{0}{0}{A\nl B})\leq \frac{k(m-\overline{\mathcal{T}}(\overline{\mathcal{C}}))}{d}+ \overline{\mathcal{T}}(\overline{\mathcal{C}}) \eqcolon \overline{\mathcal{B}}_k, \label{eq:SchmidtRkBndLoose}
\end{equation}
where the upper bound $\overline{\mathcal{B}}_k$ depends on $d$, $k$, the number of measurement bases $m$, and a quantity $\overline{\mathcal{T}}(\overline{\mathcal{C}})$ which depends on the set of minimum and maximum bases overlaps. More specifically, $\overline{\mathcal{T}}(\overline{\mathcal{C}}) \coloneqq \min\{\overline{\lambda}(\overline{\mathcal{C}}), m\}$, $\overline{\lambda}(\overline{\mathcal{C}}) \coloneqq \frac{1}{2}\left(1+\sqrt{1+2d\sum_{z\neq z'}\overline{G}(c^{z,z'}_\text{max},c^{z,z'}_\text{min})}\right)\geq1$, $\overline{G}(c^{z,z'}_\text{max},c^{z,z'}_\text{min}) \coloneqq 1-(d+1)c^{z,z'}_\text{min} + \Omega^{z,z'}$, $\Omega^{z,z'} \coloneqq L^{z,z'}(c^{z,z'}_\text{max})^2 + (d-L^{z,z'}-1)(c^{z,z'}_\text{min})^2 + [1-L^{z,z'}c^{z,z'}_\text{max} - (d-L^{z,z'}-1)c^{z,z'}_\text{min}]^2$, and 
\begin{align}
    L^{z,z'} \coloneqq 
    \begin{cases} 
        \ \floor*{\frac{1-c^{z,z'}_\text{min}d}{c^{z,z'}_\text{max}-c^{z,z'}_\text{min}}} 
        &\text{if}\ \ c^{z,z'}_\text{max}>c^{z,z'}_\text{min}\,,\\
        \ d &\text{if}\ \ c^{z,z'}_\text{max} =c^{z,z'}_\text{min}\,.   
    \end{cases}
\end{align}
Furthermore, the entanglement fidelity of $\rho\subtiny{0}{0}{A\nl B}$ can be lower bounded as follows:
\begin{equation}\label{ineq:FboundLoose}
    \mathcal{F}(\rho\subtiny{0}{0}{A\nl B})\geq \max\left\{0, \frac{\mathcal{S}_{d}\suptiny{1}{0}{(m)}(\rho\subtiny{0}{0}{A\nl B})-\overline{\mathcal{T}}(\overline{\mathcal{C}})}{m-\overline{\mathcal{T}}(\overline{\mathcal{C}})}\right\} \eqcolon \overline{\mathcal{F}}_m.
\end{equation}
\end{lemma}

Similar to Theorem~\ref{thm:SchmidtRkBound}, if $\mathcal{S}_{d}\suptiny{1}{0}{(m)}(\rho\subtiny{0}{0}{A\nl B})>\overline{\mathcal{B}}_k$, the Schmidt number of $\rho\subtiny{0}{0}{A\nl B}$ must be at least $k+1$\@. Furthermore, if all the measurement bases are MUBs, i.e., $c^{z,z'}_\text{max}=c^{z,z'}_\text{min}=\frac{1}{d}\;\forall\;z,z'$, then $\overline{\mathcal{T}}(\overline{\mathcal{C}})=1$ and $\overline{\mathcal{B}}_k$ coincides with the bound in Ref.\;\cite{MorelliHuberTavakoli2023}. Since $\overline{\mathcal{T}}(\overline{\mathcal{C}})\geq 1$, MUBs give the tightest bounds $\overline{\mathcal{B}}_k$\@.\\[-5mm]

\subsection{Examples of witness violation}\label{sec:examples}
To illustrate that our method can verify Schmidt numbers and lower bound the entanglement fidelity, we first apply our witness and the fidelity bound to a standard benchmark for entanglement witnesses, i.e., isotropic states  $\rho\subtiny{0}{0}{A\nl B}\suptiny{1}{0}{\mathrm{iso}}= (1-p)\ket{\Phi^+_{d}}\!\!\bra{\Phi^+_{d}}+\frac{p}{d^2}\identity_{d^2}$, whose Schmidt number is $k+1$ if and only if the white-noise ratio $p$ satisfies $\frac{d(d-k-1)}{d^2-1} \leq p < \frac{d(d-k)}{d^2-1} \eqcolon p_\text{iso}\suptiny{1}{0}{(k)}$~\cite{TerhalHorodecki2000}. We compare this with the noise that our witness can tolerate until we can no longer witness the actual Schmidt number of $\rho\subtiny{0}{0}{A\nl B}\suptiny{1}{0}{\mathrm{iso}}$\@. Since
\begin{equation}\label{eq:Eg1}
    \mathcal{S}_{d}\suptiny{1}{0}{(m)}(\rho\subtiny{0}{0}{A\nl B}\suptiny{1}{0}{\mathrm{iso}}) = p\frac{m}{d} + (1-p)m,
\end{equation} 
for $\mathcal{S}_{d}\suptiny{1}{0}{(m)}(\rho\subtiny{0}{0}{A\nl B}\suptiny{1}{0}{\mathrm{iso}})$ to exceed the bound $\mathcal{B}_k$ in Theorem~\ref{thm:SchmidtRkBound}, the white-noise ratio must satisfy
\begin{equation}\label{ineq:noiseBnd}
    p < \frac{(m-\mathcal{T}(\mathcal{C}))(d-k)}{m(d-1)} \eqcolon p_{c,m}\suptiny{1}{0}{(k)}.
\end{equation}
In the case when $d+1$ MUBs exist and $m=d+1$, we see that $p_{c,m}\suptiny{1}{0}{(k)} = p_\text{iso}\suptiny{1}{0}{(k)}$ for all~$k$\@.\\[-3mm]

\begin{figure}[t!]
    \centering
    \includegraphics[width=\linewidth]{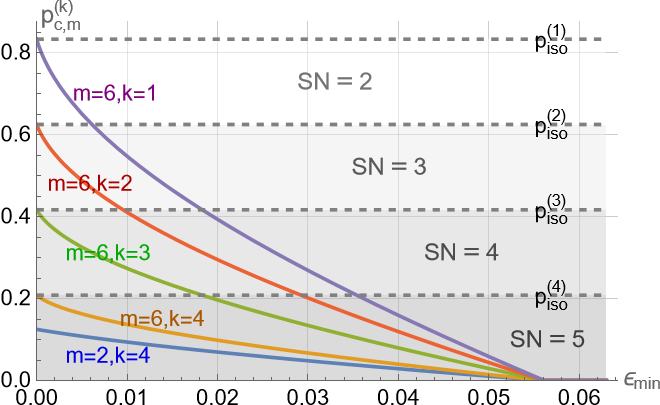}
    \caption{The upper bounds of the white-noise ratio, $p_{c,m}\suptiny{1}{0}{(k)}$ in Ineq.\;\eqref{ineq:noiseBnd}, for witnessing Schmidt number $k+1$ versus $\epsilon_\text{min} \coloneqq 1/d - c_\text{min}$ in local dimension $d=5$, where we set $c_\text{max} = 1-(d-1)c_\text{min}$\@. When $\epsilon_\text{min}=0$ (i.e., $c_\text{min}=1/d$ for MUBs), $p_{c,m}\suptiny{1}{0}{(k)}$ reaches its maximum, $(m-1)(d-k)/[m(d-1)]$\@. When $m=d+1=6$, it coincides with $p_\text{iso}\suptiny{1}{0}{(k)}$, the maximum white-noise ratio for $\rho\subtiny{0}{0}{A\nl B}\suptiny{1}{0}{\mathrm{iso}}$ having Schmidt number $k+1$\@. The noise tolerance of the witness is higher for larger $m$ or smaller~$k$ but reduces as $\epsilon_\text{min}$ increases, and eventually, when $\epsilon_\text{min} \geq 1/5-(7-\sqrt{17})/20\approx 0.0562$, we cannot witness non-trivial Schmidt numbers of $\rho\subtiny{0}{0}{A\nl B}\suptiny{1}{0}{\mathrm{iso}}$\@.}\label{fig:p_Bound}
\end{figure}
Suppose that we have the worst possible choice of measurement bases for a given $c_\text{min} \coloneqq \min_{z,z'} c^{z,z'}_\text{min}$ such that $c_\text{min} = c^{z,z'}_\text{min}$ and $c^{z,z'}_\text{max} = 1-(d-1)c_\text{min} \eqcolon c_\text{max}$ for all $z,z'$~\cite{footnote:Eg1WorstCase}. We can use the upper bound in Lemma~\ref{lemma:SchmidtRkLoose} by replacing $\mathcal{T}(\mathcal{C})$ in Ineq.\;\eqref{ineq:noiseBnd} with $\overline{\mathcal{T}}(\overline{\mathcal{C}})$, where $\overline{\lambda}(\overline{\mathcal{C}}) = \tfrac{1}{2}\bigl(1 + \sqrt{1+ 2dm(m-1) \overline{G}(c_\text{max},c_\text{min})}\bigr)$ in Ineq.\;\eqref{eq:SchmidtRkBndLoose}. To verify that our state has Schmidt number at least $k+1$, the bound $\overline{\mathcal{B}}_{k}$ must be violated. In Fig.\;\ref{fig:p_Bound}, the white-noise thresholds $p_{c,m}\suptiny{1}{0}{(k)}$ for witnessing the Schmidt number of $\rho\subtiny{0}{0}{A\nl B}\suptiny{1}{0}{\mathrm{iso}}$ in $d=5$ to be at least $k+1$ with $m$ measurement bases are plotted against the parameter $\epsilon_\text{min} \coloneqq \frac{1}{d} - c_\text{min}$, which quantifies the deviation of the measurement bases from MUBs. Note that the maximum value that $c_\text{min}$ can take is $\tfrac{1}{d}$ in order for $\sum_{a'}|\langle e^z_a|e^{z'}_{a'}\rangle|^2=1$ to hold for all $a,z$ and $z'$, so $\epsilon_\text{min}\in[0,\frac{1}{d}]$\@. When $\epsilon_\text{min}=0$, corresponding to measurements in MUBs, $p_{c,m}\suptiny{1}{0}{(k)}$ attains its maximum. In general, the witness can be violated under more white noise for smaller~$k$\@. In addition, as $p_{c,m}\suptiny{1}{0}{(k)}$ goes to zero when $\overline{\lambda}(\overline{\mathcal{C}})=m$, we see that whenever
\begin{equation}
    c_\text{min} \leq \frac{3\nr d-1-\sqrt{d^2+10\nr d-7}}{2d\nr (d-1)},\label{ineq:cminBoundForWitness}
\end{equation}
we cannot witness any Schmidt number $2\leq k+1\leq d$ of the state $\rho\subtiny{0}{0}{A\nl B}\suptiny{1}{0}{\mathrm{iso}}$ as $p_{c,m}\suptiny{1}{0}{(k)}=0$ for all $m\geq2$\@. To see how the Schmidt number is related to other standard measures of bipartite entanglement, we also compare the entanglement of formation (EoF) \cite{footnote:EoFiso}, the entanglement fidelity/singlet fraction (EF/SF), and the negativity \cite{footnote:NegativityIso} with the exact Schmidt number (SN) of the isotropic state of $d=5$ for different white-noise ratios $p$ in Fig.\;\ref{fig:NegEoF_iso}.\\[-3mm]

In Appendix \ref{app:isotropic}, we provide more elaborate analyses of our witness when applied to isotropic states. First, we show that the entanglement fidelity of $\rho\subtiny{0}{0}{A\nl B}\suptiny{1}{0}{\mathrm{iso}}$ satisfies Ineq.\;\eqref{ineq:Fbound} in Theorem~\ref{thm:SchmidtRkBound}. Then, we provide an example suggesting that our Schmidt-number witness can tolerate less bases bias in larger local dimensions. Finally, we compare the white-noise tolerance of our witness and our lower bound on the entanglement fidelity with the counterparts from Ref.\;\cite{BavarescoEtAl2018}.\\[-3mm]

To give a more comprehensive picture, we provide another example, i.e., the purified thermal states with white noise, in Appendix \ref{app:thermal} to demonstrate that our method also works in cases where the eigenvalues of the single-party reduced states are not degenerate and to compare the white-noise tolerance of our witness with that of Ref.\;\cite{BavarescoEtAl2018} in such cases. We also show that adding a third basis that is slightly biased with respect to two mutually unbiased measurement bases can increase our witness' tolerance to white noise in this example. On the other hand, adding a basis that is too biased with respect to the other bases could worsen our witnesses' performance due to an increased upper bound $\mathcal{B}_k$ ($\overline{\mathcal{B}}_k$) in Theorem~\ref{thm:SchmidtRkBound} (Lemma~\ref{lemma:SchmidtRkLoose}). We summarize this observation in the following remark.\\[-3mm]

\begin{remark}\label{remark1}
    There exist scenarios where an additional measurement basis improves the noise tolerance of our Schmidt-number witness. However, the opposite can also occur for certain choices of bases. Therefore, in order to witness the highest Schmidt number of a state, one should apply the witness inequality in Theorem~\ref{thm:SchmidtRkBound} or Lemma~\ref{lemma:SchmidtRkLoose} to all subsets of the total set of $m'$ available measurement bases and find the largest~$k$ such that $\mathcal{S}_{d}\suptiny{1}{0}{(m)}(\rho)>\mathcal{B}_k$ or $\overline{\mathcal{B}}_k$ when evaluated over all subsets of $m$ chosen bases for all $m\in\{2,\ldots,m'\}$\@.
\end{remark}
\begin{figure}[t!]
    \centering
    \includegraphics[width=\linewidth]{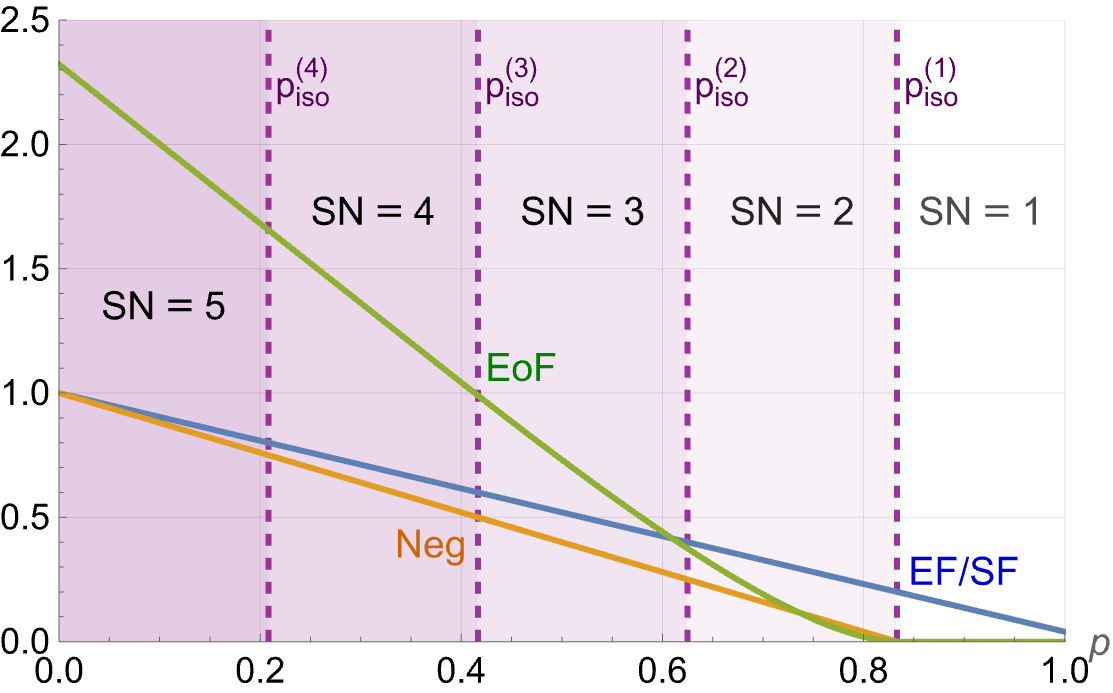}
    \caption{The entanglement of formation (EoF), the entanglement fidelity/singlet fraction (EF/SF), the negativity (Neg), and the exact Schmidt number (SN) of the isotropic state $\rho\subtiny{0}{0}{A\nl B}\suptiny{1}{0}{\mathrm{iso}}$ of $d=5$ are plotted for different white-noise ratios $p$.}\label{fig:NegEoF_iso}
\end{figure}

The intuition behind Remark~\ref{remark1} is that one can potentially certify a higher Schmidt number by post-selecting a subset of the total measurement data that achieves the optimal balance between showing the strongest measurement correlations and minimizing bias in the measurement bases. In practice, this can be realized easily by using as many local orthogonal measurement bases as possible for both parties and then calculating both the sum of expectation values $\mathcal{S}_{d}\suptiny{1}{0}{(m)}(\rho)$ and the bounds~$\mathcal{B}_k$ or~$\overline{\mathcal{B}}_k$ for each subset of measurement bases, using the corresponding subset of the full measurement data.\\[-3mm]

As a further remark, we observe that in $d=6$, adding a fourth basis to a set with three MUBs decreases the noise tolerance for witnessing Schmidt numbers in isotropic states for a wide range of choices for the additional basis. Since it is widely believed that the maximum number of MUBs in $d=6$ is 3 \cite{ColomerMortimerFrerotFarkasAcin2022}, this observation could indicate that our witness performs best when the local measurement bases consist only of the maximal set of MUBs in the given local dimension and no other bases.\\[-3mm]

\subsection{Implication of concentration of measure}\label{sec:randomBases}
We have seen from Corollary~\ref{cor:MUBoptimal} that measuring in MUBs will give the best Schmidt-number witness. However, requiring all local measurement bases to be MUBs is experimentally demanding as it requires precise control over the relative phases among all measurement bases. In light of this practical difficulty, it is natural to ask, how likely will a set of measurement bases chosen uniformly at random in $\mathbbm{C}^d$ be close to being mutually unbiased? Using L\'{e}vy's lemma~\cite{MilmanSchechtman1986,Ledoux2001,HaydenLeungWinter2006}, a result from concentration of measure, we show that the likelihood of any two randomly chosen orthonormal bases to be biased decreases exponentially with the dimension $d$, i.e., for $\epsilon>0$,
\begin{equation}
    \text{Pr}\left\{\Bigl|\nr\bigl|\langle e^{z}_{a}\nr|\nr e^{z'}_{a'}\rangle\bigr|^2-\tfrac{1}{d}\Bigr|>\epsilon\right\} \leq 2 \exp\left(-\frac{d\nr \epsilon^2}{18\pi^3\ln2}\right),\label{ineq:bases_COM}
\end{equation}
for all $a,a'$ and $z\neq z'$\@. Therefore, in large dimensions $d$, random measurement bases are likely to be sufficient for our method to witness high-dimensional entanglement. The proof of Ineq.\;\eqref{ineq:bases_COM} can be found in Appendix \ref{app:COM_Levy}.\\[-5mm]

\subsection{Maximal number of orthonormal bases}
Intuitively, one cannot construct arbitrarily many orthonormal bases when the maximal and minimal bases overlaps are specified. For example, it was known that there cannot be more than $d+1$ MUBs in $\mathbbm{C}^d$ (where $c^{z,z'}_\text{max}=c^{z,z'}_\text{min}= \frac{1}{d}$)~\cite{WoottersFields1989}. By making a connection to the Welch bounds~\cite{Welch1974}, we can upper bound the number of orthonormal bases using the function $\lambda(\mathcal{C})$ defined in Theorem~\ref{thm:SchmidtRkBound} such that
\begin{small}
\begin{equation}\label{ineq:mBound}
    m \leq \frac{d+1}{2}\left(1+\sqrt{1+\frac{8\lambda(\mathcal{C})\bigl(\lambda(\mathcal{C})-1\bigr)}{d^2-1}}\right) \eqcolon \overline{m}_d,
\end{equation}
\end{small}which is proven in Appendix \ref{app:onbNumBound}. For MUBs, we have $\lambda(\mathcal{C})=1$, so Ineq.\;\eqref{ineq:mBound} becomes $m\leq d+1$, which agrees with the known upper bound~\cite{WoottersFields1989}. In general, the bound does not imply the existence of $\overline{m}_d$ orthonormal bases. For instance, the existence of $d+1$ MUBs for non-prime-power dimensions $d$ is still an open problem~\cite{DurtEnglertBengtssonZyczkowski2010}. In Appendix \ref{app:isotropic}, we use this bound to show that Ineq.\;\eqref{ineq:Fbound} is satisfied for our example.\\[-3mm]

\begin{smallboxtable}{A comparison of different methods for the verification of high-dimensional entanglement.}{CertifyHighDimEntMethods}
\begin{center}
\renewcommand{\arraystretch}{1.4}
\begin{tabular}{|>{\centering\arraybackslash}m{0.225\textwidth}||>{\arraybackslash}m{0.755\textwidth}|
}
\hline
\begin{small} \textbf{Methods} \end{small} 	&  \multicolumn{1}{c|}{\textbf{Features \& Experimental/Computational Requirements}}\\
\hline
  & \,  $\bullet$ Measure in at least two coordinated local (arbitrary) orthonormal bases \,\\[-3pt]
  {Our witness \text{\hspace{100pt}}} & \, $\bullet$ Only need to know the (minimum and maximum) absolute values of the \newline \text{ }\quad overlaps between the local measurement bases \,\\
\hline
Ref.\;\cite{BavarescoEtAl2018}'s witness & \, $\bullet$ Precise control over the absolute values \& the complex phases of the overlaps \newline \text{ }\quad between different local measurement bases and within each basis\,\\[-3pt] 
& \, $\bullet$ Measure in the computational basis $+$ at least one coordinated ``tilted" basis\,\\
\hline
 & \,  $\bullet$ Treat bases bias as imperfect implementations of measurements in MUBs   \,\\[-2pt]
SDP witness~\cite{MorelliYamasakiHuberTavakoli2022} & \, $\bullet$ Memory issue/long computational runtime for large dimensions \,\\[-2pt]
 & \,  $\bullet$ Efficient only for small dimensions in which case it is possible to obtain tighter \newline \text{ }\quad bounds compared to our witness  \,\\
\hline
Entropic uncertainty\newline relationship~\cite{BertaChristandlColbeckRenesRenner2010,DevetakWinter2005} & \, $\bullet$ Need to know the absolute values of the bases overlaps of only one party and \newline \text{ }\quad the classical entropies corresponding to the two parties' measurements\\[-2pt] \relax
[Eq.\;(17.135) in~\cite{BertlmannFriis2023}] & \,  $\bullet$ Can lower bound the distillable entanglement instead of the Schmidt number\,\\
\hline
& \,  $\bullet$ To witness Schmidt number without any measurement assumptions, require non-\newline \text{ }\quad trivial optimization over all possible local measurements and all states with\newline \text{ }\quad Schmidt number $\leq k$ for every local dimension, which is computationally costly\\[-2pt]
Generalized Bell & \,  $\bullet$ Certify lower Schmidt numbers than other methods can in general\\[-2pt]
\text{ } inequalities~\cite{CollinsGisinLindenMassarPopescu2002,DadaEtAl2011} \newline \text{ }\newline \text{ }\newline \text{ }\newline \text{ } & \,  $\bullet$ Standard approaches involve finding the largest eigenvalue corresponding to\newline \text{ }\quad eigenvectors with Schmidt rank $\leq k$ of the Bell operator \cite{BraunsteinMannRevzen1992} associated with\newline \text{ }\quad restricted measurement settings \cite{AcinDurtGisinLatorre2002,DadaEtAl2011} to keep optimization problems\newline \text{ }\quad numerically tractable \cite{footnote:BellIneq}\;$\rightarrow$ also require measurement assumptions and are\newline \text{ }\quad computationally feasible only for small dimensions\\
\hline
& \, $\bullet$ Independent of the relative reference frame between the two parties~\cite{footnote:CorrelationMatrixRFinv} \,\\[-2pt]
Correlation-matrix & \,  $\bullet$ Require sampling local unitaries randomly from (Haar measure or) $t$-designs,\newline \text{ }\quad where exact sampling is highly inefficient~\cite{NakataZhaoEtAl2021,footnote:ApproxTdesignSample}\\[-2pt]
norms from randomized measurements~\cite{WyderkaKetterer2023} & \,\;$\bullet$ Analytic bounds of the 2- and 4-norms of the correlation matrices corresponding\newline \text{ }\quad to states with Schmidt number $\leq k$ are known only for $k=2$\\[-2pt]
& \,\;$\bullet$ For certifying Schmidt number $\geq3$, require numerical optimizations which can\newline \text{ }\quad be computationally costly for large dimensions and the bounds can be loose \cite{footnote:CorrMatrixConstraints}\\
\hline
\end{tabular}\\
\end{center}
\end{smallboxtable}
\text{ }\\[-11mm]

\subsection{Simple construction of three MUBs in any dimension}\label{sec:3MUBsAnyD}
As a by-product of investigating the use of AMUBs for witnessing high-dimensional entanglement (see Appendix \ref{app:AMUBs}), we discover a construction of three MUBs that has a simple analytic form and works for any dimension $d\in\mathbbm{N}$\@. The nice feature about this is that it does not rely on (the tensor products of) the Wootters{\textendash}Fields bases~\cite{WoottersFields1989}, which inevitably requires knowing the prime factorization of the dimension. Since factorizing a large integer is assumed to be hard (at least before any quantum device can properly implement Shor's algorithm)~\cite{Shor1994} and the description of the Wootters{\textendash}Fields bases can be non-trivial for large prime powers~\cite{WoottersFields1989}, constructing MUBs with the tensor products of the Wootters{\textendash}Fields bases will require a certain amount of computational overhead, whereas our construction does not suffer from these problems despite having only three MUBs.\\[-3mm] 

The explicit form of our three-MUBs construction is stated in the following lemma and its proof can be found in Appendix \ref{app:3MUBsAnyD}.
\begin{lemma}\label{lemma:3MUBsAnyD}
    For any $d\in \mathbbm{N}$, the three orthonormal bases $\{\ket{e^1_a}=\ket{a}\}_{a=0}^{d-1}$, $\{\ket{e^2_a}\}_{a=0}^{d-1}$, and $\{\ket{e^3_a}\}_{a=0}^{d-1}$, with
    \begin{subequations}
    \begin{flalign}
        \ket{e^2_a} & =\frac{1}{\sqrt{d}}\sum_{j=0}^{d-1} e^{i2\pi\left[\frac{aj}{d}+f(j)\right]}\ket{j},  \label{eq:3MUBsA}\\
        \ket{e^3_a} & =\frac{1}{\sqrt{d}}\sum_{j=0}^{d-1} e^{i2\pi\left[\frac{(d-p^r) j^2}{2d}+\frac{aj}{d}+f(j)\right]}\ket{j},\label{eq:3MUBsB}
    \end{flalign}
    \end{subequations}
    where $f$ is any real-valued function, $r\in\mathbbm{N}\cup\{0\}$, and $p$ is any odd prime such that $gcd(d,p)=1$ and $d>p^r$, are mutually unbiased. The simplest example would be having $p^r=1$.
\end{lemma}
Since the function $f$ (and to some degree, $p^r$) can be chosen freely, it can be optimized such that the relative phases in Eqs.\;\eqref{eq:3MUBsA} and \eqref{eq:3MUBsB} are more easily realizable for different experimental setups. This flexibility is particularly useful in cases where the dimension of the Hilbert space in which the experiment operates can change over time \cite{footnote:Advantage3MUBs} or the experiment has to measure different subsystems with distinct subspace dimensions at different times, as recalibration of the relative phases of our bases may require less drastic changes to the setup than using (the tensor product of) the Wootters-Fields bases \cite{footnote:RelPhaseChanges}.

\vspace*{-5mm}

\section{Discussion}
\vspace*{-2mm}
Our results provide a fresh perspective on the longstanding problem of detecting entanglement using only a few, potentially restricted, measurement settings. Specifically, we introduced a family of Schmidt-number witnesses based on correlations in at least two coordinated local orthonormal bases that can be chosen arbitrarily. We established analytic upper bounds for the corresponding witness operators when evaluated on any bipartite state with a Schmidt number of at most~$k$\@. The main advantage of our method is that the bounds depend solely on the absolute values of the overlaps between different measurement bases, but not on their relative phases, which are often inaccessible in experiments. These features of our witness simplify experimental requirements for certifying high-dimensional entanglement across many platforms. We demonstrated the effectiveness of our witness with two-qudit isotropic states and noisy purified thermal states. We also discussed the use of random measurement bases to witness Schmidt numbers. Finally, we compare our method with various existing approaches for certifying high-dimensional entanglement in Table~\ref{table:CertifyHighDimEntMethods}.\\[-3mm]

As Corollary~\ref{cor:MUBoptimal} suggests, one should aim at locally measuring in as many MUBs as possible to get the best performance of our witness. Sadly, the total number of MUBs is unknown for dimensions that are not prime powers~\cite{DurtEnglertBengtssonZyczkowski2010}. In fact, given the prime factorization of the dimension $d=\prod_j p_j^{n_j}$ with $p_j^{n_j}<p_{j+1}^{n_{j+1}}\;\forall\;j$, (tensor products of) the Wootters{\textendash}Fields construction only guarantees $p_1^{n_1}+1$ MUBs to exist~\cite{ColomerMortimerFrerotFarkasAcin2022}. Alternatively, if one can measure in bases that are nearly mutually unbiased, then one can construct $d+1$ AMUBs for any dimension $d$~\cite{ShparlinskiWinterhof2006}. We constructed Schmidt-number witnesses based on the AMUBs proposed in Ref.\;\cite{ShparlinskiWinterhof2006}, but we did not observe any advantage of measuring in $4\leq m\leq d+1$ AMUBs compared to measuring in three MUBs in the non-prime-powered dimensions $d=6,10,14,22$ (see Appendix \ref{app:AMUBs}). The discovery of the simple three-MUBs construction in any dimension (Lemma~\ref{lemma:3MUBsAnyD}) by modifying the AMUBs construction in Ref.\;\cite{ShparlinskiWinterhof2006} suggests that there could be other constructions of AMUBs that are more suitable for witnessing Schmidt numbers and we leave finding such bases as an open problem. Furthermore, given the flexibility of our three-MUBs construction, it may even contribute to answering a longstanding mathematical problem: Are there more than three MUBs in even, non-prime-power dimensions such as $d=6$?\\[-5mm]

\section{Methods}
\subsection{Proof of Theorem~\ref{thm:SchmidtRkBound}}\label{appendix:Thm1}

To prove Theorem~\ref{thm:SchmidtRkBound}, we start by stating the following propositions that we will need in the main proof.

\begin{proposition}
    It holds that $A\otimes\identity_d\ket{\Phi^+_d} = \identity_d\otimes A^T\ket{\Phi^+_d}$ for all $A\in M(\mathbbm{C}^d)$\@. 
    \label{prop:BellSymm}
\end{proposition}

\begin{proof}
    Let $A=\sum_{i,j=0}^{d-1} A_{ij}\ket{i}\!\!\bra{j}$\@. Then,
    \begin{subequations}
    \begin{flalign}
        A\otimes\identity_d\ket{\Phi^+_d} &= \tfrac{1}{\sqrt{d}}\sum_{i,j,k=0}^{d-1} A_{ij}\ket{i}\!\!\langle j|k\rangle\otimes\ket{k}\\
        &= \tfrac{1}{\sqrt{d}}\sum_{i,j=0}^{d-1} A_{ij}\ket{i}\otimes\ket{j},\nonumber
\end{flalign}
\end{subequations}\addtocounter{equation}{-1}\vspace{-6mm}
\begin{subequations}\addtocounter{equation}{1}
\begin{flalign}
    \identity_d\otimes A^T\ket{\Phi^+_d} &= \tfrac{1}{\sqrt{d}}\sum_{i,j,k=0}^{d-1} A_{ji}\ket{k}\otimes\ket{i}\!\!\langle j|k\rangle\\
    &= \tfrac{1}{\sqrt{d}}\sum_{i,j=0}^{d-1} A_{ji}\ket{j}\otimes\ket{i},\nonumber
\end{flalign}
\end{subequations}
so we indeed have $A\otimes\identity_d\ket{\Phi^+_d} = \identity_d\otimes A^T\ket{\Phi^+_d}$\@.
\end{proof}

\begin{proposition}\label{prop:maxEig}
    Let $\ket{\psi}$ and $\ket{\phi}$ be normalized states. Then, the eigenvalues of $\ket{\psi}\!\!\bra{\phi}+\ket{\phi}\!\!\bra{\psi}$ are upper bounded by $|\langle\psi|\phi\rangle|+1$ and lower bounded by $-(|\langle\psi|\phi\rangle|+1)$\@.
\end{proposition}

\begin{proof}
    Let $\ket{\phi}=a\ket{\psi}+b\ket{\psi^\perp}$ such that $|a|^2+|b|^2=1$, $\langle\psi|\psi^\perp\rangle=0$, and $\langle\psi|\psi\rangle=1=\langle\psi^\perp|\psi^\perp\rangle$\@. Then, $\ket{\psi}\!\!\bra{\phi}+\ket{\phi}\!\!\bra{\psi} =(a+a^*)|\psi\rangle\!\langle\psi|+b^*|\psi\rangle\!\langle\psi^\perp|+b|\psi^\perp\rangle\!\langle\psi|$, which has eigenvalues $\lambda_\pm = \text{Re}(a)\pm\sqrt{1-\text{Im}(a)^2}$\@. Since $\text{Re}(a)\leq|a|$ and $\sqrt{1-\text{Im}(a)^2}\leq1$, we have $|\lambda_\pm| \leq |a|+1 = |\langle\psi|\phi\rangle|+1$\@.
\end{proof}

We are now ready to state the formal proof of Theorem~\ref{thm:SchmidtRkBound} which follows similar arguments for proving Result 1 in Ref.\;\cite{MorelliHuberTavakoli2023} except here, the reference frame of party B relative to party A's is fixed by an arbitrary unitary and we cannot assume that $|\langle e^z_a|e^{z'}_{a'}\rangle|^2=\frac{1}{d}$ for all $z\neq z'$\@.

\begin{proof}[Proof of Theorem~\ref{thm:SchmidtRkBound}]
Let us define the witness operator 
\begin{align}
    W &=\sum_{z=1}^m\sum_{a=0}^{d-1} |e^z_a\rangle\!\langle e^z_a|\otimes|\tilde{e}^z_a{}^*\rangle\!\langle\tilde{e}^z_a{}^*|
\end{align} 
so that $\tr(W\rho\subtiny{0}{0}{A\nl B})=\mathcal{S}_{d}\suptiny{1}{0}{(m)}(\rho\subtiny{0}{0}{A\nl B})$, where $\ket{\tilde{e}^z_a{}^*} = U\ket{e^z_a{}^*}$ and $U$ is a fixed unitary that is the same for all $a$ and $z$\@. Via the definition of $\ket{\tilde{e}^z_a{}^*}$, we have $\sum_{a=0}^{d-1} \bigl(|e^z_a\rangle\!\langle e^z_a|\otimes|\tilde{e}^z_a{}^*\rangle\!\langle\tilde{e}^z_a{}^*|\bigr)(\identity_d\otimes U)\ket{\Phi^+_d} = \sum_{a=0}^{d-1} (|e^z_a\rangle\!\langle e^z_a|\otimes U|e^z_a{}^*\rangle\!\langle e^z_a{}^*|)\ket{\Phi^+_d} = \sum_{a=0}^{d-1} |e^z_a\rangle\!\langle e^z_a|e^z_a\rangle\!\langle e^z_a|\otimes U\ket{\Phi^+_d} = (\identity_d\otimes U)\ket{\Phi^+_d} \eqcolon \ket{\widetilde{\Phi}^+_d}$, where we have used Proposition~\ref{prop:BellSymm} in the second step. Hence, we have $W\ket{\widetilde{\Phi}^+_d}= m\ket{\widetilde{\Phi}^+_d}$\@. Since $W$ is positive semi-definite, it has a spectral decomposition: $W = m\ket{\widetilde{\Phi}^+_d}\!\!\bra{\widetilde{\Phi}^+_d} + \sum_{i=2}^{d^2}\lambda_i\ket{\lambda_i}\!\!\bra{\lambda_i}$ where all normalized eigenvectors $\ket{\lambda_i}$ are orthogonal to $\ket{\widetilde{\Phi}^+_d}$\@.\\[-3mm]

Next, we will derive an upper bound for all the eigenvalues $\lambda_i$ with $i=2,\ldots,d^2$\@. To do this, we consider the operator
\begin{flalign}
    W^2 &= W + \sum_{z\neq z'}\sum_{a,a'} |\langle e^z_a|e^{z'}_{a'}\rangle|^2 |e^z_a \tilde{e}^z_a{}^*\rangle\!\langle e^{z'}_{a'} \tilde{e}^{z'}_{a'}{}^*|\nonumber\\
    &= W + \sum_{z\neq z'} c^{z,z'}_\text{min} \mathcal{T}^{z,z'}_1 + \mathcal{T}_2\,,\label{eq:Wsqr}
\end{flalign}
with $\mathcal{T}^{z,z'}_1 = \sum_{a,a'} |e^z_a \tilde{e}^z_a{}^*\rangle\!\langle e^{z'}_{a'} \tilde{e}^{z'}_{a'}{}^*| = d \ket{\widetilde{\Phi}^+_d}\!\!\bra{\widetilde{\Phi}^+_d}$ since $V\otimes V^*$ with $V=\sum_{j=0}^{d-1} \ket{e^z_j}\!\!\bra{j}$ is a symmetry of $\ket{\Phi^+_d}$ by Proposition~\ref{prop:BellSymm}, and 
\begin{flalign}
    \mathcal{T}_2 &= \sum_{z\neq z'}\sum_{a,a'}(|\langle e^z_a|e^{z'}_{a'}\rangle|^2-c^{z,z'}_\text{min}) |e^z_a \tilde{e}^z_a{}^*\rangle\!\langle e^{z'}_{a'} \tilde{e}^{z'}_{a'}{}^*|\label{eq:T2L1}\\
    &= \sum_{z\neq z'}\sum_{a,a'}(|\langle e^z_a|e^{z'}_{a'}\rangle|^2-c^{z,z'}_\text{min}) \frac{1}{2}(|e^z_a \tilde{e}^z_a{}^*\rangle\!\langle e^{z'}_{a'} \tilde{e}^{z'}_{a'}{}^*|+\text{H.c.}),\nonumber
\end{flalign}
where we use the fact that $|\langle e^z_a|e^{z'}_{a'}\rangle|=|\langle e^{z'}_{a'}|e^z_a\rangle|\;\forall\;a,a',z,z'$ and ``H.c." stands for the Hermitian conjugate of the previous term. Then, we use (i) $|\langle e^z_a|e^{z'}_{a'}\rangle|^2 - c^{z,z'}_\text{min}\geq 0$, (ii) $|e^z_a \tilde{e}^z_a{}^*\rangle\!\langle e^{z'}_{a'} \tilde{e}^{z'}_{a'}{}^*|+\text{H.c.}\leq \lambda_\text{max}(|e^z_a \tilde{e}^z_a{}^*\rangle\!\langle e^{z'}_{a'} \tilde{e}^{z'}_{a'}{}^*|+\text{H.c.})\identity_{d^2}$, where $\lambda_\text{max}(A)$ denotes the largest eigenvalue of $A$, (iii) $\lambda_\text{max}(\sum_i A_i)\leq \sum_i\lambda_\text{max}(A_i)$ for all $A_i\in \text{Herm}(\mathbbm{C}^d)$~\cite{Bhatia1996}, and (iv) Proposition~\ref{prop:maxEig} to obtain
\begin{flalign}
    \mathcal{T}_2 &\leq \frac{1}{2}\sum_{z\neq z'}\sum_{a,a'}(|\langle e^z_a|e^{z'}_{a'}\rangle|^2-c^{z,z'}_\text{min})\nonumber\\ 
    &\qquad\qquad\quad\times\lambda_\text{max}(|e^z_a \tilde{e}^z_a{}^*\rangle\!\langle e^{z'}_{a'} \tilde{e}^{z'}_{a'}{}^*|+\text{H.c.})\identity_{d^2}\nonumber\\
    &\leq \frac{1}{2}\sum_{z\neq z'}\sum_{a,a'}(|\langle e^z_a|e^{z'}_{a'}\rangle|^2-c^{z,z'}_\text{min})(|\langle e^z_a|e^{z'}_{a'}\rangle|^2+1)\identity_{d^2}.\label{eq:T2L3}
\end{flalign}
Finally, we use the equality $\sum_{a'=0}^{d-1} |\langle e^z_a|e^{z'}_{a'}\rangle|^2 = 1$ to get 
\begin{flalign}
    \mathcal{T}_2 &\leq \frac{1}{2}\sum_{z\neq z'}\left[\sum_{a}(1-c^{z,z'}_\text{min}-c^{z,z'}_\text{min}d) + \sum_{a,a'}|\langle e^z_a|e^{z'}_{a'}\rangle|^4\right]\identity_{d^2}\nonumber\\
    &= \frac{d}{2}\sum_{z\neq z'}\left[1-(d+1)c^{z,z'}_\text{min} + \frac{1}{d}\sum_{a,a'}|\langle e^z_a|e^{z'}_{a'}\rangle|^4\right]\identity_{d^2}\nonumber\\
    &\eqcolon \frac{d}{2}\sum_{z\neq z'}G^{z,z'}\identity_{d^2}.\label{eq:T2L5}
\end{flalign}
After combining Eqs.\;\eqref{eq:Wsqr} and \eqref{eq:T2L5}, we obtain
\begin{equation}
    W^2 \leq W + d\sum_{z\neq z'} (c^{z,z'}_\text{min} \ket{\widetilde{\Phi}^+_d}\!\!\bra{\widetilde{\Phi}^+_d} + \frac{1}{2}G^{z,z'}\identity_{d^2}). \label{ineq:W2}
\end{equation}

Since $W^2 = m^2\ket{\widetilde{\Phi}^+_d}\!\!\bra{\widetilde{\Phi}^+_d} + \sum_{i=2}^{d^2}\lambda_i^2\ket{\lambda_i}\!\!\bra{\lambda_i}$ and due to Ineq.\;\eqref{ineq:W2}, we have that $\lambda_i^2 \leq \lambda_i + \frac{d}{2}\sum_{z\neq z'}G^{z,z'}$ for $i=2,\ldots,d^2$, which implies that $\lambda_i \leq \frac{1}{2}\left(1+\sqrt{1+2d\sum_{z\neq z'}G^{z,z'}}\right) \eqcolon \lambda(\mathcal{C})$\@. Since $G^{z,z'}\geq0\;\forall\;z\neq z'$ [see Ineq.\;\eqref{eq:T2L3}], $\lambda(\mathcal{C})\geq1$\@. On the other hand, we know that $\lambda_\text{max}(W) = \lambda_\text{max}(\sum_{z=1}^m\sum_{a=0}^{d-1} |e^z_a\rangle\!\langle e^z_a|\otimes|\tilde{e}^z_a{}^*\rangle\!\langle\tilde{e}^z_a{}^*|)\leq \sum_z\lambda_\text{max}(\sum_a |e^z_a\rangle\!\langle e^z_a|\otimes|\tilde{e}^z_a{}^*\rangle\!\langle\tilde{e}^z_a{}^*|) = m$, so $\lambda_i \leq \mathcal{T}(\mathcal{C}) \coloneqq \min\{\lambda(\mathcal{C}), m\}$\@. Therefore, with $W\ket{\widetilde{\Phi}^+_d}= m\ket{\widetilde{\Phi}^+_d}$, we have that
\begin{equation}
    W \leq (m-\mathcal{T}(\mathcal{C}))\ket{\widetilde{\Phi}^+_d}\!\!\bra{\widetilde{\Phi}^+_d} + \mathcal{T}(\mathcal{C})\identity_{d^2}.
\end{equation}
Finally, we arrive at our bound in Ineq.\;\eqref{eq:SchmidtRkBnd} in Theorem~\ref{thm:SchmidtRkBound}:
\begin{flalign}
    \mathcal{S}_{d}\suptiny{1}{0}{(m)}(\rho\subtiny{0}{0}{A\nl B}) = \tr(W\rho\subtiny{0}{0}{A\nl B}) &\leq (m-\mathcal{T}(\mathcal{C}))\mathcal{F}(\rho\subtiny{0}{0}{A\nl B})+\mathcal{T}(\mathcal{C})\nonumber\\[2mm]
    &\leq \frac{k(m-\mathcal{T}(\mathcal{C}))}{d}+\mathcal{T}(\mathcal{C}),\label{ineq:Sbound}
\end{flalign}
where we have used the fact that $\bra{\widetilde{\Phi}^+_d}\rho\subtiny{0}{0}{A\nl B}\ket{\widetilde{\Phi}^+_d}\leq\mathcal{F}(\rho\subtiny{0}{0}{A\nl B})\leq \frac{k}{d}$ for all bipartite state $\rho\subtiny{0}{0}{A\nl B}$ of equal local dimension $d$ and Schmidt number at most~$k$~\cite{TerhalHorodecki2000}. The fidelity lower bound in Ineq.\;\eqref{ineq:Fbound} can also be obtained by rearranging the first line of Ineq.\;\eqref{ineq:Sbound} and is non-trivial only if $\mathcal{T}(\mathcal{C}) < \mathcal{S}_{d}\suptiny{1}{0}{(m)}(\rho\subtiny{0}{0}{A\nl B})$\@. Otherwise, we set $\mathcal{F}_m=0$ as $\mathcal{F}(\rho\subtiny{0}{0}{A\nl B})\geq 0$ always holds.
\end{proof}

\subsection{Proof of Lemma~\ref{lemma:SchmidtRkLoose}}\label{appendix:Lemma1}\vspace*{-1.5mm}

To prove Lemma~\ref{lemma:SchmidtRkLoose}, we need the following proposition which is proven in Sec.\;\ref{app:prop_NonCVXopt}.

\begin{proposition}\label{prop:NonCVXopt}
The optimal solution to the optimization problem: max $\sum_{i=1}^{d^2} x_i^4$ subject to $\sum_{i=jd+1}^{jd+d} x_i^2=1$ for $j=0,\ldots,d-1$, and $0\leq\sqrt{c_\text{min}}\leq x_i\leq\sqrt{c_\text{max}}\leq 1\;\forall\;i$ is $d \{L c_\text{max}^2 + (d-L-1)c_\text{min}^2 + [1-L c_\text{max} -(d-L-1)c_\text{min}]^2\}$, where 
\begin{align}
    L &= \begin{cases} \floor*{\frac{1-c_\text{min}d}{c_\text{max}-c_\text{min}}} & \text{if}\ \ c_\text{max}>c_\text{min},\\
    d &\text{if}\ \ c_\text{max}=c_\text{min}.\end{cases}
\end{align}
\end{proposition}

\begin{proof}[Proof of Lemma~\ref{lemma:SchmidtRkLoose}]
We will prove that $\overline{\mathcal{T}}(\overline{\mathcal{C}})\geq \mathcal{T}(\mathcal{C})$ for all bases overlaps $\mathcal{C}$ so that 
\begin{align}
    \mathcal{S}_{d}\suptiny{1}{0}{(m)}(\rho\subtiny{0}{0}{A\nl B})
    &\leq\mathcal{B}_k=(1-\tfrac{k}{d})\mathcal{T}(\mathcal{C})+\tfrac{km}{d}\nonumber\\[1mm]
    &\leq (1-\tfrac{k}{d})\overline{\mathcal{T}}(\overline{\mathcal{C}})+\tfrac{km}{d}
\end{align} 
for all $k\leq d$ as in Ineq.\;\eqref{eq:SchmidtRkBndLoose}. This boils down to showing that $\sum_{a,a'}|\langle e^z_a|e^{z'}_{a'}\rangle|^4 \leq \Omega^{z,z'}d$ which implies $G^{z,z'}\leq \overline{G}(c^{z,z'}_\text{max},c^{z,z'}_\text{min}) \;\forall\;z,z'$\@.\\[-3mm]

For every pair of distinct bases $z,z'$, we want to maximize $\sum_{a,a'}|\langle e^z_a|e^{z'}_{a'}\rangle|^4$ given that $\sum_{a'}|\langle e^z_a|e^{z'}_{a'}\rangle|^2=1\;\forall\;a$ and $0\leq\sqrt{c_\text{min}}\leq |\langle e^z_a|e^{z'}_{a'}\rangle|\leq\sqrt{c_\text{max}}\leq 1\;\forall\;a,a'$\@. By setting $x_i=|\langle e^z_a|e^{z'}_{a'}\rangle|$ and $i=da+a'+1$, we can apply Proposition~\ref{prop:NonCVXopt} to obtain the maximum value, $\Omega^{z,z'}d$\@. Hence, $\overline{\mathcal{T}}(\overline{\mathcal{C}})\geq \mathcal{T}(\mathcal{C})$ and Ineq.\;\eqref{eq:SchmidtRkBndLoose} holds. Finally, the fidelity lower bound in Ineq.\;\eqref{ineq:FboundLoose} can be obtained in a similar fashion as in the proof of Theorem~\ref{thm:SchmidtRkBound}. The bound is non-trivial only if $\overline{\mathcal{T}}(\overline{\mathcal{C}}) < \mathcal{S}_{d}\suptiny{1}{0}{(m)}(\rho\subtiny{0}{0}{A\nl B})$\@. Otherwise, we set $\overline{\mathcal{F}}_m=0$ since $\mathcal{F}(\rho\subtiny{0}{0}{A\nl B})\geq 0$ always holds.
\end{proof}

\subsubsection{Proof of Proposition~\ref{prop:NonCVXopt}}\label{app:prop_NonCVXopt}

A general constrained optimization problem can be written in the following form~\cite{Bertsekas1999}.

\begin{problem}\label{prob:opt}
Let $\vec{x}\in \mathbbm{R}^d$\@. A constrained optimization problem can be written as
\begin{flalign}
\text{minimize}\quad& f(\vec{x})\nonumber\\
\text{subject to}\quad& h_1(\vec{x})=0,\ldots, h_m(\vec{x})=0,\\
& g_1(\vec{x})\leq0,\ldots,g_r(\vec{x})\leq0,\nonumber
\end{flalign}
where $f, h_i, g_j$ are functions mapping from $\mathbbm{R}^d$ to $\mathbbm{R}$\@. The feasible set $X\subset\mathbbm{R}^d$ is composed of all the $\vec{x}\in \mathbbm{R}^d$ that satisfy all the equality and inequality constraints. The associated Lagrangian of the problem is given by
\begin{equation}
    L(\vec{x},\vec{\lambda},\vec{\mu}) = f(\vec{x}) + \sum_{i=1}^m \lambda_i h_i(\vec{x}) + \sum_{j=1}^r \mu_j g_j(\vec{x}),
\end{equation}
where $\lambda_i,\mu_j\in \mathbbm{R}$ are Lagrange multipliers.
\end{problem}

It is useful to give the following definitions which we quote directly from Ref.\;\cite{Bertsekas1999} before stating Lemma~\ref{lemma:nonlin}.

\begin{definition}[Local minimum]
A vector $\vec{x}^*\in X$ is a local minimum of the objective function $f$ over the feasible set $X$ if there exists an $\epsilon>0$ such that $f(\vec{x}^*)\leq f(\vec{x})$ for all $\vec{x}\in X$ where $||\vec{x}-\vec{x}^*||<\epsilon$\@.
\end{definition}

\begin{definition}[Active constraints]
The set of active inequality constraints $A(\vec{x})$ at a point $\vec{x}\in X$ is the set of indices of the inequality constraints that are satisfied as equalities at $\vec{x}$, i.e., $A(\vec{x})=\{j|g_j(\vec{x})=0\}$\@.
\end{definition}

\begin{definition}[Regularity]
A feasible vector $\vec{x}$ is regular if the gradients of all the equality constraints $\nabla h_i(\vec{x}), i=1,\ldots,m$, and the gradients of all the active inequality constraints $\nabla g_j(\vec{x}), j\in A(\vec{x})$, are linearly independent.
\end{definition}

\begin{lemma}[Proposition~3.3.1 in Ref.\;\cite{Bertsekas1999}]
\label{lemma:nonlin}
Let $\vec{x}^*$ be a local minimum of Problem~\ref{prob:opt} where $f, h_i, g_j$ are continuously differentiable functions from $\mathbbm{R}^d$ to $\mathbbm{R}$, and assume that $\vec{x}^*$ is regular. Then, there exist unique Lagrange multiplier vectors $\vec{\lambda}^*=(\lambda_1^*,\ldots,\lambda_m^*)^T\in\mathbbm{R}^m$ and $\vec{\mu}^*=(\mu_1^*,\ldots,\mu_r^*)^T\in\mathbbm{R}^r$, such that
\begin{align}
&\nabla_x L(\vec{x}^*,\vec{\lambda}^*,\vec{\mu}^*)=\vec{0}\nonumber\\
&\mu_j^*\geq0, \quad j=1,\ldots,r,\\
&\mu_j^*=0, \quad \forall\;j\not\in A(\vec{x}^*),\nonumber
\end{align}
where $A(\vec{x}^*)$ is the set of active constraints at $\vec{x}^*$\@.
\end{lemma}

\begin{proof}[Proof of Proposition~\ref{prop:NonCVXopt}]
We will first translate our optimization problem into the form in Problem~\ref{prob:opt}, where we have $f(\vec{x})=-\sum_{i=1}^{d^2} x_i^4$, $h_{j+1}(\vec{x})=\sum_{i=jd+1}^{jd+d} x_i^2-1$, $g_{2jd+k+1}(\vec{x}) = x_{jd+k+1} - \sqrt{c_\text{max}}$, and $g_{(2j+1)d+k+1}(\vec{x}) = \sqrt{c_\text{min}} - x_{jd+k+1}$ for $j,k=0,\ldots,d-1$, where $0\leq c_\text{min} \leq c_\text{max} \leq 1$\@. This is in fact a sum of $d$ independent optimization problems of identical form, which simplifies our whole optimization problem to:
\begin{flalign}
\text{minimize}\quad& f(\vec{x})=-d\sum_{i=1}^{d} \nonumber x_i^4\label{prob:Actual}\\
\text{subject to}\quad& h(\vec{x})=\sum_{i=1}^d x_i^2-1=0,\\
& g_{k}(\vec{x}) = x_{k} - \sqrt{c_\text{max}} \leq 0,\nonumber\\
& g_{d+k}(\vec{x}) = \sqrt{c_\text{min}} - x_{k}\leq0, \;k=1,\ldots,d,\nonumber
\end{flalign}
where we redefine the objective function $f(\vec{x})$\@. The associated Lagrangian is given by
\begin{align}
    L(\vec{x},\lambda,\vec{\mu}) = \sum_{i=1}^{d} [-d x_i^4 + \lambda x_i^2 + (\mu_i-\mu_{i+d})x_i]\nonumber\\ 
    - \lambda - \sum_{i=1}^{d} (\sqrt{c_\text{max}}\mu_i - \sqrt{c_\text{min}}\mu_{i+d}),
\end{align}
where $\lambda,\mu_j\in \mathbbm{R}$\@. Let us determine which $\vec{x}$ is regular by evaluating the following gradients:
\begin{equation}
    \nabla h(\vec{x}) = \sum_i 2x_i\ket{i},\; \nabla g_i(\vec{x})=\ket{i},\; \nabla g_{i+d}(\vec{x})=-\ket{i}.
\end{equation}
If $g_i(\vec{x})\leq 0$ is active (i.e., $x_i=\sqrt{c_\text{max}}$), then $g_{i+d}(\vec{x})\leq 0$ for the same $i$ must be inactive unless $c_\text{min}=c_\text{max}$~\cite{footnote:irregular}, and vice versa. Also, if $\sqrt{c_\text{min}}<x_i<\sqrt{c_\text{max}}$, then both $g_i(\vec{x})\leq 0$ and $g_{i+d}(\vec{x})\leq 0$ are inactive. Hence, for $\vec{x}^*$ to be regular, the following must hold:
\begin{enumerate}[(i)]
    \item \label{cond:reg_i} At least one component $x_i^*$ must satisfy $\sqrt{c_\text{min}}<x_i^*<\sqrt{c_\text{max}}$ so that both $g_i(\vec{x}^*)\leq 0$ and $g_{i+d}(\vec{x}^*)\leq 0$ are inactive.
    \item \label{cond:reg_ii} In case $c_\text{min}=0$, if at least one component $x_i^*=\sqrt{c_\text{min}}=0$ so that the $i$-th component of $\nabla h(\vec{x}^*)$ is 0 while $\nabla g_{i+d}(\vec{x}^*)=-\ket{i}$, then condition (i) is not necessary. If $x_i^*>0\;\forall\;i$, condition (i) must hold.
\end{enumerate}

Since the functions $f, h, g_i, g_{i+d}$ for all $i=1,\ldots,d$ are continuously differentiable, all regular local minima of Problem \eqref{prob:Actual} must satisfy Lemma~\ref{lemma:nonlin}. Let us evaluate
\begin{equation}
    \nabla_x L = \sum_{i=1}^d (-4d x_i^3 + 2\lambda x_i +\mu_i-\mu_{i+d})\ket{i},
\end{equation}
and consider the case where the $i$-th component of a regular local minimum $\vec{x}^*$ satisfies $\sqrt{c_\text{min}}<x_i^*<\sqrt{c_\text{max}}$\@. By Lemma~\ref{lemma:nonlin}, $\mu_i^*=\mu_{i+d}^*=0$, implying that $-4d (x_i^*)^3 + 2\lambda^* x_i^* =0$\@. Hence,
\begin{equation}\label{eq:lambdastar}
    \lambda^* = 2d(x_i^*)^2.
\end{equation}
Since $\lambda^*$ is a parameter independent of the index $i$, we can constrain all components of $\vec{x}^*$ that lie within the interval, $(\sqrt{c_\text{min}},\sqrt{c_\text{max}})$, with this common parameter. According to Eq.\;\eqref{eq:lambdastar}, $x_i^* = x_j^* = \sqrt{\lambda^*/(2d)}$ for all $i,j$ such that $\sqrt{c_\text{min}}<x_i^*, x_j^*<\sqrt{c_\text{max}}$\@. Hence, for any regular local minimum $\vec{x}^*$ of the optimization problem \eqref{prob:Actual}, each component $x_i^*$ can only take one of the three possible values: $\sqrt{c_\text{min}}, \sqrt{\chi}, \sqrt{c_\text{max}}$, where $\chi\in(c_\text{min},c_\text{max})$\@.\\[-3mm]

Therefore, we can translate our optimization problem of~\eqref{prob:Actual} into a much simplified form:
\begin{flalign}
\text{maximize}\quad& d [L c_\text{max}^2 + \overline{L}\chi^2 + (d-L-\overline{L})c_\text{min}^2] \nonumber\\
\text{subject to}\quad& L c_\text{max} + \overline{L}\chi + (d-L-\overline{L})c_\text{min}=1,\label{eq:Lconstraint}\\
& 0\leq L+\overline{L} \leq d, \quad L, \overline{L} \in \mathbbm{N}, \nonumber\\
& 0\leq c_\text{min}<\chi<c_\text{max}\leq 1,\nonumber
\end{flalign}
where we converted our problem back to a maximization problem. Clearly, the optimal solution can be obtained by maximizing $L$ while satisfying all constraints, including Eq.\;\eqref{eq:Lconstraint} which can be rearranged into:
\begin{equation}
    L=\frac{1-c_\text{min}d}{c_\text{max}-c_\text{min}} - \overline{L}\;\frac{\chi - c_\text{min}}{c_\text{max}-c_\text{min}}.
\end{equation}
Since the last fraction lies within the interval $(0,1)$, the maximum value allowed for $L$ is $\floor{\frac{1-c_\text{min}d}{c_\text{max}-c_\text{min}}}$, leaving either $\overline{L}=0$ if $\frac{1-c_\text{min}d}{c_\text{max}-c_\text{min}}\in\mathbbm{N}$, or $\overline{L}=1$ with $\chi=1- L c_\text{max}- (d-L-1)c_\text{min}$ otherwise.

Next, we consider the remaining cases: (1) regular local minima $\vec{x}^*$ with no components satisfying $\sqrt{c_\text{min}}<x_i^*<\sqrt{c_\text{max}}$ when $c_\text{min}=0$, and (2) non-regular points $\vec{x}$ with each component $x_i\in\{\sqrt{c_\text{min}}>0, \sqrt{c_\text{max}}\}$ [see conditions \eqref{cond:reg_i} and \eqref{cond:reg_ii}]. Note that the case where $c_\text{min}=c_\text{max}$ is included here~\cite{footnote:irregular}. Similar to the previous regular cases, we can simplify our optimization problem to:
\begin{flalign}
\text{maximize}\quad& d [L c_\text{max}^2 + (d-L)c_\text{min}^2] \nonumber\\
\text{subject to}\quad& L c_\text{max} + (d-L)c_\text{min}=1,\label{eq:LconstraintIrreg}\\
& 0\leq L\leq d, \; L\in \mathbbm{N},\; 0\leq c_\text{min}\leq c_\text{max}\leq 1.\nonumber
\end{flalign}
If $c_\text{min}<c_\text{max}$, the problem has a feasible optimum only if $L=\frac{1-c_\text{min}d}{c_\text{max}-c_\text{min}}\in\mathbbm{N}$\@. If $c_\text{min}=c_\text{max}$, the problem has a feasible solution only if $c_\text{min}=c_\text{max}=\frac{1}{d}$\@.\\[-3mm]

Finally, since the global optimum to our initial Problem \eqref{prob:Actual} is the minimum over the set composed of all feasible points fulfilling Lemma~\ref{lemma:nonlin} (a set containing all regular local minima) together with all irregular feasible solutions, we can conclude that the global optimum $\vec{x}^*$ satisfies: $x_i^*=\sqrt{c_\text{max}}$ for $i=1,\ldots,L$, $x_j^*=\sqrt{c_\text{min}}$ for $j=L+1,\ldots,d-1$, and $x_d^*=\sqrt{1-L c_\text{max} - (d-L-1)c_\text{min}}$, where 
\vspace*{-1.5mm}
\begin{align}
    L = \begin{cases} \floor*{\frac{1-c_\text{min}d}{c_\text{max}-c_\text{min}}} &\text{if}\ \ c_\text{max}>c_\text{min},\\
d &\text{if}\ \ c_\text{max}=c_\text{min}.\end{cases}
\end{align}
\end{proof}


\noindent\textbf{Acknowledgments}.\
We thank Dimpi Thakuria, Paul Erker, and Alexandra Bergmayr for insightful discussions about practical challenges in experimental entanglement certification. We also thank Matej Pivoluska for the helpful discussion regarding Lemma 1 in Ref.\;\cite{Herrera-ValenciaSrivastavPivoluskaHuberFriisMcCutcheonMalik2020} (Proposition \ref{prop:GaussSum}). In addition, we thank Christopher J.~Turner for pointing out that the phase freedom in our original 3-MUBs construction can be further generalized (Lemma~\ref{lemma:3MUBsAnyD}) and Richard Kueng for making us aware of the connection between our 3-MUBs construction and properties of stabilizer states in Refs.~\cite{Gross2006,KuengGross2015}. This research was funded in whole or in part by the Austrian Science Fund (FWF) [\href{https://doi.org/10.55776/P36478}{10.55776/P36478}]. For open access purposes, the author has applied a CC BY public copyright license to any author accepted manuscript version arising from this submission. 
We further acknowledge support 
from the Austrian Federal Ministry of Education, Science and Research via the Austrian Research Promotion Agency (FFG) through the flagship project FO999897481 (HPQC), the project FO999914030 (MUSIQ), and the project FO999921407 (HDcode) funded by the European Union{\textemdash}NextGenerationEU, from the European Research Council (Consolidator grant `Cocoquest' 101043705), and the Horizon-Europe research and innovation programme under grant agreement No 101070168 (HyperSpace).\\[-3mm]

\noindent\textbf{Data Availability}.\ 
Data sharing is not applicable to this article as no datasets were generated or analysed during the current study.\\[-3mm]

\noindent\textbf{Author Contributions}.\
N.\;K.\;H.\;L.~discovered and produced all the results in this work. N.\;F.~provided the ideas that led to the consideration of the examples with purified thermal states mixed with white noise, the analytic comparison of our witnesses with the ones of Ref.\;\cite{BavarescoEtAl2018}, and random measurement bases. M.\;H.~provided the idea that led to the study of applicability of AMUBs. All authors discussed the results and contributed to the writing of the final manuscript.\\[-3mm]

\noindent\textbf{Competing Interests}.\
The authors declare no competing interests.


\bibliographystyle{apsrev4-1fixed_with_article_titles_full_names_new}
\bibliography{Master_Bib_File}

\begin{thebibliography}{79}%
\makeatletter
\providecommand \@ifxundefined [1]{%
 \@ifx{#1\undefined}
}%
\providecommand \@ifnum [1]{%
 \ifnum #1\expandafter \@firstoftwo
 \else \expandafter \@secondoftwo
 \fi
}%
\providecommand \@ifx [1]{%
 \ifx #1\expandafter \@firstoftwo
 \else \expandafter \@secondoftwo
 \fi
}%
\providecommand \natexlab [1]{#1}%
\providecommand \enquote  [1]{#1}%
\providecommand \bibnamefont  [1]{#1}%
\providecommand \bibfnamefont [1]{#1}%
\providecommand \citenamefont [1]{#1}%
\providecommand \href@noop [0]{\@secondoftwo}%
\providecommand \href [0]{\begingroup \@sanitize@url \@href}%
\providecommand \@href[1]{\@@startlink{#1}\@@href}%
\providecommand \@@href[1]{\endgroup#1\@@endlink}%
\providecommand \@sanitize@url [0]{\catcode `\\12\catcode `\$12\catcode `\&12\catcode `\#12\catcode `\^12\catcode `\_12\catcode `\%12\relax}%
\providecommand \@@startlink[1]{}%
\providecommand \@@endlink[0]{}%
\providecommand \url  [0]{\begingroup\@sanitize@url \@url }%
\providecommand \@url [1]{\endgroup\@href {#1}{\urlprefix }}%
\providecommand \urlprefix  [0]{URL }%
\providecommand \Eprint [0]{\href }%
\providecommand \doibase [0]{https://doi.org/}%
\providecommand \selectlanguage [0]{\@gobble}%
\providecommand \bibinfo  [0]{\@secondoftwo}%
\providecommand \bibfield  [0]{\@secondoftwo}%
\providecommand \translation [1]{[#1]}%
\providecommand \BibitemOpen [0]{}%
\providecommand \bibitemStop [0]{}%
\providecommand \bibitemNoStop [0]{.\EOS\space}%
\providecommand \EOS [0]{\spacefactor3000\relax}%
\providecommand \BibitemShut  [1]{\csname bibitem#1\endcsname}%
\let\auto@bib@innerbib\@empty
\bibitem [{\citenamefont {Bennett}\ and\ \citenamefont {Wiesner}(1992)}]{BennettWiesner1992}%
  \BibitemOpen
  \bibfield  {author} {\bibinfo {author} {\bibfnamefont {Charles~H.}\ \bibnamefont {Bennett}}\ and\ \bibinfo {author} {\bibfnamefont {Stephen~J.}\ \bibnamefont {Wiesner}},\ }\emph {\enquote {\bibinfo {title} {{Communication via one- and two-particle operators on Einstein-Podolsky-Rosen states}},}\ }\href {https://doi.org/10.1103/PhysRevLett.69.2881} {\bibfield  {journal} {\bibinfo  {journal} {Phys. Rev. Lett.}\ }\textbf {\bibinfo {volume} {69}},\ \bibinfo {pages} {2881{\textendash}2884} (\bibinfo {year} {1992})}\BibitemShut {NoStop}%
\bibitem [{\citenamefont {Bennett}\ \emph {et~al.}(1993)\citenamefont {Bennett}, \citenamefont {Brassard}, \citenamefont {Cr{\'e}peau}, \citenamefont {Jozsa}, \citenamefont {Peres},\ and\ \citenamefont {Wootters}}]{BennettBrassardCrepeauJozsaPeresWootters1993}%
  \BibitemOpen
  \bibfield  {author} {\bibinfo {author} {\bibfnamefont {Charles~H.}\ \bibnamefont {Bennett}}, \bibinfo {author} {\bibfnamefont {Gilles}\ \bibnamefont {Brassard}}, \bibinfo {author} {\bibfnamefont {Claude}\ \bibnamefont {Cr{\'e}peau}}, \bibinfo {author} {\bibfnamefont {Richard}\ \bibnamefont {Jozsa}}, \bibinfo {author} {\bibfnamefont {Asher}\ \bibnamefont {Peres}}, \ and\ \bibinfo {author} {\bibfnamefont {William~K.}\ \bibnamefont {Wootters}},\ }\emph {\enquote {\bibinfo {title} {{Teleporting an Unknown Quantum State via Dual Classical and {E}instein-{P}odolsky-{R}osen Channels}},}\ }\href {https://doi.org/10.1103/PhysRevLett.70.1895} {\bibfield  {journal} {\bibinfo  {journal} {Phys. Rev. Lett.}\ }\textbf {\bibinfo {volume} {70}},\ \bibinfo {pages} {1895\textendash1899} (\bibinfo {year} {1993})}\BibitemShut {NoStop}%
\bibitem [{\citenamefont {Buhrman}\ \emph {et~al.}(2010)\citenamefont {Buhrman}, \citenamefont {Cleve}, \citenamefont {Massar},\ and\ \citenamefont {de~Wolf}}]{BuhrmanCleveMassarDeWolf2010}%
  \BibitemOpen
  \bibfield  {author} {\bibinfo {author} {\bibfnamefont {Harry}\ \bibnamefont {Buhrman}}, \bibinfo {author} {\bibfnamefont {Richard}\ \bibnamefont {Cleve}}, \bibinfo {author} {\bibfnamefont {Serge}\ \bibnamefont {Massar}}, \ and\ \bibinfo {author} {\bibfnamefont {Ronald}\ \bibnamefont {de~Wolf}},\ }\emph {\enquote {\bibinfo {title} {Nonlocality and communication complexity},}\ }\href {https://doi.org/10.1103/RevModPhys.82.665} {\bibfield  {journal} {\bibinfo  {journal} {Rev. Mod. Phys.}\ }\textbf {\bibinfo {volume} {82}},\ \bibinfo {pages} {665{\textendash}698} (\bibinfo {year} {2010})},\ \Eprint {http://arxiv.org/abs/0907.3584} {arXiv:0907.3584}\BibitemShut {NoStop}%
\bibitem [{\citenamefont {Scarani}\ \emph {et~al.}(2009)\citenamefont {Scarani}, \citenamefont {Bechmann-Pasquinucci}, \citenamefont {Cerf}, \citenamefont {Du\ifmmode~\check{s}\else \v{s}\fi{}ek}, \citenamefont {L\"utkenhaus},\ and\ \citenamefont {Peev}}]{ScaraniEtAl2009}%
  \BibitemOpen
  \bibfield  {author} {\bibinfo {author} {\bibfnamefont {Valerio}\ \bibnamefont {Scarani}}, \bibinfo {author} {\bibfnamefont {Helle}\ \bibnamefont {Bechmann-Pasquinucci}}, \bibinfo {author} {\bibfnamefont {Nicolas~J.}\ \bibnamefont {Cerf}}, \bibinfo {author} {\bibfnamefont {Miloslav}\ \bibnamefont {Du\ifmmode~\check{s}\else \v{s}\fi{}ek}}, \bibinfo {author} {\bibfnamefont {Norbert}\ \bibnamefont {L\"utkenhaus}}, \ and\ \bibinfo {author} {\bibfnamefont {Momtchil}\ \bibnamefont {Peev}},\ }\emph {\enquote {\bibinfo {title} {The security of practical quantum key distribution},}\ }\href {https://doi.org/10.1103/RevModPhys.81.1301} {\bibfield  {journal} {\bibinfo  {journal} {Rev. Mod. Phys.}\ }\textbf {\bibinfo {volume} {81}},\ \bibinfo {pages} {1301{\textendash}1350} (\bibinfo {year} {2009})},\ \Eprint {http://arxiv.org/abs/0802.4155} {arXiv:0802.4155}\BibitemShut {NoStop}%
\bibitem [{\citenamefont {Xu}\ \emph {et~al.}(2020)\citenamefont {Xu}, \citenamefont {Ma}, \citenamefont {Zhang}, \citenamefont {Lo},\ and\ \citenamefont {Pan}}]{XuMaZhangLoPan2020}%
  \BibitemOpen
  \bibfield  {author} {\bibinfo {author} {\bibfnamefont {Feihu}\ \bibnamefont {Xu}}, \bibinfo {author} {\bibfnamefont {Xiongfeng}\ \bibnamefont {Ma}}, \bibinfo {author} {\bibfnamefont {Qiang}\ \bibnamefont {Zhang}}, \bibinfo {author} {\bibfnamefont {Hoi-Kwong}\ \bibnamefont {Lo}}, \ and\ \bibinfo {author} {\bibfnamefont {Jian-Wei}\ \bibnamefont {Pan}},\ }\emph {\enquote {\bibinfo {title} {Secure quantum key distribution with realistic devices},}\ }\href {https://doi.org/10.1103/RevModPhys.92.025002} {\bibfield  {journal} {\bibinfo  {journal} {Rev. Mod. Phys.}\ }\textbf {\bibinfo {volume} {92}},\ \bibinfo {pages} {025002} (\bibinfo {year} {2020})},\ \Eprint {http://arxiv.org/abs/1903.09051} {arXiv:1903.09051}\BibitemShut {NoStop}%
\bibitem [{\citenamefont {Degen}\ \emph {et~al.}(2017)\citenamefont {Degen}, \citenamefont {Reinhard},\ and\ \citenamefont {Cappellaro}}]{DegenReinhardCappellaro2017}%
  \BibitemOpen
  \bibfield  {author} {\bibinfo {author} {\bibfnamefont {Christian~L.}\ \bibnamefont {Degen}}, \bibinfo {author} {\bibfnamefont {Friedemann}\ \bibnamefont {Reinhard}}, \ and\ \bibinfo {author} {\bibfnamefont {Paola}\ \bibnamefont {Cappellaro}},\ }\emph {\enquote {\bibinfo {title} {Quantum sensing},}\ }\href {https://doi.org/10.1103/RevModPhys.89.035002} {\bibfield  {journal} {\bibinfo  {journal} {Rev. Mod. Phys.}\ }\textbf {\bibinfo {volume} {89}},\ \bibinfo {pages} {035002} (\bibinfo {year} {2017})},\ \Eprint {http://arxiv.org/abs/1611.02427} {arXiv:1611.02427}\BibitemShut {NoStop}%
\bibitem [{\citenamefont {Giovannetti}\ \emph {et~al.}(2011)\citenamefont {Giovannetti}, \citenamefont {Lloyd},\ and\ \citenamefont {Maccone}}]{GiovannettiLloydMaccone2011}%
  \BibitemOpen
  \bibfield  {author} {\bibinfo {author} {\bibfnamefont {Vittorio}\ \bibnamefont {Giovannetti}}, \bibinfo {author} {\bibfnamefont {Seth}\ \bibnamefont {Lloyd}}, \ and\ \bibinfo {author} {\bibfnamefont {Lorenzo}\ \bibnamefont {Maccone}},\ }\emph {\enquote {\bibinfo {title} {{Advances in quantum metrology}},}\ }\href {https://doi.org/10.1038/nphoton.2011.35} {\bibfield  {journal} {\bibinfo  {journal} {Nat. Photonics}\ }\textbf {\bibinfo {volume} {5}},\ \bibinfo {pages} {222} (\bibinfo {year} {2011})},\ \Eprint {http://arxiv.org/abs/1102.2318} {arXiv:1102.2318}\BibitemShut {NoStop}%
\bibitem [{\citenamefont {Jozsa}\ and\ \citenamefont {Linden}(2003)}]{JozsaLinden2003}%
  \BibitemOpen
  \bibfield  {author} {\bibinfo {author} {\bibfnamefont {Richard}\ \bibnamefont {Jozsa}}\ and\ \bibinfo {author} {\bibfnamefont {Noah}\ \bibnamefont {Linden}},\ }\emph {\enquote {\bibinfo {title} {{On the Role of Entanglement in Quantum-Computational Speed-Up}},}\ }\href {https://doi.org/10.1098/rspa.2002.1097} {\bibfield  {journal} {\bibinfo  {journal} {Proc. Roy. Soc. A Math. Phys.}\ }\textbf {\bibinfo {volume} {459}},\ \bibinfo {pages} {2011{\textendash}2032} (\bibinfo {year} {2003})},\ \Eprint {http://arxiv.org/abs/quant-ph/0201143} {arXiv:quant-ph/0201143}\BibitemShut {NoStop}%
\bibitem [{\citenamefont {Friis}\ \emph {et~al.}(2019)\citenamefont {Friis}, \citenamefont {Vitagliano}, \citenamefont {Malik},\ and\ \citenamefont {Huber}}]{FriisVitaglianoMalikHuber2019}%
  \BibitemOpen
  \bibfield  {author} {\bibinfo {author} {\bibfnamefont {Nicolai}\ \bibnamefont {Friis}}, \bibinfo {author} {\bibfnamefont {Giuseppe}\ \bibnamefont {Vitagliano}}, \bibinfo {author} {\bibfnamefont {Mehul}\ \bibnamefont {Malik}}, \ and\ \bibinfo {author} {\bibfnamefont {Marcus}\ \bibnamefont {Huber}},\ }\emph {\enquote {\bibinfo {title} {{Entanglement Certification From Theory to Experiment}},}\ }\href {https://doi.org/10.1038/s42254-018-0003-5} {\bibfield  {journal} {\bibinfo  {journal} {Nat. Rev. Phys.}\ }\textbf {\bibinfo {volume} {1}},\ \bibinfo {pages} {72{\textendash}87} (\bibinfo {year} {2019})},\ \Eprint {http://arxiv.org/abs/1906.10929} {arXiv:1906.10929}\BibitemShut {NoStop}%
\bibitem [{\citenamefont {Bavaresco}\ \emph {et~al.}(2018)\citenamefont {Bavaresco}, \citenamefont {Herrera~Valencia}, \citenamefont {Kl{\"o}ckl}, \citenamefont {Pivoluska}, \citenamefont {Erker}, \citenamefont {Friis}, \citenamefont {Malik},\ and\ \citenamefont {Huber}}]{BavarescoEtAl2018}%
  \BibitemOpen
  \bibfield  {author} {\bibinfo {author} {\bibfnamefont {Jessica}\ \bibnamefont {Bavaresco}}, \bibinfo {author} {\bibfnamefont {Natalia}\ \bibnamefont {Herrera~Valencia}}, \bibinfo {author} {\bibfnamefont {Claude}\ \bibnamefont {Kl{\"o}ckl}}, \bibinfo {author} {\bibfnamefont {Matej}\ \bibnamefont {Pivoluska}}, \bibinfo {author} {\bibfnamefont {Paul}\ \bibnamefont {Erker}}, \bibinfo {author} {\bibfnamefont {Nicolai}\ \bibnamefont {Friis}}, \bibinfo {author} {\bibfnamefont {Mehul}\ \bibnamefont {Malik}}, \ and\ \bibinfo {author} {\bibfnamefont {Marcus}\ \bibnamefont {Huber}},\ }\emph {\enquote {\bibinfo {title} {Measurements in two bases are sufficient for certifying high-dimensional entanglement},}\ }\href {https://doi.org/10.1038/s41567-018-0203-z} {\bibfield  {journal} {\bibinfo  {journal} {Nat. Phys.}\ }\textbf {\bibinfo {volume} {14}},\ \bibinfo {pages} {1032{\textendash}1037} (\bibinfo {year} {2018})},\ \Eprint {http://arxiv.org/abs/1709.07344} {arXiv:1709.07344}\BibitemShut {NoStop}%
\bibitem [{\citenamefont {Hrmo}\ \emph {et~al.}(2023)\citenamefont {Hrmo}, \citenamefont {Wilhelm}, \citenamefont {Gerster}, \citenamefont {van Mourik}, \citenamefont {Huber}, \citenamefont {Blatt}, \citenamefont {Schindler}, \citenamefont {Monz},\ and\ \citenamefont {Ringbauer}}]{HrmoEtAl2023}%
  \BibitemOpen
  \bibfield  {author} {\bibinfo {author} {\bibfnamefont {Pavel}\ \bibnamefont {Hrmo}}, \bibinfo {author} {\bibfnamefont {Benjamin}\ \bibnamefont {Wilhelm}}, \bibinfo {author} {\bibfnamefont {Lukas}\ \bibnamefont {Gerster}}, \bibinfo {author} {\bibfnamefont {Martin~W.}\ \bibnamefont {van Mourik}}, \bibinfo {author} {\bibfnamefont {Marcus}\ \bibnamefont {Huber}}, \bibinfo {author} {\bibfnamefont {Rainer}\ \bibnamefont {Blatt}}, \bibinfo {author} {\bibfnamefont {Philipp}\ \bibnamefont {Schindler}}, \bibinfo {author} {\bibfnamefont {Thomas}\ \bibnamefont {Monz}}, \ and\ \bibinfo {author} {\bibfnamefont {Martin}\ \bibnamefont {Ringbauer}},\ }\emph {\enquote {\bibinfo {title} {{Native qudit entanglement in a trapped ion quantum processor}},}\ }\href {https://doi.org/10.1038/s41467-023-37375-2} {\bibfield  {journal} {\bibinfo  {journal} {Nat. Commun.}\ }\textbf {\bibinfo {volume} {14}},\ \bibinfo {pages} {2242} (\bibinfo {year} {2023})},\ \Eprint {http://arxiv.org/abs/2206.04104} {arXiv:2206.04104}\BibitemShut
  {NoStop}%
\bibitem [{\citenamefont {Doda}\ \emph {et~al.}(2021)\citenamefont {Doda}, \citenamefont {Huber}, \citenamefont {Murta}, \citenamefont {Pivoluska}, \citenamefont {Plesch},\ and\ \citenamefont {Vlachou}}]{DodaHuberMurtaPivoluskaPleschVlachou2021}%
  \BibitemOpen
  \bibfield  {author} {\bibinfo {author} {\bibfnamefont {Mirdit}\ \bibnamefont {Doda}}, \bibinfo {author} {\bibfnamefont {Marcus}\ \bibnamefont {Huber}}, \bibinfo {author} {\bibfnamefont {Gl\'aucia}\ \bibnamefont {Murta}}, \bibinfo {author} {\bibfnamefont {Matej}\ \bibnamefont {Pivoluska}}, \bibinfo {author} {\bibfnamefont {Martin}\ \bibnamefont {Plesch}}, \ and\ \bibinfo {author} {\bibfnamefont {Chrysoula}\ \bibnamefont {Vlachou}},\ }\emph {\enquote {\bibinfo {title} {Quantum key distribution overcoming extreme noise: Simultaneous subspace coding using high-dimensional entanglement},}\ }\href {https://doi.org/10.1103/PhysRevApplied.15.034003} {\bibfield  {journal} {\bibinfo  {journal} {Phys. Rev. Appl.}\ }\textbf {\bibinfo {volume} {15}},\ \bibinfo {pages} {034003} (\bibinfo {year} {2021})},\ \Eprint {http://arxiv.org/abs/2004.12824} {arXiv:2004.12824}\BibitemShut {NoStop}%
\bibitem [{\citenamefont {Ndagano}\ \emph {et~al.}(2020)\citenamefont {Ndagano}, \citenamefont {Defienne}, \citenamefont {Lyons}, \citenamefont {Starshynov}, \citenamefont {Villa}, \citenamefont {Tisa},\ and\ \citenamefont {Faccio}}]{NdaganoDefienneLyonsStarshynovVillaTisaFaccio2020}%
  \BibitemOpen
  \bibfield  {author} {\bibinfo {author} {\bibfnamefont {Bienvenu}\ \bibnamefont {Ndagano}}, \bibinfo {author} {\bibfnamefont {Hugo}\ \bibnamefont {Defienne}}, \bibinfo {author} {\bibfnamefont {Ashley}\ \bibnamefont {Lyons}}, \bibinfo {author} {\bibfnamefont {Ilya}\ \bibnamefont {Starshynov}}, \bibinfo {author} {\bibfnamefont {Federica}\ \bibnamefont {Villa}}, \bibinfo {author} {\bibfnamefont {Simone}\ \bibnamefont {Tisa}}, \ and\ \bibinfo {author} {\bibfnamefont {Daniele}\ \bibnamefont {Faccio}},\ }\emph {\enquote {\bibinfo {title} {{Imaging and certifying high-dimensional entanglement with a single-photon avalanche diode camera}},}\ }\href {https://doi.org/10.1038/s41534-020-00324-8} {\bibfield  {journal} {\bibinfo  {journal} {npj Quantum Inf.}\ }\textbf {\bibinfo {volume} {6}},\ \bibinfo {pages} {94} (\bibinfo {year} {2020})},\ \Eprint {http://arxiv.org/abs/2001.03997} {arXiv:2001.03997}\BibitemShut {NoStop}%
\bibitem [{\citenamefont {Defienne}\ \emph {et~al.}(2021)\citenamefont {Defienne}, \citenamefont {Ndagano}, \citenamefont {Lyons},\ and\ \citenamefont {Faccio}}]{DefienneNdaganoLyonsFaccio2021}%
  \BibitemOpen
  \bibfield  {author} {\bibinfo {author} {\bibfnamefont {Hugo}\ \bibnamefont {Defienne}}, \bibinfo {author} {\bibfnamefont {Bienvenu}\ \bibnamefont {Ndagano}}, \bibinfo {author} {\bibfnamefont {Ashley}\ \bibnamefont {Lyons}}, \ and\ \bibinfo {author} {\bibfnamefont {Daniele}\ \bibnamefont {Faccio}},\ }\emph {\enquote {\bibinfo {title} {{Polarization entanglement-enabled quantum holography}},}\ }\href {https://doi.org/10.1038/s41567-020-01156-1} {\bibfield  {journal} {\bibinfo  {journal} {Nat. Phys.}\ }\textbf {\bibinfo {volume} {17}},\ \bibinfo {pages} {591{\textendash}597} (\bibinfo {year} {2021})},\ \Eprint {http://arxiv.org/abs/1911.01209} {arXiv:1911.01209}\BibitemShut {NoStop}%
\bibitem [{\citenamefont {Cameron}\ \emph {et~al.}(2024)\citenamefont {Cameron}, \citenamefont {Courme}, \citenamefont {Verni{\`e}re}, \citenamefont {Pandya}, \citenamefont {Faccio},\ and\ \citenamefont {Defienne}}]{CameronCourmeVernierePandyaFaccioDefienne2024}%
  \BibitemOpen
  \bibfield  {author} {\bibinfo {author} {\bibfnamefont {Patrick}\ \bibnamefont {Cameron}}, \bibinfo {author} {\bibfnamefont {Baptiste}\ \bibnamefont {Courme}}, \bibinfo {author} {\bibfnamefont {Chlo{\'e}}\ \bibnamefont {Verni{\`e}re}}, \bibinfo {author} {\bibfnamefont {Raj}\ \bibnamefont {Pandya}}, \bibinfo {author} {\bibfnamefont {Daniele}\ \bibnamefont {Faccio}}, \ and\ \bibinfo {author} {\bibfnamefont {Hugo}\ \bibnamefont {Defienne}},\ }\emph {\enquote {\bibinfo {title} {Adaptive optical imaging with entangled photons},}\ }\href {https://doi.org/10.1126/science.adk7825} {\bibfield  {journal} {\bibinfo  {journal} {Science}\ }\textbf {\bibinfo {volume} {383}},\ \bibinfo {pages} {1142{\textendash}1148} (\bibinfo {year} {2024})},\ \Eprint {http://arxiv.org/abs/2308.11472} {arXiv:2308.11472}\BibitemShut {NoStop}%
\bibitem [{\citenamefont {Lanyon}\ \emph {et~al.}(2009)\citenamefont {Lanyon}, \citenamefont {Barbieri}, \citenamefont {Almeida}, \citenamefont {Jennewein}, \citenamefont {Ralph}, \citenamefont {Resch}, \citenamefont {Pryde}, \citenamefont {O'Brien}, \citenamefont {Gilchrist},\ and\ \citenamefont {White}}]{LanyonEtAl2009}%
  \BibitemOpen
  \bibfield  {author} {\bibinfo {author} {\bibfnamefont {Benjamin~P.}\ \bibnamefont {Lanyon}}, \bibinfo {author} {\bibfnamefont {Marco}\ \bibnamefont {Barbieri}}, \bibinfo {author} {\bibfnamefont {Marcelo~P.}\ \bibnamefont {Almeida}}, \bibinfo {author} {\bibfnamefont {Thomas}\ \bibnamefont {Jennewein}}, \bibinfo {author} {\bibfnamefont {Timothy~C.}\ \bibnamefont {Ralph}}, \bibinfo {author} {\bibfnamefont {Kevin~J.}\ \bibnamefont {Resch}}, \bibinfo {author} {\bibfnamefont {Geoff~J.}\ \bibnamefont {Pryde}}, \bibinfo {author} {\bibfnamefont {Jeremy~L.}\ \bibnamefont {O'Brien}}, \bibinfo {author} {\bibfnamefont {Alexei}\ \bibnamefont {Gilchrist}}, \ and\ \bibinfo {author} {\bibfnamefont {Andrew~G.}\ \bibnamefont {White}},\ }\emph {\enquote {\bibinfo {title} {{Simplifying quantum logic using higher-dimensional Hilbert spaces}},}\ }\href {https://doi.org/10.1038/nphys1150} {\bibfield  {journal} {\bibinfo  {journal} {Nat. Phys.}\ }\textbf {\bibinfo {volume} {5}},\ \bibinfo {pages} {134{\textendash}140} (\bibinfo
  {year} {2009})},\ \Eprint {http://arxiv.org/abs/0804.0272} {arXiv:0804.0272}\BibitemShut {NoStop}%
\bibitem [{\citenamefont {Amico}\ \emph {et~al.}(2008)\citenamefont {Amico}, \citenamefont {Fazio}, \citenamefont {Osterloh},\ and\ \citenamefont {Vedral}}]{AmicoFazioOsterlohVedral2008}%
  \BibitemOpen
  \bibfield  {author} {\bibinfo {author} {\bibfnamefont {Luigi}\ \bibnamefont {Amico}}, \bibinfo {author} {\bibfnamefont {Rosario}\ \bibnamefont {Fazio}}, \bibinfo {author} {\bibfnamefont {Andreas}\ \bibnamefont {Osterloh}}, \ and\ \bibinfo {author} {\bibfnamefont {Vlatko}\ \bibnamefont {Vedral}},\ }\emph {\enquote {\bibinfo {title} {Entanglement in many-body systems},}\ }\href {https://doi.org/10.1103/RevModPhys.80.517} {\bibfield  {journal} {\bibinfo  {journal} {Rev. Mod. Phys.}\ }\textbf {\bibinfo {volume} {80}},\ \bibinfo {pages} {517{\textendash}576} (\bibinfo {year} {2008})},\ \Eprint {http://arxiv.org/abs/0703044} {arXiv:0703044}\BibitemShut {NoStop}%
\bibitem [{\citenamefont {Verstraete}\ \emph {et~al.}(2008)\citenamefont {Verstraete}, \citenamefont {Murg},\ and\ \citenamefont {Cirac}}]{VerstraeteMurgCirac2008}%
  \BibitemOpen
  \bibfield  {author} {\bibinfo {author} {\bibfnamefont {Frank}\ \bibnamefont {Verstraete}}, \bibinfo {author} {\bibfnamefont {V.}~\bibnamefont {Murg}}, \ and\ \bibinfo {author} {\bibfnamefont {Juan~Ignacio}\ \bibnamefont {Cirac}},\ }\emph {\enquote {\bibinfo {title} {Matrix product states, projected entangled pair states, and variational renormalization group methods for quantum spin systems},}\ }\href {https://doi.org/10.1080/14789940801912366} {\bibfield  {journal} {\bibinfo  {journal} {Adv. Phys.}\ }\textbf {\bibinfo {volume} {57}},\ \bibinfo {pages} {143{\textendash}224} (\bibinfo {year} {2008})},\ \Eprint {http://arxiv.org/abs/0907.2796} {arXiv:0907.2796}\BibitemShut {NoStop}%
\bibitem [{\citenamefont {Eisert}\ \emph {et~al.}(2010)\citenamefont {Eisert}, \citenamefont {Cramer},\ and\ \citenamefont {Plenio}}]{EisertCramerPlenio2010}%
  \BibitemOpen
  \bibfield  {author} {\bibinfo {author} {\bibfnamefont {Jens}\ \bibnamefont {Eisert}}, \bibinfo {author} {\bibfnamefont {M.}~\bibnamefont {Cramer}}, \ and\ \bibinfo {author} {\bibfnamefont {Martin~B.}\ \bibnamefont {Plenio}},\ }\emph {\enquote {\bibinfo {title} {{Colloquium: Area laws for the entanglement entropy}},}\ }\href {https://doi.org/10.1103/RevModPhys.82.277} {\bibfield  {journal} {\bibinfo  {journal} {Rev. Mod. Phys.}\ }\textbf {\bibinfo {volume} {82}},\ \bibinfo {pages} {277{\textendash}306} (\bibinfo {year} {2010})},\ \Eprint {http://arxiv.org/abs/0808.3773} {arXiv:0808.3773}\BibitemShut {NoStop}%
\bibitem [{\citenamefont {Spengler}\ \emph {et~al.}(2012)\citenamefont {Spengler}, \citenamefont {Huber}, \citenamefont {Brierley}, \citenamefont {Adaktylos},\ and\ \citenamefont {Hiesmayr}}]{SpenglerHuberBrierleyAdaktylosHiesmayr2012}%
  \BibitemOpen
  \bibfield  {author} {\bibinfo {author} {\bibfnamefont {Christoph}\ \bibnamefont {Spengler}}, \bibinfo {author} {\bibfnamefont {Marcus}\ \bibnamefont {Huber}}, \bibinfo {author} {\bibfnamefont {Stephen}\ \bibnamefont {Brierley}}, \bibinfo {author} {\bibfnamefont {Theodor}\ \bibnamefont {Adaktylos}}, \ and\ \bibinfo {author} {\bibfnamefont {Beatrix~C.}\ \bibnamefont {Hiesmayr}},\ }\emph {\enquote {\bibinfo {title} {{Entanglement detection via mutually unbiased bases}},}\ }\href {https://doi.org/10.1103/PhysRevA.86.022311} {\bibfield  {journal} {\bibinfo  {journal} {Phys. Rev. A}\ }\textbf {\bibinfo {volume} {86}},\ \bibinfo {pages} {022311} (\bibinfo {year} {2012})},\ \Eprint {http://arxiv.org/abs/1202.5058} {arXiv:1202.5058}\BibitemShut {NoStop}%
\bibitem [{\citenamefont {Erker}\ \emph {et~al.}(2017)\citenamefont {Erker}, \citenamefont {Krenn},\ and\ \citenamefont {Huber}}]{ErkerKrennHuber2017}%
  \BibitemOpen
  \bibfield  {author} {\bibinfo {author} {\bibfnamefont {Paul}\ \bibnamefont {Erker}}, \bibinfo {author} {\bibfnamefont {Mario}\ \bibnamefont {Krenn}}, \ and\ \bibinfo {author} {\bibfnamefont {Marcus}\ \bibnamefont {Huber}},\ }\emph {\enquote {\bibinfo {title} {Quantifying high dimensional entanglement with two mutually unbiased bases},}\ }\href {https://doi.org/10.22331/q-2017-07-28-22} {\bibfield  {journal} {\bibinfo  {journal} {{Quantum}}\ }\textbf {\bibinfo {volume} {1}},\ \bibinfo {pages} {22} (\bibinfo {year} {2017})},\ \Eprint {http://arxiv.org/abs/1512.05315} {arXiv:1512.05315}\BibitemShut {NoStop}%
\bibitem [{\citenamefont {Herrera~Valencia}\ \emph {et~al.}(2020)\citenamefont {Herrera~Valencia}, \citenamefont {Srivastav}, \citenamefont {Pivoluska}, \citenamefont {Huber}, \citenamefont {Friis}, \citenamefont {McCutcheon},\ and\ \citenamefont {Malik}}]{Herrera-ValenciaSrivastavPivoluskaHuberFriisMcCutcheonMalik2020}%
  \BibitemOpen
  \bibfield  {author} {\bibinfo {author} {\bibfnamefont {Natalia}\ \bibnamefont {Herrera~Valencia}}, \bibinfo {author} {\bibfnamefont {Vatshal}\ \bibnamefont {Srivastav}}, \bibinfo {author} {\bibfnamefont {Matej}\ \bibnamefont {Pivoluska}}, \bibinfo {author} {\bibfnamefont {Marcus}\ \bibnamefont {Huber}}, \bibinfo {author} {\bibfnamefont {Nicolai}\ \bibnamefont {Friis}}, \bibinfo {author} {\bibfnamefont {Will}\ \bibnamefont {McCutcheon}}, \ and\ \bibinfo {author} {\bibfnamefont {Mehul}\ \bibnamefont {Malik}},\ }\emph {\enquote {\bibinfo {title} {{High-Dimensional Pixel Entanglement: Efficient Generation and Certification}},}\ }\href {https://doi.org/10.22331/q-2020-12-24-376} {\bibfield  {journal} {\bibinfo  {journal} {Quantum}\ }\textbf {\bibinfo {volume} {4}},\ \bibinfo {pages} {376} (\bibinfo {year} {2020})},\ \Eprint {http://arxiv.org/abs/2004.04994} {arXiv:2004.04994}\BibitemShut {NoStop}%
\bibitem [{\citenamefont {Morelli}\ \emph {et~al.}(2023)\citenamefont {Morelli}, \citenamefont {Huber},\ and\ \citenamefont {Tavakoli}}]{MorelliHuberTavakoli2023}%
  \BibitemOpen
  \bibfield  {author} {\bibinfo {author} {\bibfnamefont {Simon}\ \bibnamefont {Morelli}}, \bibinfo {author} {\bibfnamefont {Marcus}\ \bibnamefont {Huber}}, \ and\ \bibinfo {author} {\bibfnamefont {Armin}\ \bibnamefont {Tavakoli}},\ }\emph {\enquote {\bibinfo {title} {{Resource-Efficient High-Dimensional Entanglement Detection via Symmetric Projections}},}\ }\href {https://doi.org/10.1103/PhysRevLett.131.170201} {\bibfield  {journal} {\bibinfo  {journal} {Phys. Rev. Lett.}\ }\textbf {\bibinfo {volume} {131}},\ \bibinfo {pages} {170201} (\bibinfo {year} {2023})},\ \Eprint {http://arxiv.org/abs/2304.04274} {arXiv:2304.04274}\BibitemShut {NoStop}%
\bibitem [{\citenamefont {Terhal}\ and\ \citenamefont {Horodecki}(2000)}]{TerhalHorodecki2000}%
  \BibitemOpen
  \bibfield  {author} {\bibinfo {author} {\bibfnamefont {Barbara~M.}\ \bibnamefont {Terhal}}\ and\ \bibinfo {author} {\bibfnamefont {Pave{\l}}\ \bibnamefont {Horodecki}},\ }\emph {\enquote {\bibinfo {title} {Schmidt number for density matrices},}\ }\href {https://doi.org/10.1103/PhysRevA.61.040301} {\bibfield  {journal} {\bibinfo  {journal} {Phys. Rev. A}\ }\textbf {\bibinfo {volume} {61}},\ \bibinfo {pages} {040301(R)} (\bibinfo {year} {2000})},\ \Eprint {http://arxiv.org/abs/quant-ph/9911117} {arXiv:quant-ph/9911117}\BibitemShut {NoStop}%
\bibitem [{\citenamefont {Berta}\ \emph {et~al.}(2010)\citenamefont {Berta}, \citenamefont {Christandl}, \citenamefont {Colbeck}, \citenamefont {Renes},\ and\ \citenamefont {Renner}}]{BertaChristandlColbeckRenesRenner2010}%
  \BibitemOpen
  \bibfield  {author} {\bibinfo {author} {\bibfnamefont {Mario}\ \bibnamefont {Berta}}, \bibinfo {author} {\bibfnamefont {Matthias}\ \bibnamefont {Christandl}}, \bibinfo {author} {\bibfnamefont {Roger}\ \bibnamefont {Colbeck}}, \bibinfo {author} {\bibfnamefont {Joseph~M.}\ \bibnamefont {Renes}}, \ and\ \bibinfo {author} {\bibfnamefont {Renato}\ \bibnamefont {Renner}},\ }\emph {\enquote {\bibinfo {title} {The uncertainty principle in the presence of quantum memory},}\ }\href {https://doi.org/10.1038/nphys1734} {\bibfield  {journal} {\bibinfo  {journal} {Nat. Phys.}\ }\textbf {\bibinfo {volume} {6}},\ \bibinfo {pages} {659{\textendash}662} (\bibinfo {year} {2010})},\ \Eprint {http://arxiv.org/abs/0909.0950} {arXiv:0909.0950}\BibitemShut {NoStop}%
\bibitem [{\citenamefont {Coles}\ \emph {et~al.}(2017)\citenamefont {Coles}, \citenamefont {Berta}, \citenamefont {Tomamichel},\ and\ \citenamefont {Wehner}}]{ColesBertaTomamichelWehner2017}%
  \BibitemOpen
  \bibfield  {author} {\bibinfo {author} {\bibfnamefont {Patrick~J.}\ \bibnamefont {Coles}}, \bibinfo {author} {\bibfnamefont {Mario}\ \bibnamefont {Berta}}, \bibinfo {author} {\bibfnamefont {Marco}\ \bibnamefont {Tomamichel}}, \ and\ \bibinfo {author} {\bibfnamefont {Stephanie}\ \bibnamefont {Wehner}},\ }\emph {\enquote {\bibinfo {title} {{Entropic uncertainty relations and their applications}},}\ }\href {https://doi.org/10.1103/RevModPhys.89.015002} {\bibfield  {journal} {\bibinfo  {journal} {Rev. Mod. Phys.}\ }\textbf {\bibinfo {volume} {89}},\ \bibinfo {pages} {015002} (\bibinfo {year} {2017})},\ \Eprint {http://arxiv.org/abs/1511.04857} {arXiv:1511.04857}\BibitemShut {NoStop}%
\bibitem [{\citenamefont {Devetak}\ and\ \citenamefont {Winter}(2005)}]{DevetakWinter2005}%
  \BibitemOpen
  \bibfield  {author} {\bibinfo {author} {\bibfnamefont {Igor}\ \bibnamefont {Devetak}}\ and\ \bibinfo {author} {\bibfnamefont {Andreas}\ \bibnamefont {Winter}},\ }\emph {\enquote {\bibinfo {title} {Distillation of secret key and entanglement from quantum states},}\ }\href {https://doi.org/10.1098/rspa.2004.1372} {\bibfield  {journal} {\bibinfo  {journal} {Proc. R. Soc. Lond. A}\ }\textbf {\bibinfo {volume} {461}},\ \bibinfo {pages} {207{\textendash}235} (\bibinfo {year} {2005})},\ \Eprint {http://arxiv.org/abs/quant-ph/0306078} {arXiv:quant-ph/0306078}\BibitemShut {NoStop}%
\bibitem [{\citenamefont {Bennett}\ \emph {et~al.}(1996)\citenamefont {Bennett}, \citenamefont {Di~Vincenzo}, \citenamefont {Smolin},\ and\ \citenamefont {Wootters}}]{BennettDiVincenzoSmolinWootters1996}%
  \BibitemOpen
  \bibfield  {author} {\bibinfo {author} {\bibfnamefont {Charles~H.}\ \bibnamefont {Bennett}}, \bibinfo {author} {\bibfnamefont {David~P.}\ \bibnamefont {Di~Vincenzo}}, \bibinfo {author} {\bibfnamefont {John~A.}\ \bibnamefont {Smolin}}, \ and\ \bibinfo {author} {\bibfnamefont {William~K.}\ \bibnamefont {Wootters}},\ }\emph {\enquote {\bibinfo {title} {Mixed-state entanglement and quantum error correction},}\ }\href {\doibase 10.1103/PhysRevA.54.3824} {\bibfield  {journal} {\bibinfo  {journal} {Phys. Rev. A}\ }\textbf {\bibinfo {volume} {54}},\ \bibinfo {pages} {3824} (\bibinfo {year} {1996})},\ \Eprint {http://arxiv.org/abs/quant-ph/9604024} {arXiv:quant-ph/9604024}\BibitemShut {NoStop}%
\bibitem [{\citenamefont {Horodecki}\ \emph {et~al.}(1999)\citenamefont {Horodecki}, \citenamefont {Horodecki},\ and\ \citenamefont {Horodecki}}]{HorodeckiMPR1999}%
  \BibitemOpen
  \bibfield  {author} {\bibinfo {author} {\bibfnamefont {Micha{\l}}\ \bibnamefont {Horodecki}}, \bibinfo {author} {\bibfnamefont {Pawe{\l}}\ \bibnamefont {Horodecki}}, \ and\ \bibinfo {author} {\bibfnamefont {Ryszard}\ \bibnamefont {Horodecki}},\ }\emph {\enquote {\bibinfo {title} {General teleportation channel, singlet fraction and quasidistillation},}\ }\href {\doibase 10.1103/PhysRevA.60.1888} {\bibfield  {journal} {\bibinfo  {journal} {Phys. Rev. A}\ }\textbf {\bibinfo {volume} {60}},\ \bibinfo {pages} {1888} (\bibinfo {year} {1999})},\ \Eprint {http://arxiv.org/abs/quant-ph/9807091} {arXiv:quant-ph/9807091}\BibitemShut {NoStop}%
\bibitem [{\citenamefont {Wootters}\ and\ \citenamefont {Fields}(1989)}]{WoottersFields1989}%
  \BibitemOpen
  \bibfield  {author} {\bibinfo {author} {\bibfnamefont {William~K.}\ \bibnamefont {Wootters}}\ and\ \bibinfo {author} {\bibfnamefont {Brian~D.}\ \bibnamefont {Fields}},\ }\emph {\enquote {\bibinfo {title} {Optimal state-determination by mutually unbiased measurements},}\ }\href {https://doi.org/10.1016/0003-4916(89)90322-9} {\bibfield  {journal} {\bibinfo  {journal} {Ann. Phys.}\ }\textbf {\bibinfo {volume} {191}},\ \bibinfo {pages} {363{\textendash}381} (\bibinfo {year} {1989})}\BibitemShut {NoStop}%
\bibitem [{foo({\natexlab{a}})}]{footnote:Eg1WorstCase}%
  \BibitemOpen
  \href@noop {} {}\bibinfo {note} {As it is harder to witness a state to have Schmidt number $k+1$ with a larger bound $\mathcal{B}_k$, this is the worst-case bases choice for a given $c_\text{min} = \min_{z,z'} c^{z,z'}_\text{min}$ because (i) it gives the maximal value allowed for $c_\text{max} = \max_{z,z'} c^{z,z'}_\text{max}$ such that $\sum_{a'}|\langle e^z_a|e^{z'}_{a'}\rangle|^2=1$ holds for all $a$, and (ii) $\mathcal{T}(\mathcal{C})$ (and therefore $\mathcal{B}_k$) increases with $c_\text{max}$ for a fixed $c_\text{min}$ due to the enlargement of the feasible set of the optimization problem in Proposition \ref{prop:NonCVXopt}.}\BibitemShut {Stop}%
\bibitem [{foo({\natexlab{b}})}]{footnote:EoFiso}%
  \BibitemOpen
  \href@noop {} {}\bibinfo {note} {The analytic expression for the negativity of isotropic states can be found in Eq.\;(4.2) of Ref.\;\cite{RungtaCaves2003}.}\BibitemShut {Stop}%
\bibitem [{foo({\natexlab{c}})}]{footnote:NegativityIso}%
  \BibitemOpen
  \href@noop {} {}\bibinfo {note} {The analytic expression for the entanglement of formation of isotropic states is stated in Eq.\;(26) of Ref.\;\cite{LeeChiOhKim2003}.}\BibitemShut {Stop}%
\bibitem [{\citenamefont {Colomer}\ \emph {et~al.}(2022)\citenamefont {Colomer}, \citenamefont {Mortimer}, \citenamefont {Fr{\'{e}}rot}, \citenamefont {Farkas},\ and\ \citenamefont {Ac{\'{i}}n}}]{ColomerMortimerFrerotFarkasAcin2022}%
  \BibitemOpen
  \bibfield  {author} {\bibinfo {author} {\bibfnamefont {Maria~Prat}\ \bibnamefont {Colomer}}, \bibinfo {author} {\bibfnamefont {Luke}\ \bibnamefont {Mortimer}}, \bibinfo {author} {\bibfnamefont {Ir{\'{e}}n{\'{e}}e}\ \bibnamefont {Fr{\'{e}}rot}}, \bibinfo {author} {\bibfnamefont {M{\'{a}}t{\'{e}}}\ \bibnamefont {Farkas}}, \ and\ \bibinfo {author} {\bibfnamefont {Antonio}\ \bibnamefont {Ac{\'{i}}n}},\ }\emph {\enquote {\bibinfo {title} {{Three numerical approaches to find mutually unbiased bases using {B}ell inequalities}},}\ }\href {https://doi.org/10.22331/q-2022-08-17-778} {\bibfield  {journal} {\bibinfo  {journal} {{Quantum}}\ }\textbf {\bibinfo {volume} {6}},\ \bibinfo {pages} {778} (\bibinfo {year} {2022})},\ \Eprint {http://arxiv.org/abs/2203.09429} {arXiv:2203.09429} [quant-ph]\BibitemShut {NoStop}%
\bibitem [{\citenamefont {Milman}\ and\ \citenamefont {Schechtman}(1986)}]{MilmanSchechtman1986}%
  \BibitemOpen
  \bibfield  {author} {\bibinfo {author} {\bibfnamefont {Vitali~D.}\ \bibnamefont {Milman}}\ and\ \bibinfo {author} {\bibfnamefont {Gideon}\ \bibnamefont {Schechtman}},\ }\href {https://doi.org/10.1007/978-3-540-38822-7} {\emph {\bibinfo {title} {{Asymptotic Theory of Finite Dimensional Normed Spaces}}}},\ Lecture Notes in Mathematics\ (\bibinfo  {publisher} {Springer},\ \bibinfo {address} {Berlin, Heidelberg},\ \bibinfo {year} {1986})\BibitemShut {NoStop}%
\bibitem [{\citenamefont {Ledoux}(2001)}]{Ledoux2001}%
  \BibitemOpen
  \bibfield  {author} {\bibinfo {author} {\bibfnamefont {Michel}\ \bibnamefont {Ledoux}},\ }\href {https://doi.org/10.1090/surv/089} {\emph {\bibinfo {title} {The Concentration of Measure Phenomenon}}},\ \bibinfo {series} {Mathematical Surveys and Monographs}, Vol.~\bibinfo {volume} {89}\ (\bibinfo  {publisher} {American Mathematical Society},\ \bibinfo {address} {Providence, Rhode Island},\ \bibinfo {year} {2001})\BibitemShut {NoStop}%
\bibitem [{\citenamefont {Hayden}\ \emph {et~al.}(2006)\citenamefont {Hayden}, \citenamefont {Leung},\ and\ \citenamefont {Winter}}]{HaydenLeungWinter2006}%
  \BibitemOpen
  \bibfield  {author} {\bibinfo {author} {\bibfnamefont {Patrick}\ \bibnamefont {Hayden}}, \bibinfo {author} {\bibfnamefont {Debbie~W.}\ \bibnamefont {Leung}}, \ and\ \bibinfo {author} {\bibfnamefont {Andreas}\ \bibnamefont {Winter}},\ }\emph {\enquote {\bibinfo {title} {Aspects of generic entanglement},}\ }\href {https://doi.org/10.1007/s00220-006-1535-6} {\bibfield  {journal} {\bibinfo  {journal} {Commun. Math. Phys.}\ }\textbf {\bibinfo {volume} {265}},\ \bibinfo {pages} {95{\textendash}117} (\bibinfo {year} {2006})},\ \Eprint {http://arxiv.org/abs/0407049} {arXiv:0407049}\BibitemShut {NoStop}%
\bibitem [{\citenamefont {Welch}(1974)}]{Welch1974}%
  \BibitemOpen
  \bibfield  {author} {\bibinfo {author} {\bibfnamefont {Lloyd~R.}\ \bibnamefont {Welch}},\ }\emph {\enquote {\bibinfo {title} {{Lower bounds on the maximum cross correlation of signals (Corresp.)}},}\ }\href {https://doi.org/10.1109/TIT.1974.1055219} {\bibfield  {journal} {\bibinfo  {journal} {IEEE Trans. Inf. Theory}\ }\textbf {\bibinfo {volume} {20}},\ \bibinfo {pages} {397{\textendash}399} (\bibinfo {year} {1974})}\BibitemShut {NoStop}%
\bibitem [{\citenamefont {Durt}\ \emph {et~al.}(2010)\citenamefont {Durt}, \citenamefont {Englert}, \citenamefont {Bengtsson},\ and\ \citenamefont {\.{Z}yczkowski}}]{DurtEnglertBengtssonZyczkowski2010}%
  \BibitemOpen
  \bibfield  {author} {\bibinfo {author} {\bibfnamefont {Thomas}\ \bibnamefont {Durt}}, \bibinfo {author} {\bibfnamefont {Berthold-Georg}\ \bibnamefont {Englert}}, \bibinfo {author} {\bibfnamefont {Ingemar}\ \bibnamefont {Bengtsson}}, \ and\ \bibinfo {author} {\bibfnamefont {Karol}\ \bibnamefont {\.{Z}yczkowski}},\ }\emph {\enquote {\bibinfo {title} {{On mutually unbiased bases}},}\ }\href {https://doi.org/10.1142/S0219749910006502} {\bibfield  {journal} {\bibinfo  {journal} {Int. J. Quantum Inf.}\ }\textbf {\bibinfo {volume} {08}},\ \bibinfo {pages} {535{\textendash}640} (\bibinfo {year} {2010})},\ \Eprint {http://arxiv.org/abs/1004.3348} {arXiv:1004.3348}\BibitemShut {NoStop}%
\bibitem [{\citenamefont {Morelli}\ \emph {et~al.}(2022)\citenamefont {Morelli}, \citenamefont {Yamasaki}, \citenamefont {Huber},\ and\ \citenamefont {Tavakoli}}]{MorelliYamasakiHuberTavakoli2022}%
  \BibitemOpen
  \bibfield  {author} {\bibinfo {author} {\bibfnamefont {Simon}\ \bibnamefont {Morelli}}, \bibinfo {author} {\bibfnamefont {Hayata}\ \bibnamefont {Yamasaki}}, \bibinfo {author} {\bibfnamefont {Marcus}\ \bibnamefont {Huber}}, \ and\ \bibinfo {author} {\bibfnamefont {Armin}\ \bibnamefont {Tavakoli}},\ }\emph {\enquote {\bibinfo {title} {Entanglement detection with imprecise measurements},}\ }\href {https://doi.org/10.1103/PhysRevLett.128.250501} {\bibfield  {journal} {\bibinfo  {journal} {Phys. Rev. Lett.}\ }\textbf {\bibinfo {volume} {128}},\ \bibinfo {pages} {250501} (\bibinfo {year} {2022})},\ \Eprint {http://arxiv.org/abs/2202.13131} {arXiv:2202.13131}\BibitemShut {NoStop}%
\bibitem [{\citenamefont {Bertlmann}\ and\ \citenamefont {Friis}(2023)}]{BertlmannFriis2023}%
  \BibitemOpen
  \bibfield  {author} {\bibinfo {author} {\bibfnamefont {Reinhold~A.}\ \bibnamefont {Bertlmann}}\ and\ \bibinfo {author} {\bibfnamefont {Nicolai}\ \bibnamefont {Friis}},\ }\href {https://doi.org/10.1093/oso/9780199683338.001.0001} {\emph {\bibinfo {title} {Modern Quantum Theory {\textendash} From Quantum Mechanics to Entanglement and Quantum Information}}}\ (\bibinfo  {publisher} {Oxford University Press},\ \bibinfo {address} {Oxford, U.K.},\ \bibinfo {year} {2023})\BibitemShut {NoStop}%
\bibitem [{\citenamefont {Collins}\ \emph {et~al.}(2002)\citenamefont {Collins}, \citenamefont {Gisin}, \citenamefont {Linden}, \citenamefont {Massar},\ and\ \citenamefont {Popescu}}]{CollinsGisinLindenMassarPopescu2002}%
  \BibitemOpen
  \bibfield  {author} {\bibinfo {author} {\bibfnamefont {D.}~\bibnamefont {Collins}}, \bibinfo {author} {\bibfnamefont {Nicolas}\ \bibnamefont {Gisin}}, \bibinfo {author} {\bibfnamefont {N.}~\bibnamefont {Linden}}, \bibinfo {author} {\bibfnamefont {S.}~\bibnamefont {Massar}}, \ and\ \bibinfo {author} {\bibfnamefont {Sandu}\ \bibnamefont {Popescu}},\ }\emph {\enquote {\bibinfo {title} {{Bell Inequalities for Arbitrarily High-Dimensional Systems}},}\ }\href {\doibase 10.1103/PhysRevLett.88.040404} {\bibfield  {journal} {\bibinfo  {journal} {Phys. Rev. Lett.}\ }\textbf {\bibinfo {volume} {88}},\ \bibinfo {pages} {040404} (\bibinfo {year} {2002})},\ \Eprint {http://arxiv.org/abs/quant-ph/0106024} {arXiv:quant-ph/0106024}\BibitemShut {NoStop}%
\bibitem [{\citenamefont {Dada}\ \emph {et~al.}(2011)\citenamefont {Dada}, \citenamefont {Leach}, \citenamefont {Buller}, \citenamefont {Padgett},\ and\ \citenamefont {Andersson}}]{DadaEtAl2011}%
  \BibitemOpen
  \bibfield  {author} {\bibinfo {author} {\bibfnamefont {Adetunmise~C.}\ \bibnamefont {Dada}}, \bibinfo {author} {\bibfnamefont {Jonathan}\ \bibnamefont {Leach}}, \bibinfo {author} {\bibfnamefont {Gerald~S.}\ \bibnamefont {Buller}}, \bibinfo {author} {\bibfnamefont {Miles~J.}\ \bibnamefont {Padgett}}, \ and\ \bibinfo {author} {\bibfnamefont {Erika}\ \bibnamefont {Andersson}},\ }\emph {\enquote {\bibinfo {title} {{Experimental high-dimensional two-photon entanglement and violations of generalized Bell inequalities}},}\ }\href {https://doi.org/10.1038/nphys1996} {\bibfield  {journal} {\bibinfo  {journal} {Nat. Phys.}\ }\textbf {\bibinfo {volume} {7}},\ \bibinfo {pages} {677{\textendash}680} (\bibinfo {year} {2011})},\ \Eprint {http://arxiv.org/abs/1104.5087} {arXiv:1104.5087}\BibitemShut {NoStop}%
\bibitem [{\citenamefont {Braunstein}\ \emph {et~al.}(1992)\citenamefont {Braunstein}, \citenamefont {Mann},\ and\ \citenamefont {Revzen}}]{BraunsteinMannRevzen1992}%
  \BibitemOpen
  \bibfield  {author} {\bibinfo {author} {\bibfnamefont {Samuel~L.}\ \bibnamefont {Braunstein}}, \bibinfo {author} {\bibfnamefont {A.}~\bibnamefont {Mann}}, \ and\ \bibinfo {author} {\bibfnamefont {M.}~\bibnamefont {Revzen}},\ }\emph {\enquote {\bibinfo {title} {{Maximal violation of Bell inequalities for mixed states}},}\ }\href {https://doi.org/10.1103/PhysRevLett.68.3259} {\bibfield  {journal} {\bibinfo  {journal} {Phys. Rev. Lett.}\ }\textbf {\bibinfo {volume} {68}},\ \bibinfo {pages} {3259} (\bibinfo {year} {1992})}\BibitemShut {NoStop}%
\bibitem [{\citenamefont {Ac\'{\i}n}\ \emph {et~al.}(2002)\citenamefont {Ac\'{\i}n}, \citenamefont {Durt}, \citenamefont {Gisin},\ and\ \citenamefont {Latorre}}]{AcinDurtGisinLatorre2002}%
  \BibitemOpen
  \bibfield  {author} {\bibinfo {author} {\bibfnamefont {A.}~\bibnamefont {Ac\'{\i}n}}, \bibinfo {author} {\bibfnamefont {T.}~\bibnamefont {Durt}}, \bibinfo {author} {\bibfnamefont {N.}~\bibnamefont {Gisin}}, \ and\ \bibinfo {author} {\bibfnamefont {J.~I.}\ \bibnamefont {Latorre}},\ }\emph {\enquote {\bibinfo {title} {Quantum nonlocality in two three-level systems},}\ }\href {https://doi.org/10.1103/PhysRevA.65.052325} {\bibfield  {journal} {\bibinfo  {journal} {Phys. Rev. A}\ }\textbf {\bibinfo {volume} {65}},\ \bibinfo {pages} {052325} (\bibinfo {year} {2002})},\ \Eprint {http://arxiv.org/abs/quant-ph/0111143} {arXiv:quant-ph/0111143}\BibitemShut {NoStop}%
\bibitem [{foo({\natexlab{d}})}]{footnote:BellIneq}%
  \BibitemOpen
  \href@noop {} {}\bibinfo {note} {To further simplify computations, Ref.\;\cite{DadaEtAl2011} uses physical arguments to restrict the maximization of the Bell operator's expectation value to restricted sets of states with different maximum Schmidt numbers (referred to therein as \textit{entanglement dimensions}), of which the union is believed to contain the experimental states.}\BibitemShut {Stop}%
\bibitem [{foo({\natexlab{e}})}]{footnote:CorrelationMatrixRFinv}%
  \BibitemOpen
  \href@noop {} {}\bibinfo {note} {Since any local unitary transformation of a bipartite state corresponds to two orthogonal transformations of the associated correlation matrix \cite{WyderkaKetterer2023}, the matrix $p$-norms for all even $p\in\mathbbm{N}$ of the correlation matrix remain unchanged.}\BibitemShut {Stop}%
\bibitem [{\citenamefont {Nakata}\ \emph {et~al.}(2021)\citenamefont {Nakata}, \citenamefont {Zhao}, \citenamefont {Okuda}, \citenamefont {Bannai}, \citenamefont {Suzuki}, \citenamefont {Tamiya}, \citenamefont {Heya}, \citenamefont {Yan}, \citenamefont {Zuo}, \citenamefont {Tamate}, \citenamefont {Tabuchi},\ and\ \citenamefont {Nakamura}}]{NakataZhaoEtAl2021}%
  \BibitemOpen
  \bibfield  {author} {\bibinfo {author} {\bibfnamefont {Yoshifumi}\ \bibnamefont {Nakata}}, \bibinfo {author} {\bibfnamefont {Da}~\bibnamefont {Zhao}}, \bibinfo {author} {\bibfnamefont {Takayuki}\ \bibnamefont {Okuda}}, \bibinfo {author} {\bibfnamefont {Eiichi}\ \bibnamefont {Bannai}}, \bibinfo {author} {\bibfnamefont {Yasunari}\ \bibnamefont {Suzuki}}, \bibinfo {author} {\bibfnamefont {Shiro}\ \bibnamefont {Tamiya}}, \bibinfo {author} {\bibfnamefont {Kentaro}\ \bibnamefont {Heya}}, \bibinfo {author} {\bibfnamefont {Zhiguang}\ \bibnamefont {Yan}}, \bibinfo {author} {\bibfnamefont {Kun}\ \bibnamefont {Zuo}}, \bibinfo {author} {\bibfnamefont {Shuhei}\ \bibnamefont {Tamate}}, \bibinfo {author} {\bibfnamefont {Yutaka}\ \bibnamefont {Tabuchi}}, \ and\ \bibinfo {author} {\bibfnamefont {Yasunobu}\ \bibnamefont {Nakamura}},\ }\emph {\enquote {\bibinfo {title} {Quantum circuits for exact unitary $t$-designs and applications to higher-order randomized benchmarking},}\ }\href
  {https://doi.org/10.1103/PRXQuantum.2.030339} {\bibfield  {journal} {\bibinfo  {journal} {PRX Quantum}\ }\textbf {\bibinfo {volume} {2}},\ \bibinfo {pages} {030339} (\bibinfo {year} {2021})},\ \Eprint {http://arxiv.org/abs/2102.12617} {arXiv:2102.12617}\BibitemShut {NoStop}%
\bibitem [{foo({\natexlab{f}})}]{footnote:ApproxTdesignSample}%
  \BibitemOpen
  \href@noop {} {}\bibinfo {note} {While approximate sampling from $t$-designs can be efficient~\cite{HarrowMehraban2023}, it is unclear how approximate sampling can affect the Schmidt-number witness in Ref.\;\cite{WyderkaKetterer2023}.}\BibitemShut {Stop}%
\bibitem [{\citenamefont {Wyderka}\ and\ \citenamefont {Ketterer}(2023)}]{WyderkaKetterer2023}%
  \BibitemOpen
  \bibfield  {author} {\bibinfo {author} {\bibfnamefont {Nikolai}\ \bibnamefont {Wyderka}}\ and\ \bibinfo {author} {\bibfnamefont {Andreas}\ \bibnamefont {Ketterer}},\ }\emph {\enquote {\bibinfo {title} {Probing the geometry of correlation matrices with randomized measurements},}\ }\href {https://doi.org/10.1103/PRXQuantum.4.020325} {\bibfield  {journal} {\bibinfo  {journal} {PRX Quantum}\ }\textbf {\bibinfo {volume} {4}},\ \bibinfo {pages} {020325} (\bibinfo {year} {2023})},\ \Eprint {http://arxiv.org/abs/2211.09610} {arXiv:2211.09610}\BibitemShut {NoStop}%
\bibitem [{foo({\natexlab{g}})}]{footnote:CorrMatrixConstraints}%
  \BibitemOpen
  \href@noop {} {}\bibinfo {note} {Since there are no tighter known constraints on the singular values of the correlation matrix other than the purity bound $\tr(\rho^2)<1$ (see Appendix C of Ref.\;\cite{WyderkaKetterer2023}), the bounds on the 4-norm are likely to be loose.}\BibitemShut {Stop}%
\bibitem [{\citenamefont {Shor}(1994)}]{Shor1994}%
  \BibitemOpen
  \bibfield  {author} {\bibinfo {author} {\bibfnamefont {Peter~W.}\ \bibnamefont {Shor}},\ }\emph {\enquote {\bibinfo {title} {{Algorithms for quantum computation: Discrete logarithms and factoring}},}\ }in\ \href {https://doi.org/10.1109/SFCS.1994.365700} {\emph {\bibinfo {booktitle} {{Proceedings 35th Annual Symposium on Foundations of Computer Science}}}}\ (\bibinfo  {publisher} {{IEEE}},\ \bibinfo {year} {1994})\ p.\ \bibinfo {pages} {124{\textendash}134}\BibitemShut {NoStop}%
\bibitem [{foo({\natexlab{h}})}]{footnote:Advantage3MUBs}%
  \BibitemOpen
  \href@noop {} {}\bibinfo {note} {This may happen in scenarios where one has a dynamical encoding protocol that encodes information in different subspaces of the physical system at different times to protect against some time-dependent noise which affects various parts of the system in distinct ways.}\BibitemShut {Stop}%
\bibitem [{foo({\natexlab{i}})}]{footnote:RelPhaseChanges}%
  \BibitemOpen
  \href@noop {} {}\bibinfo {note} {For the simplest example, if one needs to encode information or measure in 3 MUBs and the corresponding Hilbert space switches from dimension 7 to 6, then our simple construction only requires a slight change in the phases (especially with the free choices of $f(j)$ and $p^r$ for compensation), whereas using the Wootters-Fields bases will need to switch from the prime-dimension Wootters-Fields basis for $d=7$ to the tensor products of the Wootters-Fields bases of dimensions 2 and 3, which introduce more drastic changes to the relative phases. For large $d$, the change to the relative phases can be even more drastic.}\BibitemShut {Stop}%
\bibitem [{\citenamefont {Shparlinski}\ and\ \citenamefont {Winterhof}(2006)}]{ShparlinskiWinterhof2006}%
  \BibitemOpen
  \bibfield  {author} {\bibinfo {author} {\bibfnamefont {Igor~E.}\ \bibnamefont {Shparlinski}}\ and\ \bibinfo {author} {\bibfnamefont {Arne}\ \bibnamefont {Winterhof}},\ }\emph {\enquote {\bibinfo {title} {Constructions of approximately mutually unbiased bases},}\ }in\ \href {https://doi.org/10.1007/11682462_72} {\emph {\bibinfo {booktitle} {LATIN 2006: Theoretical Informatics}}},\ \bibinfo {editor} {edited by\ \bibinfo {editor} {\bibfnamefont {Jos{\'e}~R.}\ \bibnamefont {Correa}}, \bibinfo {editor} {\bibfnamefont {Alejandro}\ \bibnamefont {Hevia}}, \ and\ \bibinfo {editor} {\bibfnamefont {Marcos}\ \bibnamefont {Kiwi}}}\ (\bibinfo  {publisher} {Springer},\ \bibinfo {address} {Berlin, Heidelberg},\ \bibinfo {year} {2006})\ p.\ \bibinfo {pages} {793{\textendash}799}\BibitemShut {NoStop}%
\bibitem [{\citenamefont {Bhatia}(1996)}]{Bhatia1996}%
  \BibitemOpen
  \bibfield  {author} {\bibinfo {author} {\bibfnamefont {Rajendra}\ \bibnamefont {Bhatia}},\ }\href {https://doi.org/10.1007/978-1-4612-0653-8} {\emph {\bibinfo {title} {{Matrix Analysis}}}},\ \bibinfo {edition} {1st}\ ed.\ (\bibinfo  {publisher} {Springer},\ \bibinfo {address} {Berlin},\ \bibinfo {year} {1996})\BibitemShut {NoStop}%
\bibitem [{\citenamefont {Bertsekas}(1999)}]{Bertsekas1999}%
  \BibitemOpen
  \bibfield  {author} {\bibinfo {author} {\bibfnamefont {Dimitri~P.}\ \bibnamefont {Bertsekas}},\ }\href@noop {} {\emph {\bibinfo {title} {{Nonlinear Programming}}}},\ \bibinfo {edition} {2nd}\ ed.\ (\bibinfo  {publisher} {Athena Scientific},\ \bibinfo {address} {Massachusetts},\ \bibinfo {year} {1999})\BibitemShut {NoStop}%
\bibitem [{foo({\natexlab{j}})}]{footnote:irregular}%
  \BibitemOpen
  \href@noop {} {}\bibinfo {note} {Note that the only feasible solution of the case when $c_\text{min}=c_\text{max}$ are non-regular since both $g_i(\vec{x})\leq 0$ and $g_{i+d}(\vec{x})\leq 0$ are active for all $i$, so $\nabla h(\vec{x}), \nabla g_i(\vec{x})$, and $\nabla g_{i+d}(\vec{x})$ for $i=1,\ldots,d$ are linearly dependent.}\BibitemShut {Stop}%
\bibitem [{\citenamefont {Gross}(2006)}]{Gross2006}%
  \BibitemOpen
  \bibfield  {author} {\bibinfo {author} {\bibfnamefont {David}\ \bibnamefont {Gross}},\ }\emph {\enquote {\bibinfo {title} {{Hudson's theorem for finite-dimensional quantum systems}},}\ }\href {\doibase 10.1063/1.2393152} {\bibfield  {journal} {\bibinfo  {journal} {J. Math. Phys.}\ }\textbf {\bibinfo {volume} {47}},\ \bibinfo {pages} {122107} (\bibinfo {year} {2006})},\ \Eprint {http://arxiv.org/abs/quant-ph/0602001} {arXiv:quant-ph/0602001}\BibitemShut {NoStop}%
\bibitem [{\citenamefont {Kueng}\ and\ \citenamefont {Gross}(2015)}]{KuengGross2015}%
  \BibitemOpen
  \bibfield  {author} {\bibinfo {author} {\bibfnamefont {Richard}\ \bibnamefont {Kueng}}\ and\ \bibinfo {author} {\bibfnamefont {David}\ \bibnamefont {Gross}},\ }\href@noop {} {\emph {\enquote {\bibinfo {title} {{Qubit stabilizer states are complex projective 3-designs}},}\ }}\Eprint {http://arxiv.org/abs/1510.02767} {arXiv:1510.02767} (\bibinfo {year} {2015})\BibitemShut {NoStop}%
\bibitem [{\citenamefont {Rungta}\ and\ \citenamefont {Caves}(2003)}]{RungtaCaves2003}%
  \BibitemOpen
  \bibfield  {author} {\bibinfo {author} {\bibfnamefont {Pranaw}\ \bibnamefont {Rungta}}\ and\ \bibinfo {author} {\bibfnamefont {Carlton~M.}\ \bibnamefont {Caves}},\ }\emph {\enquote {\bibinfo {title} {Concurrence-based entanglement measures for isotropic states},}\ }\href {https://doi.org/10.1103/PhysRevA.67.012307} {\bibfield  {journal} {\bibinfo  {journal} {Phys. Rev. A}\ }\textbf {\bibinfo {volume} {67}},\ \bibinfo {pages} {012307} (\bibinfo {year} {2003})},\ \Eprint {http://arxiv.org/abs/quant-ph/0208002} {arXiv:quant-ph/0208002}\BibitemShut {NoStop}%
\bibitem [{\citenamefont {Lee}\ \emph {et~al.}(2003)\citenamefont {Lee}, \citenamefont {Chi}, \citenamefont {Oh},\ and\ \citenamefont {Kim}}]{LeeChiOhKim2003}%
  \BibitemOpen
  \bibfield  {author} {\bibinfo {author} {\bibfnamefont {Soojoon}\ \bibnamefont {Lee}}, \bibinfo {author} {\bibfnamefont {Dong~Pyo}\ \bibnamefont {Chi}}, \bibinfo {author} {\bibfnamefont {Sung~Dahm}\ \bibnamefont {Oh}}, \ and\ \bibinfo {author} {\bibfnamefont {Jaewan}\ \bibnamefont {Kim}},\ }\emph {\enquote {\bibinfo {title} {Convex-roof extended negativity as an entanglement measure for bipartite quantum systems},}\ }\href {https://doi.org/10.1103/PhysRevA.68.062304} {\bibfield  {journal} {\bibinfo  {journal} {Phys. Rev. A}\ }\textbf {\bibinfo {volume} {68}},\ \bibinfo {pages} {062304} (\bibinfo {year} {2003})},\ \Eprint {http://arxiv.org/abs/quant-ph/0310027} {arXiv:quant-ph/0310027}\BibitemShut {NoStop}%
\bibitem [{\citenamefont {Harrow}\ and\ \citenamefont {Mehraban}(2023)}]{HarrowMehraban2023}%
  \BibitemOpen
  \bibfield  {author} {\bibinfo {author} {\bibfnamefont {Aram~W.}\ \bibnamefont {Harrow}}\ and\ \bibinfo {author} {\bibfnamefont {Saeed}\ \bibnamefont {Mehraban}},\ }\emph {\enquote {\bibinfo {title} {{Approximate Unitary t-Designs by Short Random Quantum Circuits Using Nearest-Neighbor and Long-Range Gates}},}\ }\href {https://doi.org/10.1007/s00220-023-04675-z} {\bibfield  {journal} {\bibinfo  {journal} {Commun. Math. Phys.}\ }\textbf {\bibinfo {volume} {401}},\ \bibinfo {pages} {1531{\textendash}1626} (\bibinfo {year} {2023})},\ \Eprint {http://arxiv.org/abs/1809.06957} {arXiv:1809.06957}\BibitemShut {NoStop}%
\bibitem [{Note1()}]{Note1}%
  \BibitemOpen
  \bibinfo {note} {The Schmidt-number witness in Ref.\protect \tmspace +\thickmuskip {.2777em}\cite {BavarescoEtAl2018} is proven to be optimal only for all pure states and dephased maximally entangled states, $(1-p)\mathinner {|{\Phi ^+_{d}}\rangle }\protect \!\protect \!\mathinner {\langle {\Phi ^+_{d}}|}+\protect \frac {p}{d}\DOTSB \tsum \slimits@ _{i=0}^{d-1} \mathinner {|{ii}\rangle }\protect \!\protect \!\mathinner {\langle {ii}|}$\spacefactor \@m {}.}\BibitemShut {Stop}%
\bibitem [{Note2()}]{Note2}%
  \BibitemOpen
  \bibinfo {note} {Note that we plot $\protect \sqrt {\Delta \protect \overline {\protect \mathcal {F}}\protect \ensuremath {^{\protect \hspace {1 pt}\protect \raisebox {-1.5 pt}{\relax \protect \fontsize {5}{6}\protect \selectfont {$ \protect \mathrm {iso}$}}}}_m}$ instead of $\Delta \protect \overline {\protect \mathcal {F}}\protect \ensuremath {^{\protect \hspace {1 pt}\protect \raisebox {-1.5 pt}{\relax \protect \fontsize {5}{6}\protect \selectfont {$ \protect \mathrm {iso}$}}}}_m$ for clearer presentation because the curves appear too close to one another when we plot $\Delta \protect \overline {\protect \mathcal {F}}\protect \ensuremath {^{\protect \hspace {1 pt}\protect \raisebox {-1.5 pt}{\relax \protect \fontsize {5}{6}\protect \selectfont {$ \protect \mathrm {iso}$}}}}_m$ versus $p$\spacefactor \@m {}.}\BibitemShut {Stop}%
\bibitem [{\citenamefont {Ivonovi\'{c}}(1981)}]{Ivonovic1981}%
  \BibitemOpen
  \bibfield  {author} {\bibinfo {author} {\bibfnamefont {I.~D.}\ \bibnamefont {Ivonovi\'{c}}},\ }\emph {\enquote {\bibinfo {title} {Geometrical description of quantal state determination},}\ }\href {https://doi.org/10.1088/0305-4470/14/12/019} {\bibfield  {journal} {\bibinfo  {journal} {J. Phys. A: Math. Gen.}\ }\textbf {\bibinfo {volume} {14}},\ \bibinfo {pages} {3241} (\bibinfo {year} {1981})}\BibitemShut {NoStop}%
\bibitem [{foo({\natexlab{k}})}]{footnote:p_th_symm_theta}%
  \BibitemOpen
  \href@noop {} {}\bibinfo {note} {The only $\theta$-dependent term in Eq.\;\eqref{eq:p^k_th,m} is the quantity $\overline{\mathcal{T}}(\overline{\mathcal{C}})$ which depends only on $c^{z,z'}_\text{min}$ and $c^{z,z'}_\text{max}$ (see Lemma~\ref{lemma:SchmidtRkLoose}). Since $c^{z,z'}_\text{min}$ and $c^{z,z'}_\text{max}$ in Eqs.\;\eqref{eq:c^1z_thermal} and~\eqref{eq:cmin_max_th} are invariant under $\theta\leftrightarrow 2\pi-\theta$ for all $z\neq z'$, $p\suptiny{1}{0}{(k)}_{\text{th},m}(\beta,\theta)=p\suptiny{1}{0}{(k)}_{\text{th},m}(\beta,2\pi-\theta)$.}\BibitemShut {Stop}%
\bibitem [{Note3()}]{Note3}%
  \BibitemOpen
  \bibinfo {note} {Since any unitary $U$ has a spectral decomposition $U=\DOTSB \tsum \slimits@ _i \lambda _i|\lambda _i\rangle \protect \!\langle \lambda _i|$, where all eigenvalues satisfy $|\lambda _i|=1$ and $\{|\lambda _i\rangle \}_i$ is an orthonormal basis, any matrix element of $U$ satisfies $|U_{jk}| = |\DOTSB \tsum \slimits@ _i \lambda _i\langle j|\lambda _i\rangle \protect \!\langle \lambda _i|k\rangle | \leq \protect \sqrt {\DOTSB \tsum \slimits@ _i |\lambda _i||\langle j|\lambda _i\rangle |^2}\protect \sqrt {\DOTSB \tsum \slimits@ _i |\lambda _i||\langle k|\lambda _i\rangle |^2} = 1$, where we use the Cauchy{\textendash }Schwarz inequality.}\BibitemShut {Stop}%
\bibitem [{\citenamefont {Chang}\ \emph {et~al.}(2024)\citenamefont {Chang}, \citenamefont {Sarihan}, \citenamefont {Cheng}, \citenamefont {Erker}, \citenamefont {Li}, \citenamefont {Mueller}, \citenamefont {Spiropulu}, \citenamefont {Shaw}, \citenamefont {Korzh}, \citenamefont {Huber},\ and\ \citenamefont {Wong}}]{ChangSarihanChengErkerLiEtAl2024}%
  \BibitemOpen
  \bibfield  {author} {\bibinfo {author} {\bibfnamefont {Kai-Chi}\ \bibnamefont {Chang}}, \bibinfo {author} {\bibfnamefont {Murat~C.}\ \bibnamefont {Sarihan}}, \bibinfo {author} {\bibfnamefont {Xiang}\ \bibnamefont {Cheng}}, \bibinfo {author} {\bibfnamefont {Paul}\ \bibnamefont {Erker}}, \bibinfo {author} {\bibfnamefont {Nicky Kai~Hong}\ \bibnamefont {Li}}, \bibinfo {author} {\bibfnamefont {Andrew}\ \bibnamefont {Mueller}}, \bibinfo {author} {\bibfnamefont {Maria}\ \bibnamefont {Spiropulu}}, \bibinfo {author} {\bibfnamefont {Matthew~D.}\ \bibnamefont {Shaw}}, \bibinfo {author} {\bibfnamefont {Boris}\ \bibnamefont {Korzh}}, \bibinfo {author} {\bibfnamefont {Marcus}\ \bibnamefont {Huber}}, \ and\ \bibinfo {author} {\bibfnamefont {Chee~Wei}\ \bibnamefont {Wong}},\ }\href@noop {} {\emph {\enquote {\bibinfo {title} {Experimental high-dimensional entanglement certification and quantum steering with time-energy measurements},}\ }}\Eprint {http://arxiv.org/abs/2310.20694} {arXiv:2310.20694} (\bibinfo {year}
  {2024})\BibitemShut {NoStop}%
\bibitem [{\citenamefont {Euler}\ and\ \citenamefont {G\"{a}rttner}(2023)}]{EulerGaerttner2023}%
  \BibitemOpen
  \bibfield  {author} {\bibinfo {author} {\bibfnamefont {Niklas}\ \bibnamefont {Euler}}\ and\ \bibinfo {author} {\bibfnamefont {Martin}\ \bibnamefont {G\"{a}rttner}},\ }\emph {\enquote {\bibinfo {title} {{Detecting high-dimensional entanglement in cold-atom quantum simulators}},}\ }\href {\doibase 10.1103/PRXQuantum.4.040338} {\bibfield  {journal} {\bibinfo  {journal} {PRX Quantum}\ }\textbf {\bibinfo {volume} {4}},\ \bibinfo {pages} {040338} (\bibinfo {year} {2023})},\ \Eprint {http://arxiv.org/abs/2305.07413} {arXiv:2305.07413}\BibitemShut {NoStop}%
\bibitem [{Note4()}]{Note4}%
  \BibitemOpen
  \bibinfo {note} {Note that for all odd-prime dimensions $d$, $\protect \tilde {F}\protect \ensuremath {^{\protect \hspace {1 pt}\protect \raisebox {0 pt}{\relax \protect \fontsize {5}{6}\protect \selectfont {$ (M')$}}}}\geq \protect \tilde {F}\protect \ensuremath {^{\protect \hspace {1 pt}\protect \raisebox {0 pt}{\relax \protect \fontsize {5}{6}\protect \selectfont {$ (M)$}}}}$ holds for all $M'\geq M$, whereas for all other dimensions, only $\protect \tilde {F}\protect \ensuremath {^{\protect \hspace {1 pt}\protect \raisebox {0 pt}{\relax \protect \fontsize {5}{6}\protect \selectfont {$ (M)$}}}}\geq \protect \tilde {F}\protect \ensuremath {^{\protect \hspace {1 pt}\protect \raisebox {0 pt}{\relax \protect \fontsize {5}{6}\protect \selectfont {$ (1)$}}}}$ is guaranteed to hold $\forall \protect \tmspace +\thickmuskip {.2777em}M\geq 1$~\cite {BavarescoEtAl2018}.}\BibitemShut {Stop}%
\bibitem [{\citenamefont {Dirichlet}(1829)}]{Dirichlet1829}%
  \BibitemOpen
  \bibfield  {author} {\bibinfo {author} {\bibfnamefont {Johann Peter Gustav~Lejeune}\ \bibnamefont {Dirichlet}},\ }\emph {\enquote {\bibinfo {title} {Sur la convergence des s\'{e}ries trigonom\'{e}triques qui servent \`{a} repr\'{e}senter une fonction arbitraire entre des limites donn\'{e}es.}}\ }\href {https://doi.org/10.1515/crll.1829.4.157} {\bibfield  {journal} {\bibinfo  {journal} {Journal f\"{u}r die reine und angewandte Mathematik}\ }\textbf {\bibinfo {volume} {4}},\ \bibinfo {pages} {157{\textendash}169} (\bibinfo {year} {1829})},\ \bibinfo {note} {\href{http://eudml.org/doc/183134}{http://eudml.org/doc/183134}}\BibitemShut {NoStop}%
\bibitem [{\citenamefont {Blumenson}(1960)}]{Blumenson1960}%
  \BibitemOpen
  \bibfield  {author} {\bibinfo {author} {\bibfnamefont {L.~E.}\ \bibnamefont {Blumenson}},\ }\emph {\enquote {\bibinfo {title} {{A Derivation of n-Dimensional Spherical Coordinates}},}\ }\href {https://doi.org/10.2307/2308932} {\bibfield  {journal} {\bibinfo  {journal} {Am. Math. Mon.}\ }\textbf {\bibinfo {volume} {67}},\ \bibinfo {pages} {63{\textendash}66} (\bibinfo {year} {1960})}\BibitemShut {NoStop}%
\bibitem [{\citenamefont {Berndt}\ \emph {et~al.}(1998)\citenamefont {Berndt}, \citenamefont {Evans},\ and\ \citenamefont {Williams}}]{BerndtEvansWilliams1998}%
  \BibitemOpen
  \bibfield  {author} {\bibinfo {author} {\bibfnamefont {{Bruce C.}}\ \bibnamefont {Berndt}}, \bibinfo {author} {\bibfnamefont {{Ronald J.}}\ \bibnamefont {Evans}}, \ and\ \bibinfo {author} {\bibfnamefont {{Kenneth S.}}\ \bibnamefont {Williams}},\ }\href@noop {} {\emph {\bibinfo {title} {{Gauss and Jacobi Sums}}}},\ Wiley-Interscience and Canadian Mathematics Series of Monographs and Texts\ (\bibinfo  {publisher} {Wiley},\ \bibinfo {address} {New York},\ \bibinfo {year} {1998})\BibitemShut {NoStop}%
\bibitem [{\citenamefont {Cochrane}\ and\ \citenamefont {Peral}(2001)}]{CochranePeral2001}%
  \BibitemOpen
  \bibfield  {author} {\bibinfo {author} {\bibfnamefont {Todd}\ \bibnamefont {Cochrane}}\ and\ \bibinfo {author} {\bibfnamefont {J.~C.}\ \bibnamefont {Peral}},\ }\emph {\enquote {\bibinfo {title} {{An Asymptotic Formula for a Trigonometric Sum of Vinogradov}},}\ }\href {https://doi.org/10.1006/jnth.2001.2679} {\bibfield  {journal} {\bibinfo  {journal} {J. Number Theory}\ }\textbf {\bibinfo {volume} {91}},\ \bibinfo {pages} {1{\textendash}19} (\bibinfo {year} {2001})}\BibitemShut {NoStop}%
\bibitem [{\citenamefont {Bertrand}(1845)}]{Bertrand1845}%
  \BibitemOpen
  \bibfield  {author} {\bibinfo {author} {\bibfnamefont {Joseph}\ \bibnamefont {Bertrand}},\ }\emph {\enquote {\bibinfo {title} {M\'{e}moire sur le nombre de valeurs que peut prendre une fonction quand on y permute les lettres qu'elle renferme.}}\ }\href@noop {} {\bibfield  {journal} {\bibinfo  {journal} {Journal de l'\'{E}cole Royale Polytechnique}\ }\textbf {\bibinfo {volume} {18}},\ \bibinfo {pages} {123{\textendash}140} (\bibinfo {year} {1845})},\ \bibinfo {note} {in {French}, available at \href{https://gallica.bnf.fr/ark:/12148/bpt6k4336867/f126}{https://gallica.bnf.fr/ark:/12148/bpt6k4336867/f126}, last accessed 6 June 2024}\BibitemShut {NoStop}%
\bibitem [{\citenamefont {Chebyshev}(1852)}]{Chebyshev1852}%
  \BibitemOpen
  \bibfield  {author} {\bibinfo {author} {\bibfnamefont {Pafnuty}\ \bibnamefont {Chebyshev}},\ }\emph {\enquote {\bibinfo {title} {M\'{e}moire sur les nombres premiers.}}\ }\href@noop {} {\bibfield  {journal} {\bibinfo  {journal} {Journal de math\'{e}matiques pures et appliqu\'{e}es, S\'{e}rie 1}\ ,\ \bibinfo {pages} {366{\textendash}390}} (\bibinfo {year} {1852})},\ \bibinfo {note} {in {French}, available at \href{http://www.numdam.org/item/JMPA_1852_1_17__366_0/}{http://www.numdam.org/item/JMPA$\_$1852$\_$1$\_$17$\_\_$366$\_$0/}, last accessed 6 June 2024}\BibitemShut {NoStop}%
\bibitem [{OEI()}]{OEIS}%
  \BibitemOpen
  \href@noop {} {}\bibinfo {note} {N. J. A. Sloane and The OEIS Foundation Inc., ``The online encyclopedia of integer sequences". Retrieved from \url{http://oeis.org/} on \today.}\BibitemShut {Stop}%
\bibitem [{Note5()}]{Note5}%
  \BibitemOpen
  \bibinfo {note} {They are all products of 2 and another odd prime, so the tensor products of the Wootters{\textendash }Fields constructions will only guarantee 3 MUBs in $d=6,10,14$, and 22.}\BibitemShut {Stop}%
\end{thebibliography}%


\hypertarget{sec:appendix}
\appendix

\section*{Appendix: Supplemental Information}

\renewcommand{\thesubsubsection}{A.\Roman{subsection}.\arabic{subsubsection}}
\renewcommand{\thesubsection}{A.\Roman{subsection}}
\renewcommand{\thesection}{}
\setcounter{equation}{0}
\numberwithin{equation}{section}
\setcounter{figure}{0}
\renewcommand{\theequation}{A.\arabic{equation}}
\renewcommand{\thefigure}{A.\arabic{figure}}


\noindent This appendix, where we present additional details and explicit calculations that support our results, is structured as follows: in Appendix~\ref{appendix:Examples}, we provide additional analyses on the performance of our Schmidt-number witness and lower bound of the entanglement fidelity (or the singlet fraction \cite{BennettDiVincenzoSmolinWootters1996,HorodeckiMPR1999}) when applied to isotropic states and noisy purified thermal states. In the same section, we also compare the noise tolerance of our witness with that of Ref.\;\cite{BavarescoEtAl2018} when applied to the two types of states, but the full details of applying Ref.\;\cite{BavarescoEtAl2018}'s witness are deferred to Appendix~\ref{app:JessicaWitness}. In Appendix~\ref{app:COM_Levy}, we use L\'{e}vy's lemma to prove that any two bases from a set of orthonormal bases which are chosen uniformly at random have exponentially decreasing probability of being biased as the dimension increases. In Appendix~\ref{app:onbNumBound}, we prove Ineq.\;\eqref{ineq:mBound} which relates the maximal number of orthonormal bases to the function $\lambda(\mathcal{C})$ defined in Theorem~\ref{thm:SchmidtRkBound}. In Appendix~\ref{app:3MUBsAnyD}, we present the proof of Lemma~\ref{lemma:3MUBsAnyD} which states a new construction of three MUBs that has a simple analytic form and works for any dimension $d\in\mathbbm{N}$\@. Finally, in Appendix~\ref{app:AMUBs}, we explore whether using AMUBs has any advantage in witnessing the Schmidt number using our witness in dimensions where the maximum number of MUBs is unknown.


\subsection{Examples for violation of our Schmidt-number witness}\label{appendix:Examples}
\vspace*{-1.5mm}
In this appendix, we provide the full details of how we certify the Schmidt number and lower bound the entanglement fidelity of isotropic states in Appendix~\ref{app:isotropic} and the purified thermal states mixed with white noise in Appendix~\ref{app:thermal}.

\vspace*{-1.5mm}
\subsubsection{Isotropic states}\label{app:isotropic}
\vspace*{-1.5mm}
Let us first recall some properties of an isotropic state, i.e., a qudit Bell state mixed with a certain amount of white noise: $\rho\subtiny{0}{0}{A\nl B}\suptiny{1}{0}{\mathrm{iso}}= (1-p)\ket{\Phi^+_{d}}\!\!\bra{\Phi^+_{d}}+\frac{p}{d^2}\identity_{d^2}$\@. It has been shown that its Schmidt number is exactly $k+1$ if and only if the white-noise ratio satisfies $\frac{d(d-k-1)}{d^2-1} \leq p < 1-\frac{kd-1}{d^2-1} = \frac{d(d-k)}{d^2-1} \eqcolon p_\text{iso}\suptiny{1}{0}{(k)}$~\cite{TerhalHorodecki2000}. As mentioned in the main text,
\begin{equation}\label{eq:Eg1_app}
    \mathcal{S}_{d}\suptiny{1}{0}{(m)}(\rho\subtiny{0}{0}{A\nl B}\suptiny{1}{0}{\mathrm{iso}}) = p\frac{m}{d} + (1-p)m
\end{equation} 
and for it to exceed the bound $\mathcal{B}_k$ in Theorem~\ref{thm:SchmidtRkBound}, the white-noise ratio must satisfy
\begin{equation}\label{ineq:noiseBnd_app}
    p < \frac{(m-\mathcal{T}(\mathcal{C}))(d-k)}{m(d-1)} \eqcolon p_{c,m}\suptiny{1}{0}{(k)}.
\end{equation}
In the case when $d+1$ MUBs exist and $m=d+1$, it holds that $p_{c,m}\suptiny{1}{0}{(k)} = p_\text{iso}\suptiny{1}{0}{(k)}$ for all~$k$\@.\\[-3mm]

To illustrate that our lower bound for entanglement fidelity [Ineq.\;\eqref{ineq:Fbound} in Theorem~\ref{thm:SchmidtRkBound}] is valid, we will show that the inequality holds for isotropic states. As the entanglement fidelity of $\rho\subtiny{0}{0}{A\nl B}\suptiny{1}{0}{\mathrm{iso}}$ is given by
\begin{equation}
    \mathcal{F}(\rho\subtiny{0}{0}{A\nl B}\suptiny{1}{0}{\mathrm{iso}}) = 1-p+\frac{p}{d^2},
\end{equation}
it is easy to see that Ineq.\;\eqref{ineq:Fbound} holds whenever $m\leq \mathcal{T}(\mathcal{C})(d+1)$\@. When $\mathcal{T}(\mathcal{C})=m$, this inequality is trivially satisfied. When $\mathcal{T}(\mathcal{C})=\lambda(\mathcal{C}) < m$, the upper bound of orthonormal bases, $\overline{m}_d$, defined in Corollary~\ref{corollary:MaxBasesNum} satisfies $\overline{m}_d < \lambda(\mathcal{C})(d+1)$ for all $d\geq 2$\@. Therefore, we have $m\leq \overline{m}_d < \mathcal{T}(\mathcal{C})(d+1)$, implying that $\mathcal{F}(\rho\subtiny{0}{0}{A\nl B}\suptiny{1}{0}{\mathrm{iso}}) \geq \mathcal{F}_m$\@.

An important aspect that we should consider is how the performance of our Schmidt-number witness scales with the local dimension. With $m$ ``worst-case" measurement bases parametrized by $c_\text{min}$~\cite{footnote:Eg1WorstCase}, which corresponds to the same measurement settings of the main text's example, it must hold that
{\begin{small}
\begin{equation}
    c_\text{min} > \frac{3d-1-\sqrt{(d+1)^2+\frac{2(d-1)}{m(m-1)} \left\{\left[1\!+\!2m\left(\frac{p(d-1)}{d-k}\!-\!1\right)\right]^2\!\!-\!1\right\}}}{2d(d-1)} \label{ineq:BasesBiasTolerance_Iso}
\end{equation}
\end{small}}for our method to witness the Schmidt number of the isotropic state $\rho\suptiny{1}{0}{\mathrm{iso}}$ with white-noise ratio $p$ to be at least $k+1$\@.
In Fig.\;\ref{fig:BasesBiasnessVSdim}, we plot the suprema of the dimension-rescaled bases bias, $d\nr\epsilon_\text{min} \coloneqq 1- d\nr c_\text{min} \in [0,1]$, below which our method using two measurement bases can witness the Schmidt number of $\rho\suptiny{1}{0}{\mathrm{iso}}$ with $p=0.005$ to be at least $k+1$, against the local dimension $d$\@. As a benchmark, we compare these with the suprema of bases-bias tolerance for witnessing Schmidt numbers $2\leq k+1\leq d$ in the maximally entangled state (i.e., $p=0$), which turns out to be independent of~$k$ as one can see from Ineq.\;\eqref{ineq:BasesBiasTolerance_Iso} when $p=0$\@. At least for the worst-case bases choice, this figure and Ineq.\;\eqref{ineq:BasesBiasTolerance_Iso} suggest that our witness tolerates less and less bias in the measurement bases until it cannot certify any Schmidt number of isotropic states as the local dimension increases. We suspect that this is a general feature of our witness as the bounds $\mathcal{B}_k$ ($\overline{\mathcal{B}_k}$) in Theorem~\ref{thm:SchmidtRkBound} (Lemma~\ref{lemma:SchmidtRkLoose}) have a bases-bias-dependent term that grows at least in $\mathcal{O}(\sqrt{d})$, reducing the tolerance to bases bias in larger dimensions.\\[-3mm]

\begin{figure}[h!]
    \centering
    \includegraphics[width=\linewidth]{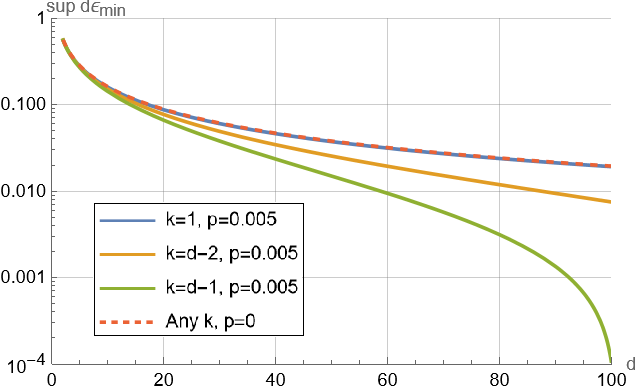}
    \caption{The suprema of the dimension-rescaled bases bias, $d\nr\epsilon_\text{min} \coloneqq 1- d\nr c_\text{min} \in [0,1]$, below which our method using 2 measurement bases can witness the Schmidt number of the isotropic state with $p=0.005$ to be at least $k+1$, are plotted on a log scale against the local dimension $d$\@. These are compared with the supremum of bases-bias tolerance for witnessing any Schmidt number $2\leq k+1\leq d$ in the maximally entangled state (i.e., $p=0$), represented by the dotted curve.}\label{fig:BasesBiasnessVSdim}
\end{figure}

Next, we will compare the performance of our witness with the one proposed in Ref.\;\cite{BavarescoEtAl2018} using the example of isotropic states $\rho\subtiny{0}{0}{A\nl B}\suptiny{1}{0}{\mathrm{iso}}$ \footnote{The Schmidt-number witness in Ref.\;\cite{BavarescoEtAl2018} is proven to be optimal only for all pure states and dephased maximally entangled states, $(1-p)\ket{\Phi^+_{d}}\!\!\bra{\Phi^+_{d}}+\frac{p}{d}\sum_{i=0}^{d-1} \ket{ii}\!\!\bra{ii}$\@.}. The main idea behind the Schmidt-number witness proposed in Ref.\;\cite{BavarescoEtAl2018} is summarized in Appendix~\ref{app:JessicaWitnessSummary}. The analytic expressions of the entanglement fidelity lower bound and the white-noise tolerance for each Schmidt number of the isotropic states $\rho\subtiny{0}{0}{A\nl B}\suptiny{1}{0}{\mathrm{iso}}$ are derived in Appendix~\ref{app:JessicaWitnessIsotropic}.\\[-3mm]

In Fig.\;\ref{fig:compareJessica_iso}, we fix the local dimension $d=5$ and plot the maximal white-noise ratios, $p_{c,m}\suptiny{1}{0}{(\nl k\nl\nl=\nl\nl4\nl)}$ and $\tilde{p}_{\text{iso},M}\suptiny{1}{0}{(\nl k\nl\nl=\nl\nl4\nl)}$, below which our witness and the witness in Ref.\;\cite{BavarescoEtAl2018} can certify $\rho\suptiny{1}{0}{\mathrm{iso}}\subtiny{0}{0}{A\nl B}$'s Schmidt number to be $k+1=5$, respectively. Note that $m$ denotes the total number of orthonormal measurement bases used in our setting, whereas $M$ denotes the total number of ``tilted" (potentially non-orthogonal) bases used in Ref.\;\cite{BavarescoEtAl2018}, which does not count the computational basis that they also need to measure in (see Appendix~\ref{app:JessicaWitnessSummary}). For isotropic states, the ``tilted" bases are orthonormal and form a set of MUBs (see Appendix~\ref{app:JessicaWitnessIsotropic}). Hence, we can make the association, $m=M+1$, for an intuitive comparison between Ref.\;\cite{BavarescoEtAl2018}'s and our methods. As the minimum bases overlap $c_\text{min}$ reduces from the maximum value $\frac{1}{d}=0.2$ (i.e., $\epsilon_\text{min}=\frac{1}{d}-c_\text{min}$ increases from 0), our white-noise threshold $p_{c,m}\suptiny{1}{0}{(\nl k\nl\nl=\nl\nl4\nl)}$ decreases. For $m=6$, $p_{c,m}\suptiny{1}{0}{(\nl k\nl\nl=\nl\nl4\nl)}$ coincides with Ref.\;\cite{BavarescoEtAl2018}'s thresholds $\tilde{p}_{\text{iso},M}\suptiny{1}{0}{(\nl k\nl\nl=\nl\nl4\nl)}$ at $\epsilon_\text{min} \approx 0.00065\;(M=1)$, $0.00040\;(M=2)$, $0.00026\;(M=3)$, $0.00015\;(M=4)$, and precisely 0 $(M=5)$\@. For $m=5$, $p_{c,m}\suptiny{1}{0}{(\nl k\nl\nl=\nl\nl4\nl)}$ only coincides with $\tilde{p}_{\text{iso},M=1}\suptiny{1}{0}{(\nl k\nl\nl=\nl\nl4\nl)}$ at $\epsilon_\text{min} \approx 0.00011$ and not for any $M\geq2$\@. Furthermore, $p_{c,m}\suptiny{1}{0}{(\nl k\nl\nl=\nl\nl4\nl)}$ does not coincide with $\tilde{p}_{\text{iso},M}\suptiny{1}{0}{(\nl k\nl\nl=\nl\nl4\nl)}$ for all $m\leq4$, $1\leq M\leq 5$ and $\epsilon_\text{min}\in[0,0.2]$\@.\\[-3mm]

\begin{figure}[h!]
    \centering\includegraphics[width=\linewidth]{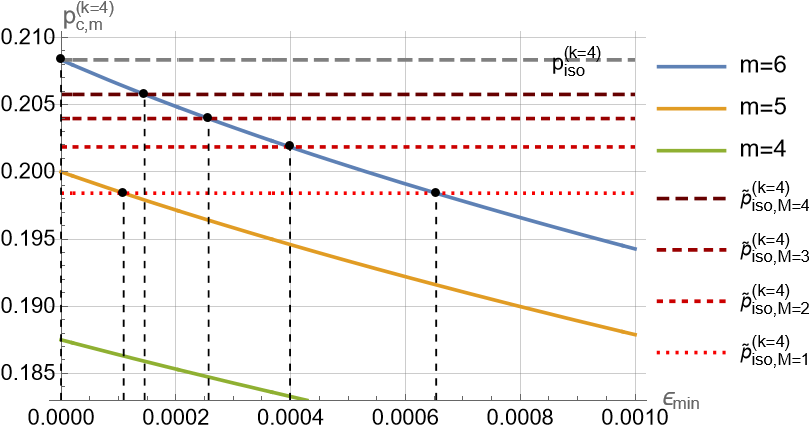}\caption{The maximal white-noise ratio, $p_{c,m}\suptiny{1}{0}{(\nl k\nl\nl=\nl\nl4\nl)}$, below which our witness can certify $\rho\suptiny{1}{0}{\mathrm{iso}}\subtiny{0}{0}{A\nl B}$'s Schmidt number to be $k+1=5$ for different number of measurement bases $m$ are plotted against the bases-bias parameter, $\epsilon_\text{min}\coloneqq 1/d -c_\text{min}$\@. For $m=5$ and 6, they coincide with the white-noise thresholds of Ref.\;\cite{BavarescoEtAl2018}, $\tilde{p}_{\text{iso},M}\suptiny{1}{0}{(\nl k\nl\nl=\nl\nl4\nl)}$, at different values of $\epsilon_\text{min}$ (labelled by black dots), where $M$ is the number of ``tilted" measurement bases used. For $m\leq 4$, $p_{c,m}\suptiny{1}{0}{(\nl k\nl\nl=\nl\nl4\nl)}$ are strictly smaller than $\tilde{p}_{\text{iso},M}\suptiny{1}{0}{(\nl k\nl\nl=\nl\nl4\nl)}$ for all $1\leq M\leq d$\@. The white-noise threshold $\tilde{p}_{\text{iso},M=5}\suptiny{1}{0}{(\nl k\nl\nl=\nl\nl4\nl)}$, which is not shown here, coincides with the white-noise ratio $p\suptiny{1}{0}{(\nl k\nl\nl=\nl\nl4\nl)}_\text{iso}$ for $\rho\suptiny{1}{0}{\mathrm{iso}}\subtiny{0}{0}{A\nl B}$ to have the maximum Schmidt number ($d=5$).}\label{fig:compareJessica_iso}
\end{figure}

We also compare our entanglement-fidelity lower bounds of the isotropic states $\rho\suptiny{1}{0}{\mathrm{iso}}$ with the bounds from Ref.\;\cite{BavarescoEtAl2018}. In Fig.\;\ref{fig:compareJessica_Fidelity_iso} we plot the square root of the relative differences between the actual entanglement fidelity of $\rho\suptiny{1}{0}{\mathrm{iso}}\subtiny{0}{0}{A\nl B}$ and our fidelity lower bounds, $\sqrt{\Delta\overline{\mathcal{F}}\suptiny{1}{-1.5}{\mathrm{iso}}_m} \coloneqq \sqrt{1 - \overline{\mathcal{F}}_m/\mathcal{F}(\rho\subtiny{0}{0}{A\nl B}\suptiny{1}{0}{\mathrm{iso}})}$ \footnote{Note that we plot $\sqrt{\Delta\overline{\mathcal{F}}\suptiny{1}{-1.5}{\mathrm{iso}}_m}$ instead of $\Delta\overline{\mathcal{F}}\suptiny{1}{-1.5}{\mathrm{iso}}_m$ for clearer presentation because the curves appear too close to one another when we plot $\Delta\overline{\mathcal{F}}\suptiny{1}{-1.5}{\mathrm{iso}}_m$ versus $p$\@.}, associated to measuring in $m$ orthonormal bases with different minimum overlaps $c_\text{min}$, and the equivalent square-root relative fidelity differences from Ref.\;\cite{BavarescoEtAl2018}, $\sqrt{\Delta\tilde{\mathcal{F}}\suptiny{1}{0}{(M)}} \coloneqq \sqrt{1 - \tilde{\mathcal{F}}\suptiny{1}{0}{(M)}/\mathcal{F}(\rho\subtiny{0}{0}{A\nl B}\suptiny{1}{0}{\mathrm{iso}})}$, for measuring in $M$ ``tilted" bases (which are MUBs in this example), against different white-noise ratios $p$\@.\\[-3mm]

Although Figs.\;\ref{fig:compareJessica_iso} and~\ref{fig:compareJessica_Fidelity_iso} suggest that our Schmidt-number witness and entanglement fidelity lower bound require more measurement bases than those in Ref.\;\cite{BavarescoEtAl2018} to achieve comparable performance, one should keep in mind that Ref.\;\cite{BavarescoEtAl2018} requires the measurement bases to be exact MUBs in this example or, in general, to be the ``tilted" bases with strict specifications of the relative phases. On the contrary, our method does not require the measurement bases to satisfy any relative phase relationship. In particular, we can still witness Schmidt number and lower bound the entanglement fidelity when measurements in MUBs are not experimentally feasible.\\[-3mm]

\begin{figure}[ht!]
    \centering
    \includegraphics[width=\linewidth]{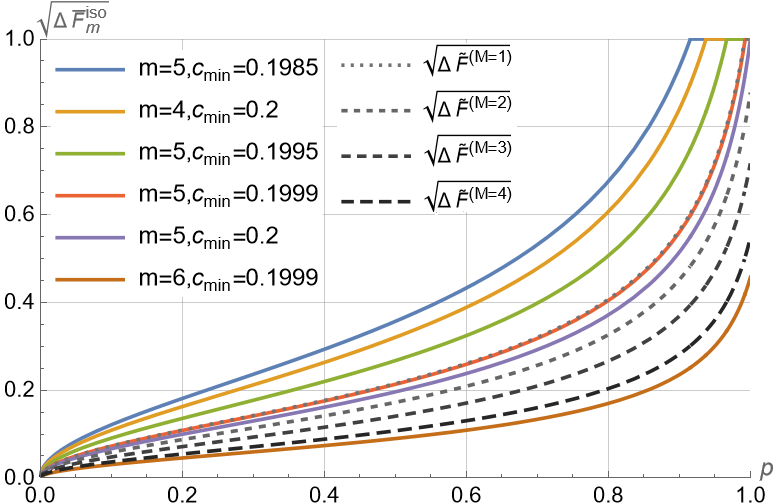}
    \caption{The square-root relative differences between the actual entanglement fidelity of $\rho\suptiny{1}{0}{\mathrm{iso}}\subtiny{0}{0}{A\nl B}$ and our fidelity lower bounds, $\sqrt{\Delta\overline{\mathcal{F}}\suptiny{1}{-1.5}{\mathrm{iso}}_m} \coloneqq \sqrt{1 - \overline{\mathcal{F}}_m/\mathcal{F}(\rho\subtiny{0}{0}{A\nl B}\suptiny{1}{0}{\mathrm{iso}})}$, associated to measuring in $m$ bases with different minimum overlaps $c_\text{min}$ are plotted against different white-noise ratios $p$\@. They are compared against the square-root relative differences between the actual entanglement fidelity and the fidelity lower bounds in Ref.\;\cite{BavarescoEtAl2018}, $\sqrt{\Delta\tilde{\mathcal{F}}\suptiny{1}{0}{(M)}} \coloneqq \sqrt{1 - \tilde{\mathcal{F}}\suptiny{1}{0}{(M)}/\mathcal{F}(\rho\subtiny{0}{0}{A\nl B}\suptiny{1}{0}{\mathrm{iso}})}$, for measuring in $M$ ``tilted" bases (which are coincidentally MUBs in this example). Both $\Delta\overline{\mathcal{F}}\suptiny{1}{0}{\mathrm{iso}}_{d+1}$ at $c_\text{min}=1/d$ and $\Delta\tilde{\mathcal{F}}\suptiny{1}{0}{(M=d)}$ for $\rho\suptiny{1}{0}{\mathrm{iso}}\subtiny{0}{0}{A\nl B}$ are zero, and they are omitted from the plot.}\label{fig:compareJessica_Fidelity_iso}
\end{figure}

Furthermore, Figs.\;\ref{fig:compareJessica_iso} and~\ref{fig:compareJessica_Fidelity_iso} may suggest that our method only tolerates very small bases bias. However, one should recall that we picked the worst-case bases choice for any given $c_\text{min}$~\cite{footnote:Eg1WorstCase}. In practice, measuring in any other bases associated with a fixed $c_\text{min}$ will give better noise tolerance (i.e., witnessing higher Schmidt number or providing higher fidelity lower bound).

\vspace*{-1.5mm}
\subsubsection{Purified thermal states mixed with white noise}\label{app:thermal}
\vspace*{-1.5mm}

In this example, we want to investigate (1) how the evenness of the eigenvalues of the one-party reduced density matrix $\rho\subtiny{0}{0}{A}$ affects the performance of our witness when applied to $\rho\subtiny{0}{0}{A\nl B}$, and (2) whether measuring in an additional basis that is not unbiased with respect to the other measurement bases which are MUBs improve the performance of our witness.\\[-3mm]
    
We consider the purified thermal states mixed with white noise parametrized by $p\in[0,1]$ and $\beta\in[0,\infty)$:
\begin{subequations}
\begin{flalign}
    &\rho_\text{th}(p,\beta) = (1-p)\ket{\psi_\text{th}(\beta)}\!\!\bra{\psi_\text{th}(\beta)} + \frac{p}{d^2}\identity_{d^2},\\
    &\ket{\psi_\text{th}(\beta)} = \frac{1}{\sqrt{\mathcal{Z}}}\sum_{n=0}^{d-1} e^{-\frac{\beta n}{2}}\ket{n}\otimes \ket{n},
\end{flalign}
\end{subequations}
where $\mathcal{Z}=\sum_{n=0}^{d-1} e^{-\beta n}$\@. The reduced density matrix is
\begin{equation}
    \rho_{\text{th},A}(p,\beta) = \sum_{n=0}^{d-1} \left[(1-p)\frac{e^{-\beta n}}{\mathcal{Z}} + \frac{p}{d}\right] \ket{n}\!\!\bra{n}\,, \label{eq:reducedDensity_th}
\end{equation}
which has uneven eigenvalues for $\beta>0$\@. These states are interesting since they are known to be not optimal for the Schmidt-number witness in Ref.\;\cite{BavarescoEtAl2018}. Also, by changing just one parameter $\beta$, $\ket{\psi_\text{th}(\beta)}$ goes from being the maximally entangled state to an entangled state with $d$ positive uneven Schmidt coefficients, and finally to the separable state $\ket{0}\otimes\ket{0}$ as $\beta$ goes from 0 to $\infty$\@.\\[-3mm]

For simplicity, we consider only odd-prime dimensions $d$ in this example. Suppose that there are three local measurement bases available in our experiment: $\{\ket{e^1_i}=\ket{i}\}_{i=0}^{d-1}$, $\{\ket{e^2_j}=\frac{1}{\sqrt{d}}\sum_{n=0}^{d-1} \omega^{jn} \ket{n}\}_{j=0}^{d-1}$ and $\{\ket{e^3_j}\}_{j=0}^{d-1}$ with
\begin{equation}
    \ket{e^3_j} = Z_\alpha(\theta)\sum_{k=0}^{d-1} \frac{\omega^{jk+k^2}}{\sqrt{d}} \ket{k}
\end{equation}
where $\omega=e^{i\frac{2\pi}{d}}$, $Z_\alpha(\theta)=\sum_{k=0}^{d-1} e^{i\delta_{\alpha,k}\theta}\ket{k}\!\!\bra{k}$ with $\alpha\in[d]\coloneqq\{0,\ldots,d-1\}$ and $\theta\in[0,2\pi)$, and $\delta_{\alpha,k}$ is the Kronecker delta. The first two bases and bases 1 and 3 are both pairwise mutually unbiased, but the three bases together do not form a set of MUBs except when $\theta=0$~\cite{Ivonovic1981}. Therefore, one can interpret them as an imperfect implementation of three MUBs with the third basis subject to a phase drift $Z_\alpha(\theta)$ relative to the second basis.\\[-3mm]

Let us first evaluate the expectation value of our witness operator [see Eq.\;\eqref{eq:witness}]
\begin{equation}
    \mathcal{S}_{d}\suptiny{1}{0}{(m)}(\rho_\text{th}) = (1-p)\sum_{z=1}^m\sum_{a=0}^{d-1} |\langle e^z_a,\tilde{e}^z_a{}^*|\psi_\text{th}(\beta)\rangle|^2 + \frac{mp}{d}. \label{eq:eg2_Smd1}
\end{equation}
For each matching pair of local measurement bases, we have
\begin{subequations}
\begin{flalign}
    \sum_{a=0}^{d-1}|\langle e^1_a,\tilde{e}^1_a{}^*|\psi_\text{th}(\beta)\rangle|^2 &= 1,\\
    \sum_{a=0}^{d-1}|\langle e^2_a,\tilde{e}^2_a{}^*|\psi_\text{th}(\beta)\rangle|^2 
    &= \sum_{a=0}^{d-1}|\langle e^3_a,\tilde{e}^3_a{}^*|\psi_\text{th}(\beta)\rangle|^2\label{eq:eg2_Smd4}\\
    &= \begin{cases}
        1 & \text{if}\ \ \beta=0,\\
        \frac{1}{d\mathcal{Z}}\left(\frac{1-e^{-d\beta/2}}{1-e^{-\beta/2}}\right)^2 & \text{if}\ \ \beta>0.
    \end{cases} \nonumber
\end{flalign}
\end{subequations}
If we combine Eqs.\;\eqref{eq:eg2_Smd1}--\eqref{eq:eg2_Smd4} and simplify the expression with $\mathcal{Z}=\sum_{n=0}^{d-1} e^{-\beta n} = \frac{1-e^{-d\beta}}{1-e^{-\beta}}$ for all $\beta>0$, $\tanh(x)=\frac{1-e^{-2x}}{1+e^{-2x}}$, and $\frac{1-e^{-x}}{1-e^{-x/2}} = 1+e^{-\frac{x}{2}}$, we get
\begin{subequations}
\begin{flalign}
    &\mathcal{S}_{d}\suptiny{1}{0}{(m)}(\rho_\text{th}) = (1-p)\tau\suptiny{1}{0}{(m)}_d(\beta) + \frac{mp}{d},
    \\
    &\tau\suptiny{1}{0}{(m)}_d(\beta) \coloneqq \begin{cases}
        m &\text{if}\ \ \beta=0,\\
        1+\frac{(m-1)\tanh\left(d\beta/4\right)}{d\tanh\left(\beta/4\right)} &\text{if}\ \ \beta>0,
    \end{cases}\label{eq:Smd_th}
\end{flalign}
\end{subequations}
depending on whether we use only the first two bases ($m=2$) or all three bases ($m=3$).\\[-3mm]

Our next step is to determine the upper bound of $\mathcal{S}_{d}\suptiny{1}{0}{(m)}$ for any bipartite state with Schmidt number at most~$k$\@. This requires us to calculate all the bases overlaps:
\begin{subequations}
\begin{flalign}
    &|\langle e^1_a|e^z_{a'}\rangle|^2 = \frac{1}{d}\;(=c_\text{min}^{1,z}= c_\text{max}^{1,z}),\quad\forall\;z\in\{2,3\},\label{eq:c^1z_thermal}\\
    &\langle e^2_a|e^3_{a'}\rangle = \frac{1}{d}\left[\omega^{(a'-a)\alpha+\alpha^2}(e^{i\theta}-1) + \sum_{k=0}^{d-1} \omega^{(a'-a)k+k^2}\right],\\
    &|\langle e^2_a|e^3_{a'}\rangle|^2 = \frac{1}{d^2}\left\{\left|\sum_{k=0}^{d-1} \omega^{(a'-a)k+k^2}\right|^2 + |e^{i\theta}-1|^2\right.\\
    &\qquad\qquad +\left. 2\text{Re}\left[\omega^{(a'-a)\alpha+\alpha^2}(e^{i\theta}-1)\sum_{k=0}^{d-1} \omega^{(a'-a)k+k^2}\right]\right\}\nonumber
\end{flalign}
\end{subequations}
for all $a,a'\in[d]$\@. With $-|y|\leq \text{Re}(y)\leq |y|$ for all $y\in\mathbbm{C}$ and a fact regarding quadratic Gauss sums~\cite{Ivonovic1981}, i.e.,
\begin{equation}
    \left|\sum_{k=0}^{d-1} \omega^{ak+k^2}\right| = \sqrt{d} \quad\forall\;a\in\mathbbm{Z} \text{ \& odd prime }d,
\end{equation}
we obtain the upper and lower bounds for $|\langle e^2_a|e^3_{a'}\rangle|^2$ as
\begin{flalign}
    c_\pm^{2,3} = \frac{1}{d^2}\left(\sqrt{d}\pm 2\left|\sin(\theta/2)\right|\right)^2,
    \label{eq:cmin_max_th}
\end{flalign}
such that $c_\text{min}^{2,3} = c_-^{2,3} \leq |\langle e^2_a|e^3_{a'}\rangle|^2 \leq c_+^{2,3} = c_\text{max}^{2,3}$\@. Then, applying the formula for $\overline{\mathcal{B}}_k$ in Lemma~\ref{lemma:SchmidtRkLoose}, we obtain for $m=2$ or 3, $\overline{\mathcal{T}}(\overline{\mathcal{C}}) \coloneqq \min\{\overline{\lambda}(\overline{\mathcal{C}}), m\}$ with
\begin{equation}
    \overline{\lambda}(\overline{\mathcal{C}}) = \frac{1}{2}\left(1+\sqrt{1+4d\delta_{m,3}\overline{G}\supscr{1.5}{-2}{2,3}}\right),
\end{equation}
with the shorthand notation $\overline{G}\supscr{1.5}{-2}{z,z'} \coloneqq \overline{G}(c^{z,z'}_\text{max},c^{z,z'}_\text{min})$ and the fact that $\overline{G}\supscr{1.5}{-2}{1,z} = \overline{G}\supscr{1.5}{-2}{z,1} = 0$ and $\overline{G}\supscr{1.5}{-2}{2,3} = \overline{G}\supscr{1.5}{-2}{3,2} = 1-(d+1)c_\text{min}^{2,3} + \Omega^{2,3}$ as defined in Lemma~\ref{lemma:SchmidtRkLoose}.\\[-3mm]

In order to witness $\rho_\text{th}(p,\beta)$'s Schmidt number is at least $k+1$, the relationship $\mathcal{S}_{d}\suptiny{1}{0}{(m)}(\rho_\text{th})>\overline{\mathcal{B}}_k$ must be satisfied, which after some manipulation gives
\begin{equation}
    p< \frac{\tau\suptiny{1}{0}{(m)}_d(\beta)d - km - (d-k)\overline{\mathcal{T}}(\overline{\mathcal{C}})}{\tau\suptiny{1}{0}{(m)}_d(\beta)d - m}\eqcolon p\suptiny{1}{0}{(k)}_{\text{th},m}(\beta,\theta)
    \label{eq:p^k_th,m}
\end{equation}
if $\tau\suptiny{1}{0}{(m)}_d(\beta)>\frac{m}{d}$ which holds for all $\beta\in[0,\infty)$ according to the definition in Eq.\;\eqref{eq:Smd_th}. We plot $p\suptiny{1}{0}{(k)}_{\text{th},m}(\beta,\theta)$, the threshold of white-noise ratio for detecting $\rho_\text{th}(p,\beta)$'s Schmidt number larger than~$k$ against $\theta$ with $\beta=0.5$ and against $\beta$ with $\theta=0.05$ for $d=5$, $m=2,3$ and $k=1$ to 4 in Figs.\;\ref{fig:p_Bound_thermal_thetaVar} and~\ref{fig:p_Bound_thermal_betaVar}, respectively. Both figures suggest that measuring in an additional basis that is close to forming a set of MUBs with the already measured bases can improve the noise tolerance of our witness. On the other hand, an additional measurement basis that is far from being mutually unbiased with the other bases could adversely affect the performance of our witness due to a significant increase in the upper bound $\mathcal{B}_k$ ($\overline{\mathcal{B}}_k$) in Theorem~\ref{thm:SchmidtRkBound} (Lemma~\ref{lemma:SchmidtRkLoose}). We summarize this observation in the following remark. 

\begin{repremark}{remark1}
    There exist scenarios where an additional measurement basis improves the noise tolerance of our Schmidt-number witness. However, the opposite case can also happen depending on the choice of the bases. Therefore, in order to witness the highest Schmidt number of a state, one should apply the witness inequality in Theorem~\ref{thm:SchmidtRkBound} or Lemma~\ref{lemma:SchmidtRkLoose} to all subsets of the total set of $m'$ available measurement bases and find the largest~$k$ such that $\mathcal{S}_{d}\suptiny{1}{0}{(m)}(\rho)>\mathcal{B}_k$ or $\overline{\mathcal{B}}_k$ evaluated over all subsets of $m$ chosen bases for all $m\in\{2,\ldots,m'\}$\@.
\end{repremark}

\begin{figure}[ht!]
    \centering
    \includegraphics[width=\linewidth]{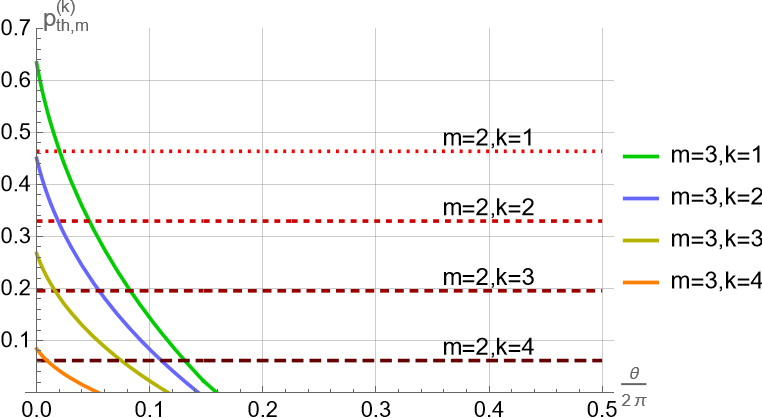}
    \caption{The white-noise threshold for detecting $\rho_\text{th}(p,\beta)$'s Schmidt number to be larger than~$k$, $p\suptiny{1}{0}{(k)}_{\text{th},m}(\beta,\theta)$, plotted against the third basis' phase-drift parameter $\theta$ for $\beta=0.5$, $d=5,\nr m=2,3$ and $k=1$ to 4. Since $p\suptiny{1}{0}{(k)}_{\text{th},m}(\beta,\theta)=p\suptiny{1}{0}{(k)}_{\text{th},m}(\beta,2\pi-\theta)$~\cite{footnote:p_th_symm_theta}, we omit the plot for $\pi\leq\theta\leq 2\pi$\@. The threshold for measuring in all three bases $p\suptiny{1}{0}{(k)}_{\text{th},m=3}(\beta,\theta)$ decreases as $\theta$ increases and eventually gets below $p\suptiny{1}{0}{(k)}_{\text{th},m=2}(\beta,\theta)$, the threshold for measuring only in the first two bases for all~$k$\@.}\label{fig:p_Bound_thermal_thetaVar}
\end{figure}

We can also see how the additional measurement basis and the evenness of the reduced density matrix's eigenvalues affect our lower bound of the entanglement fidelity. The entanglement fidelity of $\rho_\text{th}(p,\beta)$ can be found by noticing that the absolute value of any matrix element of a unitary is upper bounded by 1 \footnote{Since any unitary $U$ has a spectral decomposition $U=\sum_i \lambda_i|\lambda_i\rangle\!\langle\lambda_i|$, where all eigenvalues satisfy $|\lambda_i|=1$ and $\{|\lambda_i\rangle\}_i$ is an orthonormal basis, any matrix element of $U$ satisfies $|U_{jk}| = |\sum_i \lambda_i\langle j|\lambda_i\rangle\!\langle\lambda_i|k\rangle| \leq \sqrt{\sum_i |\lambda_i||\langle j|\lambda_i\rangle|^2}\sqrt{\sum_i |\lambda_i||\langle k|\lambda_i\rangle|^2} = 1$, where we use the Cauchy{\textendash}Schwarz inequality.}, which implies
\begin{flalign}
    \max_U|\langle\psi_\text{th}|\identity\otimes U|\Phi^+\rangle| &= \max_U\frac{1}{\sqrt{d\mathcal{Z}}}\left|\sum_{n=0}^{d-1} e^{-\frac{\beta n}{2}} \langle n|U|n\rangle\right|\nonumber\\
    &= \frac{1}{\sqrt{d\mathcal{Z}}}\sum_{n=0}^{d-1} e^{-\frac{\beta n}{2}}.
\end{flalign}
Then, applying the same methods that resulted in Eq.\;\eqref{eq:Smd_th} gives us $\rho_\text{th}(p,\beta)$'s entanglement fidelity,
\begin{equation}
    \mathcal{F}(\rho_\text{th}(p,\beta)) = 
    \begin{cases}
        1-p + \frac{p}{d^2} &\text{if}\ \ \beta=0,\\
        \frac{(1-p)\tanh\left(d\beta/4\right)}{d\tanh\left(\beta/4\right)} + \frac{p}{d^2} 
        &\text{if}\ \ \beta>0.
    \end{cases}
\end{equation}
We compare\nl
{\begin{figure}[ht!]
    \centering
    \includegraphics[width=0.9\linewidth]{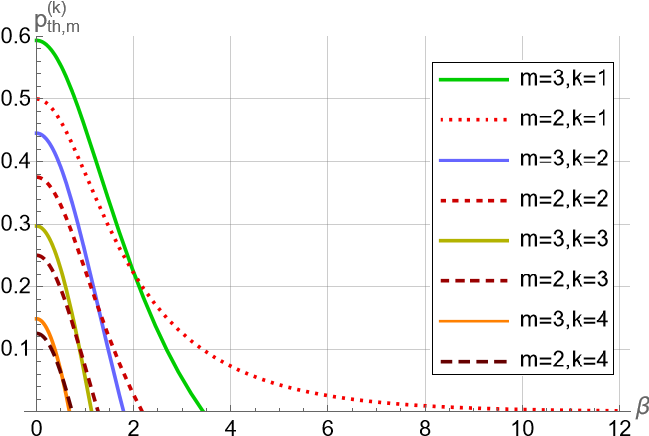}
    \caption{The white-noise threshold for detecting $\rho_\text{th}(p,\beta)$'s Schmidt number being larger than~$k$, $p\suptiny{1}{0}{(k)}_{\text{th},m}(\beta,\theta)$, plotted against $\beta$ for $\theta=0.05$, $d=5$, $m=2,3$ and $k=1$ to 4. All thresholds decrease as $\beta$ increases, which corresponds to increasing skewness of  eigenvalues of the reduced density matrix in Eq.\;\eqref{eq:reducedDensity_th}. Due to the bias of the third basis, $p\suptiny{1}{0}{(k)}_{\text{th},m=3}(\beta,0.05)-p\suptiny{1}{0}{(k)}_{\text{th},m=2}(\beta,0.05)$ goes from being strictly positive to negative as $\beta$ increases from 0.}\label{fig:p_Bound_thermal_betaVar}
\end{figure}}
our fidelity lower bound with the true value in Fig.\;\ref{fig:FidelityDiff_thermal}, which shows the relative difference $\Delta\overline{\mathcal{F}}^\text{th}_m \coloneqq 1 - \overline{\mathcal{F}}_m/\mathcal{F}(\rho_\text{th}(p,\beta))$ for different values of $p, \beta$ and $\theta$\@.\\[-3mm]

Finally, we compare the performance of our Schmidt-number witness with the one in Ref.\;\cite{BavarescoEtAl2018}. 
In Fig.\;\ref{fig:compareJessica_SN_th}, we plot the white-noise thresholds for detecting $\rho_\text{th}(p,\beta)$'s Schmidt number larger than~$k$, $p\suptiny{1}{0}{(k)}_{\text{th},m}(\beta,\theta)$, against the third basis' phase-drift parameter $\theta$ for $\beta=0.05$, $d=5$, and $k=1$ to 4. These are compared with the white-noise thresholds for measuring in the computational basis together with $M=2$ ``tilted" bases of Ref.\;\cite{BavarescoEtAl2018}, $\tilde{p}\suptiny{1}{0}{(k)}_{\text{th},M=2}$, which are numerically computed as described in Appendix~\ref{app:JessicaWitnessThermal}. Similar to the example of isotropic states, Fig.\;\ref{fig:compareJessica_SN_th} suggests that the witness of Ref.\;\cite{BavarescoEtAl2018} is more noise-tolerant than our method in certifying Schmidt numbers of $\rho_\text{th}(p,\beta)$\@. However, we should not overlook that the ``tilted" measurement bases are not orthogonal and require precise control in the relative phases among the bases vectors. Therefore, the choice of bases in Ref.\;\cite{BavarescoEtAl2018} may not be applicable in general experimental settings. For instance, the natural measurement bases for certifying entanglement in certain photonic  systems are the temporal and frequency bases \cite{ChangSarihanChengErkerLiEtAl2024}, or the position and momentum bases in cold atoms \cite{EulerGaerttner2023}, both of which are inherently orthogonal. It remains unclear how practical it is to realize non-orthogonal measurement bases or bases with specifically chosen relative phases in these settings. 
Note that full derivations of the noise tolerance of Ref.\;\cite{BavarescoEtAl2018}'s witness can be found in Appendix~\ref{app:JessicaWitnessThermal}.

\begin{figure}[t!]
    \centering
    \includegraphics[width=\linewidth]{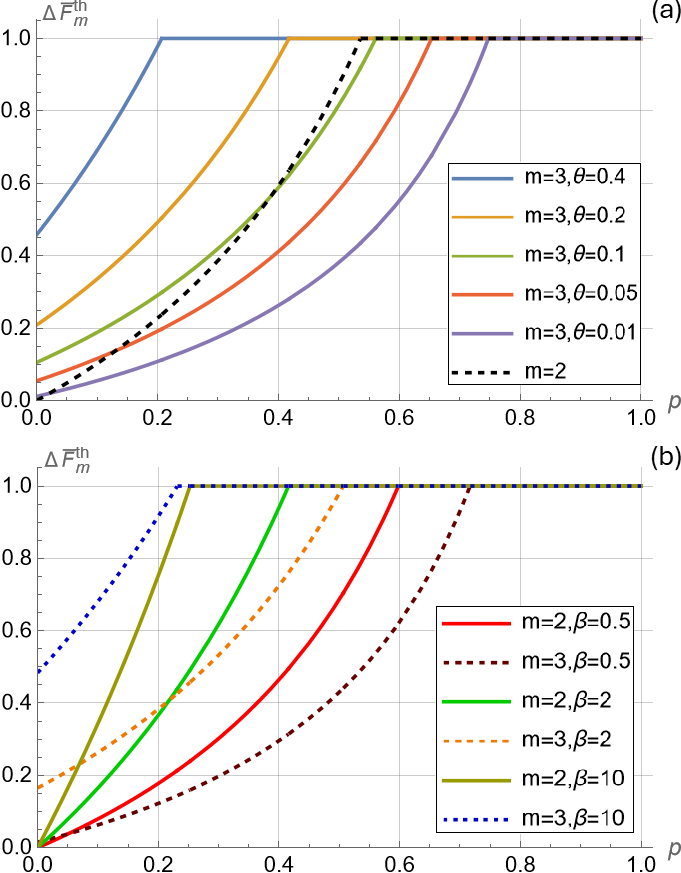}
    \caption{The relative differences between the actual entanglement fidelity of $\rho_\text{th}$ and our fidelity lower bounds, $\Delta\overline{\mathcal{F}}^{\text{th}}_m \coloneqq 1 - \overline{\mathcal{F}}_m/\mathcal{F}(\rho_\text{th}(p,\beta))$, associated to measuring in $m=2,3$ bases are plotted against different white-noise ratios $p$, where only the third basis is parameterized by $\theta$\@. In (a), $\beta$ is fixed to be 1 and we plot $\Delta\overline{\mathcal{F}}^{\text{th}}_m$ for different values of $\theta$\@. When the third measurement basis is used, $\Delta\overline{\mathcal{F}}^{\text{th}}_m$ is lower than that of using 2 MUBs for small $\theta$ and it grows beyond the relative fidelity difference of 2 MUBs as $\theta$ increases. In (b), $\theta$ is fixed to be 0.05 and we plot $\Delta\overline{\mathcal{F}}^{\text{th}}_m$ for different values of $\beta$\@. For small $p$, using 3 bases is worse than using only 2 bases for the given $\theta$\@. For mid-range values of $p$, the bounds using 3 bases are tighter for $\beta=0.5$ and 2 but not for 10. In general, the relative fidelity difference appears to grow with $\beta$\@. The cusps in both plots occur when $\overline{\mathcal{F}}_m$ in Eq.\;\eqref{ineq:FboundLoose} transits from $\overline{\mathcal{T}}(\overline{\mathcal{C}})< \mathcal{S}_d\suptiny{1}{0}{(m)}(\rho_\text{th})$ to $\overline{\mathcal{T}}(\overline{\mathcal{C}})\geq \mathcal{S}_d\suptiny{1}{0}{(m)}(\rho_\text{th})$, where the fidelity lower bound becomes trivial (i.e., $\overline{\mathcal{F}}_m= 0$), as $p$ increases.}\label{fig:FidelityDiff_thermal}
\end{figure}

\vspace*{-1.5mm}

\subsection{Comparison with the Schmidt-number witness in Ref.\;\cite{BavarescoEtAl2018}}\label{app:JessicaWitness}\vspace*{-1.5mm}

In this appendix, we first summarize the main idea behind the Schmidt-number witness proposed in Ref.\;\cite{BavarescoEtAl2018}. After defining the witness, we will apply it to isotropic states and noisy purified thermal states and calculate the white-noise tolerances in both cases, which are plotted in Figs.\;\ref{fig:compareJessica_iso} and~\ref{fig:compareJessica_SN_th}, respectively.\\[-3mm]

\begin{figure}[t!]
    \centering
    \includegraphics[width=\linewidth]{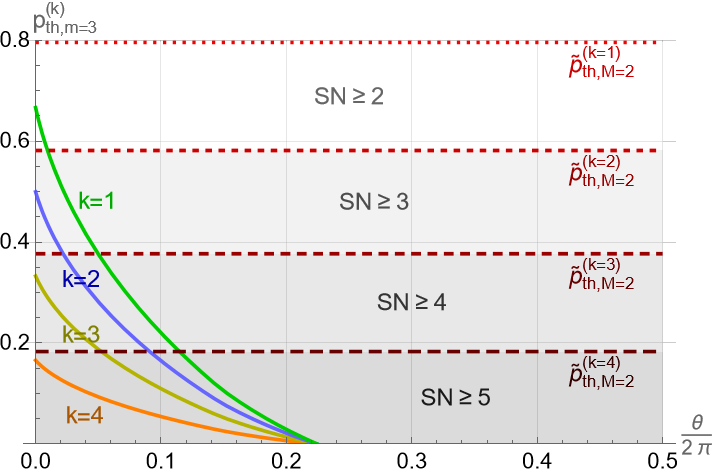}
    \caption{The white-noise thresholds for detecting $\rho_\text{th}(p,\beta)$'s Schmidt number larger than~$k$, $p\suptiny{1}{0}{(k)}_{\text{th},m=3}(\beta,\theta)$, plotted against the third basis' phase-drift parameter $\theta$ for $\beta=0.05$, $d=5$, and $k=1$ to 4. Since $p\suptiny{1}{0}{(k)}_{\text{th},m=3}(\beta,\theta)=p\suptiny{1}{0}{(k)}_{\text{th},m=3}(\beta,2\pi-\theta)$~\cite{footnote:p_th_symm_theta}, we omit the plot for $\pi\leq\theta\leq 2\pi$\@. These are compared with the white-noise thresholds for measuring in the computational basis together with $M=2$ ``tilted" bases of Ref.\;\cite{BavarescoEtAl2018}, $\tilde{p}\suptiny{1}{0}{(k)}_{\text{th},M=2}$\@. Below the dotted lines labelled by $\tilde{p}\suptiny{1}{0}{(k)}_{\text{th},M=2}$ are regions of white-noise ratios, in which $\rho_\text{th}(p,\beta)$'s Schmidt number is certified to be at least $k+1$ by the witness of Ref.\;\cite{BavarescoEtAl2018}.}\label{fig:compareJessica_SN_th}
\end{figure}

\subsubsection{Summary of the witness in Ref.\;\cite{BavarescoEtAl2018}}\label{app:JessicaWitnessSummary}
For any bipartite state $\rho\subtiny{0}{0}{A\nl B}\suptiny{1}{0}{(\leq\nl\nl k\nl)}$ of Schmidt number at most~$k$, $\mathcal{F}(\rho\subtiny{0}{0}{A\nl B}\suptiny{1}{0}{(\leq\nl\nl k\nl)},\ket{\psi})$, the fidelity between $\rho\subtiny{0}{0}{A\nl B}\suptiny{1}{0}{(\leq\nl\nl k\nl)}$ and a pure state $\ket{\psi}=\sum_{i=0}^{d-1} \lambda_i\ket{e_i}\ket{f_i}$ with known Schmidt coefficients $\{\lambda_i\}_{i=0}^{d-1}$ such that $\lambda_i\geq\lambda_j\;\forall\;i<j$ and $\sum_{i=0}^{d-1} \lambda_i^2 = 1$, where $\{\ket{e_i}\}_{i}, \{\ket{f_i}\}_{i}$ are orthonormal bases, must satisfy
\begin{equation}
    \mathcal{F}(\rho\subtiny{0}{0}{A\nl B}\suptiny{1}{0}{(\leq\nl\nl k\nl)},\ket{\psi}) \leq \widetilde{\mathcal{B}}_k(\ket{\psi}) \coloneqq \sum_{i=0}^{k-1} \lambda_i^2.
    \label{ineq:JessicaBk}
\end{equation}
Hence, if $\mathcal{F}(\rho\subtiny{0}{0}{A\nl B},\ket{\psi}) > \widetilde{\mathcal{B}}_k(\ket{\psi})$, then we certify the Schmidt number of $\rho\subtiny{0}{0}{A\nl B}$ must be at least $k+1$\@.\\[-3mm]

Using as few as two measurement bases, it is possible to obtain a good estimate of $\mathcal{F}(\rho\subtiny{0}{0}{A\nl B},\ket{\psi})$ for any target state $\ket{\psi}$ in any odd-prime local dimension if the measurement bases are constructed in the way specified by Ref.\;\cite{BavarescoEtAl2018}. For that to work, the target state $\ket{\psi}=\sum_{i=0}^{d-1} \lambda_i\ket{e_i}\ket{f_i}$ is chosen such that
\begin{equation}
    \lambda_i = \sqrt{\frac{\langle e_i f_i|\nr\rho\subtiny{0}{0}{A\nl B}\nr|e_i f_i\rangle}{\sum_{j=0}^{d-1} \langle e_j f_j|\nr\rho\subtiny{0}{0}{A\nl B}\nr|e_j f_j\rangle}} \label{eq:JessicaTargetLambda}
\end{equation}
and $\{\ket{e_i}\}_{i} \;(\{\ket{f_i}\}_{i})$ is the local orthonormal measurement basis of party A (B). From now on, w.l.o.g., we set $\ket{e_i}=\ket{i}$ and $\ket{f_i}=\ket{i}$ for all $i$\@.\\[-3mm]

Since calculating $\mathcal{F}(\rho\subtiny{0}{0}{A\nl B},\ket{\psi})$ (which will be abbreviated as $\mathcal{F}$ from now on) requires some knowledge of the off-diagonal entries of $\rho\subtiny{0}{0}{A\nl B}$ with respect to the local computational basis $\{\ket{i}\}$, both parties must measure in more than one basis. In Ref.\;\cite{BavarescoEtAl2018}, these extra $1\leq M\leq d$ bases are chosen to be the ``tilted" bases $\{\ket{\tilde{j}_\alpha}\}_{j=0}^{d-1}$ with $\alpha\in\{0,\ldots,M-1\}$, which are defined by
\begin{equation}
    \ket{\tilde{j}_\alpha} = \frac{1}{\sqrt{\sum_i \lambda_i}}\sum_{n=0}^{d-1} \omega^{jn+\alpha n^2} \sqrt{\lambda_n}\ket{n},
\end{equation}
where $\omega=e^{i\frac{2\pi}{d}}$. 
Given that party A (B) measures in local bases $\{i\}_{i=0}^{d-1}$ and $\{\ket{\tilde{j}_\alpha}\}_{j=0}^{d-1}$ ($\{\ket{\tilde{j}_\alpha^*}\}_{j=0}^{d-1}$), the two parties can evaluate a lower bound of the fidelity $\mathcal{F}$ given by $\tilde{F}\suptiny{1}{0}{(M)}\coloneqq F_1 + \tilde{F}_2\suptiny{1}{0}{(M)}$ \footnote{Note that for all odd-prime dimensions $d$, $\tilde{F}\suptiny{1}{0}{(M')}\geq \tilde{F}\suptiny{1}{0}{(M)}$ holds for all $M'\geq M$, whereas for all other dimensions, only $\tilde{F}\suptiny{1}{0}{(M)}\geq \tilde{F}\suptiny{1}{0}{(1)}$ is guaranteed to hold $\forall\;M\geq 1$~\cite{BavarescoEtAl2018}.}, where
\begin{subequations}
\begin{flalign}
    &F_1 \coloneqq \sum_{n=0}^{d-1} \lambda_n^2\langle nn|\nr\rho\subtiny{0}{0}{A\nl B}\nr|nn\rangle,\\
    &\tilde{F}_2\suptiny{1}{0}{(M)} \coloneqq \frac{(\sum_n \lambda_n)^2}{d} \Sigma\suptiny{1}{0}{(M)} \,-\! \sum_{m,n=0}^{d-1} \lambda_m \lambda_{n} \langle mn|\nr\rho\subtiny{0}{0}{A\nl B}\nr|mn\rangle \nonumber\\
    &\ \ - \!\!\!\sum_{\substack{m\neq m', m\neq n,\\n\neq n',n'\neq m'}} \!\!\! \tilde{\gamma}\suptiny{1}{0}{(M)}_{mm'nn'}\sqrt{\langle m'n'|\nr\rho\subtiny{0}{0}{A\nl B}\nr|m'n'\rangle\!\langle mn|\nr\rho\subtiny{0}{0}{A\nl B}\nr|mn\rangle}\ ,\\
    &\Sigma\suptiny{1}{0}{(M)} \coloneqq \frac{1}{M}\sum_{\alpha=0}^{M-1}\sum_{j=0}^{d-1} \langle \tilde{j}_\alpha \tilde{j}_\alpha^*|\nr\rho\subtiny{0}{0}{A\nl B}\nr|\tilde{j}_\alpha \tilde{j}_\alpha^*\rangle,
\end{flalign}
\end{subequations}
with the quantity $\tilde{\gamma}\suptiny{1}{0}{(M)}_{mm'nn'}$ given by
\begin{subequations}
\begin{flalign}
    &\tilde{\gamma}\suptiny{1}{0}{(M)}_{mm'nn'} \coloneqq \frac{\tilde{\gamma}_{mm'nn'}}{M}\left|\sum_{\alpha=0}^{M-1}\omega^{\alpha(m^2-m'^2-n^2+n'^2)}\right|,\\
    &\tilde{\gamma}_{mm'nn'} \coloneqq {\begin{cases} 0 &\hspace*{-3mm}\text{if}\ (m-m'-n+n')\text{mod}\;d\neq 0,\\ \sqrt{\lambda_m \lambda_{n}\lambda_{m'} \lambda_{n'}} &\text{otherwise.}\end{cases}}
\end{flalign}
\end{subequations}
Since $\tilde{F}\suptiny{1}{0}{(M)}\leq \mathcal{F}\leq \widetilde{\mathcal{B}}_k(\ket{\psi})$ for all $\rho\subtiny{0}{0}{A\nl B}$ with Schmidt number at most~$k$, if $\tilde{F}\suptiny{1}{0}{(M)}> \widetilde{\mathcal{B}}_{k'}(\ket{\psi})$, then $\rho\subtiny{0}{0}{A\nl B}$ must have Schmidt number at least $k'+1$\@.

\subsubsection{Witnessing isotropic states}\label{app:JessicaWitnessIsotropic}
In this appendix, we will derive $\tilde{F}\suptiny{1}{0}{(M)}$ for the isotropic state $\rho\suptiny{1}{0}{\mathrm{iso}} = (1-p)\ket{\Phi^+_{d}}\!\!\bra{\Phi^+_{d}}+\frac{p}{d^2}\identity_{d^2}$\@. Since
\begin{equation}
    \langle ij|\rho\suptiny{1}{0}{\mathrm{iso}}|ij\rangle =\frac{1-p}{d}\delta_{ij}+\frac{p}{d^2}, \quad\forall\;i,j\in[d],
\end{equation}
the target state $\ket{\psi}$ has $\lambda_i=\frac{1}{\sqrt{d}}\;\forall\;i$ (i.e., $\ket{\psi} = \ket{\Phi^+_d}$) according to Eq.\;\eqref{eq:JessicaTargetLambda}, and the ``tilted" bases together with the computational basis are MUBs in prime-power dimensions~\cite{WoottersFields1989}. Furthermore, we have
\begin{subequations}
\begin{flalign}
    F_1 &=\, \frac{1}{d}\left(1-p+\frac{p}{d}\right),\\[1mm]
    \Sigma\suptiny{1}{0}{(M)} &=\, \frac{1}{M}\sum_{\alpha=0}^{M-1}\left(1-p+\frac{p}{d}\right)=1-p+\frac{p}{d},
\end{flalign}
\end{subequations}
where we use the fact that $U_\alpha\otimes U_\alpha^*\ket{\Phi^+_{d}}=\ket{\Phi^+_{d}}$ with $U_\alpha = \sum_j\ket{\tilde{j}_\alpha}\!\!\bra{j}$, and for the last two terms of $\tilde{F}_2\suptiny{1}{0}{(M)}$ we have
\begin{equation}
    \sum_{m,n=0}^{d-1} \lambda_m \lambda_{n} \langle mn|\rho\suptiny{1}{0}{\mathrm{iso}}|mn\rangle = \frac{1}{d}\tr(\rho\suptiny{1}{0}{\mathrm{iso}}) = \frac{1}{d},
\end{equation}
as well as
\begin{equation}
    \tilde{\gamma}\suptiny{1}{0}{(M)}_{mm'nn'} = \begin{cases} 0\ \ \ \ \ \ \ \ \ \text{ if}\ \  (m-m'-n+n')\text{ mod}\;d\neq 0,\\ \frac{1}{Md}\left|\sum_{\alpha=0}^{M-1}\omega^{\alpha(m^2-m'^2-n^2+n'^2)}\right| \text{ otherwise,}\end{cases}
\end{equation}
such that
\begin{flalign}
    &\sum_{\substack{m\neq m', m\neq n,\\n\neq n',n'\neq m'}}\!\!\! \tilde{\gamma}\suptiny{1}{0}{(M)}_{mm'nn'}\sqrt{\langle m'n'|\rho\suptiny{1}{0}{\mathrm{iso}}|m'n'\rangle\!\langle mn|\rho\suptiny{1}{0}{\mathrm{iso}}|mn\rangle} \nonumber\\
    = & \sum_{\substack{m\neq m', m\neq n,\\n\neq n',n'\neq m'}}\!\!\! \tilde{\gamma}\suptiny{1}{0}{(M)}_{mm'nn'}\sqrt{\left(\frac{1-p}{d}\delta_{m'n'}+\frac{p}{d^2}\right)\left(\frac{1-p}{d}\delta_{mn}+\frac{p}{d^2}\right)}\nonumber\\
    =\, & \frac{p}{Md^3}\!\!\!\!\!\!\sum_{\substack{m\neq m', m\neq n,\\n\neq n',n'\neq m',\\(m-m'-n+n')\text{ mod}\;d=0}}\!\!\!\!\!\! \left|\sum_{\alpha=0}^{M-1}\omega^{\alpha(m^2-m'^2-n^2+n'^2)}\right|.
\end{flalign}
Using a variant of the Dirichlet kernel~\cite{Dirichlet1829},
\begin{equation}
    \sum_{\alpha=0}^{M-1} e^{i\alpha x} = e^{\frac{i(M-1)x}{2}} \frac{\sin\left(\frac{Mx}{2}\right)}{\sin\left(\frac{x}{2}\right)},\quad\forall\;M\in\mathbbm{N},
\end{equation}
and setting $x=\frac{2\pi}{d}(m^2-m'^2-n^2+n'^2)$, we have
\begin{equation}
    \left|\sum_{\alpha=0}^{M-1}\omega^{\alpha(m^2-m'^2-n^2+n'^2)}\right| = \frac{\left|\sin\left(\frac{\pi M(m^2-m'^2-n^2+n'^2)}{d}\right)\right|}{\left|\sin\left(\frac{\pi(m^2-m'^2-n^2+n'^2)}{d}\right)\right|}.
\end{equation}
Combining all terms of $\tilde{F}\suptiny{1}{0}{(M)}= F_1 + \tilde{F}_2\suptiny{1}{0}{(M)}$, we obtain
\begin{equation}
    \tilde{F}\suptiny{1}{0}{(M)} = 1-p+\frac{p}{d^2} - \frac{p}{Md^3}D\suptiny{1}{0}{(M)}_d,
\end{equation}
where we define
\begin{equation}
    D\suptiny{1}{0}{(M)}_d \coloneqq \!\!\!\!\!\!\!\!\!\sum_{\substack{m\neq m', m\neq n,\\n\neq n',n'\neq m',\\(m-m'-n+n')\text{ mod}\;d=0}}\!\!\!\!\!\!\!\!\! \frac{\left|\sin\left(\frac{\pi M(m^2-m'^2-n^2+n'^2)}{d}\right)\right|}{\left|\sin\left(\frac{\pi(m^2-m'^2-n^2+n'^2)}{d}\right)\right|}.
    \label{eq:DirichletSum}
\end{equation}
For all odd prime $d$, $m^2-m'^2-n^2+n'^2 \neq 0$ if $m, m', n, n'$ satisfy all the constraints under the summation symbol in Eq.\;\eqref{eq:DirichletSum}~\cite{BavarescoEtAl2018}. Hence, $D\suptiny{1}{0}{(d)}_d = 0$ for all odd prime $d$\@.\\[-3mm]

To witness the Schmidt number of $\rho\suptiny{1}{0}{\mathrm{iso}}$ to be at least $k+1$, $\tilde{F}\suptiny{1}{0}{(M)}$ must satisfy
\begin{equation}
    \tilde{F}\suptiny{1}{0}{(M)} > \widetilde{\mathcal{B}}_k(\ket{\psi}= \ket{\Phi^+_d}) = \frac{k}{d},
    \label{ineq:JessicaWitnessIsoSchmidtNum}
\end{equation}
where $\widetilde{\mathcal{B}}_k(\ket{\psi})$ is defined in Eq.\;\eqref{ineq:JessicaBk}. With some simple algebra, we see that Ineq.\;\eqref{ineq:JessicaWitnessIsoSchmidtNum} is satisfied if and only if the white-noise ratio in $\rho\suptiny{1}{0}{\mathrm{iso}}$ satisfies
\begin{equation}
    p < \frac{d(d-k)}{d^2-1+\frac{D\suptiny{1}{0}{(M)}_d}{Md}} \eqcolon \tilde{p}_{\text{iso},M}\suptiny{1}{0}{(k)}\quad\forall\;1\leq M\leq d.
\end{equation}
Since $D\suptiny{1}{0}{(M)}_d\geq 0$ for all $M,d\in\mathbbm{N}$, $\tilde{p}_{\text{iso},M}\suptiny{1}{0}{(k)} \leq p_\text{iso}\suptiny{1}{0}{(k)} \coloneqq \frac{d(d-k)}{d^2-1}$ where $p_\text{iso}\suptiny{1}{0}{(k)}$ is the minimum white-noise ratio for $\rho\suptiny{1}{0}{\mathrm{iso}}$ to have Schmidt number~$k$~\cite{TerhalHorodecki2000}. Moreover, for odd-prime dimensions $d$, $\tilde{p}_{\text{iso},M=d}\suptiny{1}{0}{(k)} = p_\text{iso}\suptiny{1}{0}{(k)}$ since $D\suptiny{1}{0}{(d)}_d = 0$\@.

\vspace*{1.5mm}

\subsubsection{Witnessing noisy purified thermal states}\label{app:JessicaWitnessThermal}\vspace*{-1.5mm}
In this appendix, we will derive $\tilde{F}\suptiny{1}{0}{(M)} = F_1 + \tilde{F}_2\suptiny{1}{0}{(M)}$ for the noisy thermal state $\rho_\text{th}(p,\beta) = (1-p)\ket{\psi_\text{th}(\beta)}\!\!\bra{\psi_\text{th}(\beta)} + \frac{p}{d^2}\identity_{d^2}$ with $\ket{\psi_\text{th}(\beta)} = \frac{1}{\sqrt{\mathcal{Z}}}\sum_{n=0}^{d-1} e^{-\frac{\beta n}{2}}\ket{n}\otimes \ket{n}$, where $\mathcal{Z}=\sum_{n=0}^{d-1} e^{-\beta n}$\@. Since $\rho_\text{th}(p,\beta=0) = \rho\suptiny{1}{0}{\mathrm{iso}}$ and the corresponding fidelity bound $\tilde{F}\suptiny{1}{0}{(M)}$ and Schmidt-number witness bound $\widetilde{\mathcal{B}}_k(\ket{\psi})$ as defined in Eq.\;\eqref{ineq:JessicaBk} will be the same as in Sec.\;\ref{app:JessicaWitnessIsotropic}, we will consider only the case with $\beta>0$ in this appendix. Knowing that
\begin{equation}
    \langle ij|\rho_\text{th}|ij\rangle =\frac{1-p}{\mathcal{Z}}e^{-\beta j}\delta_{ij}+\frac{p}{d^2} \quad\forall\;i,j\in[d],\label{eq:thermaldiag}
\end{equation}
the target state $\ket{\psi}$ has Schmidt coefficients
\begin{equation}
    \lambda_j = \sqrt{\frac{(1-p)\frac{e^{-\beta j}}{\mathcal{Z}}+\frac{p}{d^2}}{1-p+\frac{p}{d}}}\label{eq:JessicaThermal_lambdas}
\end{equation}
according to Eq.\;\eqref{eq:JessicaTargetLambda}. Note that the ``tilted" bases in this case do not form an orthogonal bases for $\beta>0$\@. Using Eqs.\;\eqref{eq:thermaldiag} and \eqref{eq:JessicaThermal_lambdas}, we can easily obtain
\begin{equation}
    F_1 = \frac{1}{1-p+\frac{p}{d}}\left[\frac{\left(1-p\right)^2\tanh\frac{\beta}{2}}{\tanh\frac{d\beta}{2}}+\frac{2p(1-p)}{d^2}+\frac{p^2}{d^3}\right],
\end{equation}
along with
\begin{equation}
    \Sigma\suptiny{1}{0}{(M)} \,=\, \frac{d(1-p)}{\mathcal{Z}(\sum_j\lambda_j)^2}\left(\sum_{n=0}^{d-1}\lambda_n e^{-\frac{\beta n}{2}}\right)^2 + \frac{p}{d},
\end{equation}
and
\begin{flalign}
    &\sum_{m,n=0}^{d-1} \lambda_m \lambda_{n} \langle mn|\rho_\text{th}|mn\rangle\\
    &\ = \frac{1-p}{1-p+\frac{p}{d}}\left[\frac{(1-p)\tanh\frac{\beta}{2}}{\tanh\frac{d\beta}{2}}+\frac{p}{d^2}\right] +\frac{p(\sum_j\lambda_j)^2}{d^2},\nonumber
\end{flalign}
as well as the expressions
\begin{flalign}
    &\sum_{\substack{m\neq m', m\neq n,\\n\neq n',n'\neq m'}} \tilde{\gamma}\suptiny{1}{0}{(M)}_{mm'nn'}\sqrt{\langle m'n'|\rho_\text{th}|m'n'\rangle\!\langle mn|\rho_\text{th}|mn\rangle}\nonumber\\ 
    &\qquad\qquad =\, \frac{p}{Md^2}\overline{D}\suptiny{1}{0}{(M)}_d(p,\beta),
\end{flalign}
and finally
\begin{flalign}
    \overline{D}\suptiny{1}{0}{(M)}_{d}\nl(p,\beta) \,\coloneqq 
    \!\!\!\!\!\!\!\!\!\!\!\!\!\!\!\!\!\!\!\!\!\!\!\sum_{\substack{m\neq m', m\neq n,\\n\neq n',n'\neq m',\\(m-m'-n+n')\text{ mod}\;d=0}}
    \!\!\!\!\!\!\!\!\!\!\!\!\!\!\!\!\!\!\!\!\!\!\sqrt{\lambda_m \lambda_{n}\lambda_{m'} \lambda_{n'}}
    \,\frac{\left|\sin\left(\frac{\pi M(m^2-m'^2-n^2+n'^2)}{d}\right)\right|}{\left|\sin\left(\frac{\pi(m^2-m'^2-n^2+n'^2)}{d}\right)\right|},
\end{flalign}
where we again use the Dirichlet kernel~\cite{Dirichlet1829} as in the previous appendix. Combining all terms of $\tilde{F}\suptiny{1}{0}{(M)}= F_1 + \tilde{F}_2\suptiny{1}{0}{(M)}$ and after some simplification, we obtain
\begin{flalign}
    \tilde{F}\suptiny{1}{0}{(M)} = \frac{p}{d^2}\left[1-\frac{\overline{D}\suptiny{1}{0}{(M)}_d(p,\beta)}{M}\right] + (1-p)\kappa_d(p,\beta), \label{eq:JessicaF^M_thermal}
\end{flalign}
where $\kappa_d(p,\beta) \coloneqq \frac{1}{\mathcal{Z}}\left(\sum_{n=0}^{d-1}\lambda_n e^{-\frac{\beta n}{2}}\right)^2$\@.\\[-3mm]

To witness the Schmidt number of $\rho_\text{th}(p,\beta)$ to be at least $k+1$, $\tilde{F}\suptiny{1}{0}{(M)}$ must satisfy $\tilde{F}\suptiny{1}{0}{(M)} > \widetilde{\mathcal{B}}_k(\ket{\psi})$ where
\begin{flalign}
    \widetilde{\mathcal{B}}_k(\ket{\psi}) &= \sum_{j=0}^{k-1} \frac{(1-p)\frac{e^{-\beta j}}{\mathcal{Z}}+\frac{p}{d^2}}{1-p+\frac{p}{d}}\nonumber\\
    &= \frac{1}{1-p+\frac{p}{d}}\left[\frac{(1-p)(1-e^{-k\beta})}{1-e^{-d\beta}}+\frac{pk}{d^2}\right],
    \label{eq:JessicaWitnessThermalBound}
\end{flalign}
as defined in Eq.\;\eqref{ineq:JessicaBk}. We can then numerically solve for the white-noise threshold ratio $\tilde{p}_{\text{th},M}\suptiny{1}{0}{(k)}$ such that $\tilde{F}\suptiny{1}{0}{(M)} > \widetilde{\mathcal{B}}_k(\ket{\psi})$ is satisfied for all $p<\tilde{p}_{\text{th},M}\suptiny{1}{0}{(k)}$ (see Fig.\;\ref{fig:compareJessica_SN_th}).

\vspace*{-1.5mm}

\subsection{Random measurement bases for witnessing Schmidt number in high dimensions}\label{app:COM_Levy}\vspace*{-1.5mm}

In this appendix, we will use concentration of measure to show that the likelihood of a set of orthonormal bases that are chosen uniformly at random in $\mathbbm{C}^d$ to be biased decreases with the dimension $d$ exponentially, i.e., proving Ineq.\;\eqref{ineq:bases_COM}.\\[-3mm]

First of all, we remark that there exists a bijection between all normalized pure states in $\mathbbm{C}^d$ and the $(2d-1)$-sphere, $\mathbbm{S}^{2d-1}=\{\vec{x}\in\mathbbm{R}^{2d}:||\vec{x}||_2 = 1\}$ such that any normalized pure state in $\mathbbm{C}^d$ can be written as
\begin{equation}
    \ket{\vec{x}} = \sum_{k=0}^{d-1} (x_{2k+1} +i x_{2k+2})\ket{k}, \label{eq:C^d<->S^2d-1}
\end{equation}
where $\vec{x}=\sum_{j=1}^{2d}x_j \vec{e}_j \in \mathbbm{S}^{2d-1}$ and $\vec{e}^T_i\vec{e}_j=\delta_{ij}\;\forall\;i,j$\@.\\[-3mm]

Suppose that the first local measurement basis is $\{\ket{i}\}_{i=0}^{d-1}$ and the $z$-th basis with $z\in\{2,\ldots,m\}$ is chosen to be $\{\ket{e^z_a}\}_{a=0}^{d-1}$ such that $\ket{e^z_0}\equiv \ket{\vec{x}\suptiny{1}{0}{(z)}_0}$ with $\vec{x}\suptiny{1}{0}{(z)}_0$ sampled uniformly at random on $\mathbbm{S}^{2d-1}$ and $\ket{e^z_a}$ for $a=1,\ldots,d-1$ can be constructed by the Gram-Schmidt process using $\{\ket{i}\}_{i=0}^{d-1}$ as the reference basis. Let us define the set of rotation matrices $\{R\suptiny{1}{0}{(z)}_a\}_{a=1}^{d-1}\subset O(2d)$ such that $\ket{e^z_a} \equiv \ket{\vec{x}\suptiny{1}{0}{(z)}_a} \equiv \ket{R\suptiny{1}{0}{(z)}_a\vec{x}\suptiny{1}{0}{(z)}_0}$\@.
Using the bijection in Eq.\;\eqref{eq:C^d<->S^2d-1}, the overlap between the first and the $z$-th bases is given by
\begin{equation}
    f_j(\vec{x}\suptiny{1}{0}{(z)}_a) \coloneqq |\langle j|\vec{x}\suptiny{1}{0}{(z)}_a\rangle|^2 = (x\suptiny{1}{0}{(z)}_a)^2_{2j+1} + (x\suptiny{1}{0}{(z)}_a)^2_{2j+2} \label{eq:basesOverlapFunc_RandSample}
\end{equation}
for $a, j\in\{0,\ldots,d-1\}$\@.\\[-3mm]

To find the average of the bases overlap $\mathbbm{E}_{\mathbbm{S}^{2d-1}}[f_j(\vec{x}\suptiny{1}{0}{(z)}_a)]$, we integrate $f_j(\vec{x}\suptiny{1}{0}{(z)}_a)$ over $\mathbbm{S}^{2d-1}$ in the spherical coordinates with respect to the uniform spherical measure $\mu$ on $\mathbbm{S}^{2d-1}$\@. Since $\mu$ is invariant under any rotation $R\in O(2d)$ and $\{\vec{x}\suptiny{1}{0}{(z)}_a\}_{a,z}$ are all related by some rotations in $O(2d)$,
\begin{equation}
    \mathbbm{E}_{\mathbbm{S}^{2d-1}}[f_j(\vec{x}\suptiny{1}{0}{(z)}_a)] = \mathbbm{E}_{\mathbbm{S}^{2d-1}}[f_j(\vec{x})]\quad\forall\;a,z, \label{eq:AverageOverlap_BasisIndep}
\end{equation}
where $\vec{x}\in\mathbbm{S}^{2d-1}$\@. Hence, we can drop the indices $a$ and $z$ in the integral. As a remark, the average $\mathbbm{E}_{\mathbbm{S}^{2d-1}}[f_j(\vec{x})]$ is not affected if $\ket{\vec{x}}$ is defined in another basis, i.e., $\ket{\vec{x}} = \sum_{k=0}^{d-1} (x_{2k+1} +i x_{2k+2})\ket{\tilde{k}}$ where $\{\ket{i}\}_i\neq\{\ket{\tilde{k}}\}_k$, as there exist $\vec{x}'=R \vec{x}\in\mathbbm{S}^{2d-1}$ with $R\in O(2d)$ such that $\ket{\vec{x}} = \sum_{k=0}^{d-1} (x'_{2k+1} +i x'_{2k+2})\ket{k}$\@.\\[-3mm] 

Moreover, the average overlap between the $z$-th and $z'$-th bases for any $z,z'\in\{2,\ldots,m\}$ also equals to $\mathbbm{E}_{\mathbbm{S}^{2d-1}}[f_j(\vec{x})]$\@. To illustrate that, we consider the overlap $|\langle \vec{x}\suptiny{1}{0}{(z)}_a|\vec{x}\suptiny{1}{0}{(z')}_{b}\rangle|^2$\@. Since there exists a unitary $U\in U(d)$ such that $U\ket{\vec{x}\suptiny{1}{0}{(z)}_a} = \ket{j}$ and $U\ket{\vec{x}\suptiny{1}{0}{(z')}_{b}} = \ket{R_U\vec{x}\suptiny{1}{0}{(z')}_{b}}$ where $R_U\in O(2d)$ and $R_U\vec{x}\suptiny{1}{0}{(z')}_{b}\in\mathbbm{S}^{2d-1}$, we have
\begin{flalign}
    |\langle \vec{x}\suptiny{1}{0}{(z)}_a|\vec{x}\suptiny{1}{0}{(z')}_{b}\rangle|^2 =|\langle \vec{x}\suptiny{1}{0}{(z)}_a|U^\dagger U|\vec{x}\suptiny{1}{0}{(z')}_{b}\rangle|^2 = |\langle j|R_U\vec{x}\suptiny{1}{0}{(z')}_{b}\rangle|^2,
\end{flalign}
and the average of the bases overlap is
\begin{flalign}
    &\int_{\vec{x}\suptiny{1}{0}{(z')}_{b}\in\mathbbm{S}^{2d-1}} |\langle \vec{x}\suptiny{1}{0}{(z)}_a|\vec{x}\suptiny{1}{0}{(z')}_{b}\rangle|^2 d\mu(\mathbbm{S}^{2d-1})\\
    &=\, \int_{\vec{x}\suptiny{1}{0}{(z')}_{b}\in\mathbbm{S}^{2d-1}} |\langle j|R_U\vec{x}\suptiny{1}{0}{(z')}_{b}\rangle|^2 d\mu(\mathbbm{S}^{2d-1})\nonumber\\
    &=\, \int_{\vec{x}\suptiny{1}{0}{(z')}_{b}\in\mathbbm{S}^{2d-1}} |\langle j|\vec{x}\suptiny{1}{0}{(z')}_{b}\rangle|^2 d\mu(R_U^{-1}\mathbbm{S}^{2d-1})\nonumber\\
    &=\, \int_{\vec{x}\suptiny{1}{0}{(z')}_{b}\in\mathbbm{S}^{2d-1}} f_j(\vec{x}\suptiny{1}{0}{(z')}_{b}) d\mu(\mathbbm{S}^{2d-1})\nonumber\\ 
    &=\, \mathbbm{E}_{\mathbbm{S}^{2d-1}}[f_j(\vec{x}\suptiny{1}{0}{(z')}_b)] = \mathbbm{E}_{\mathbbm{S}^{2d-1}}[f_j(\vec{x})]\,,\nonumber
\end{flalign}
where we use the invariance of the spherical measure $\mu(R_U^{-1}\mathbbm{S}^{2d-1}) = \mu(\mathbbm{S}^{2d-1})$ and Eq.\;\eqref{eq:AverageOverlap_BasisIndep}. Therefore, the overlaps between all $m$ bases (with $m-1$ randomly chosen ones) will all follow the same statistics, so we can consider only the overlaps between the first two bases from now on without loss of generality.\\[-3mm]

To perform the integration for averaging $f_j(\vec{x})$, we use the standard conversion between the Cartesian, of which $\vec{x}$ is written in Eqs.\;\eqref{eq:C^d<->S^2d-1} and \eqref{eq:basesOverlapFunc_RandSample}, and the spherical coordinates~\cite{Blumenson1960} on a unit sphere, i.e.,
\begin{subequations}
\begin{flalign}
    x_{1\leq j\leq 2d-1} &= \cos(\varphi_j) \prod_{i=1}^{j-1} \sin(\varphi_i),\label{eq:spherical1}\\
    x_{2d} &= \prod_{i=1}^{2d-1} \sin(\varphi_i),\label{eq:spherical2}
\end{flalign}
\end{subequations}
where $\varphi_1,\ldots,\varphi_{2d-2}\in[0,\pi]$ and $\varphi_{2d-1}\in[0,2\pi)$\@. Using the expression of the surface area element of $\mathbbm{S}^{2d-1}$ in the spherical coordinates~\cite{Blumenson1960},
\begin{equation}
    dA_{\mathbbm{S}^{2d-1}} = d\varphi_{2d-1} \prod_{i=1}^{2d-2} \sin^{2d-i-1}\varphi_i d\varphi_i,
\end{equation}
together with the identities
\begin{equation}
    \int_0^\pi \sin^\alpha\theta\;d\theta = \sqrt{\pi}\frac{\Gamma\left(\frac{1+\alpha}{2}\right)}{\Gamma\left(1+\frac{\alpha}{2}\right)} \quad\forall\;\text{Re}(\alpha)>-1,
\end{equation}
and
\begin{equation}
    \prod_{j=1}^k \frac{\Gamma\left(d-\frac{j}{2}\right)}{\Gamma\left(d+1-\frac{j+1}{2}\right)} = \frac{\Gamma\left(d-\frac{k}{2}\right)}{\Gamma\left(d\right)},
\end{equation}
with $\Gamma(n)=(n-1)!\;\forall\;n\in\mathbbm{N}$, the total area of $\mathbbm{S}^{2d-1}$ is
\begin{flalign}
    A_{\mathbbm{S}^{2d-1}} &= \prod_{i=1}^{2d-2} \int_0^\pi \sin^{2d-i-1}\varphi_i d\varphi_i \int_0^{2\pi} d\varphi_{2d-1} = \frac{2\pi^d}{(d-1)!}\,.
\end{flalign}
Next, we use Eqs.\;\eqref{eq:spherical1} and \eqref{eq:spherical2} to write
\begin{equation}
    f_j(\vec{x}) = \prod_{i=1}^{2j} \sin^{2}\varphi_i (1-\sin^{2}\varphi_{2j+1}\sin^{2}\varphi_{2j+2})
\end{equation}
in the spherical coordinates, where we define $\varphi_{2d}=0$, and by using
\begin{equation}
    \prod_{j=1}^{2k} \frac{\Gamma\left(d+1-\frac{j}{2}\right)}{\Gamma\left(d+1-\frac{j-1}{2}\right)} = \frac{(d-k)!}{d!},
\end{equation}
we can calculate the average bases overlap to be
\begin{flalign}
    \mathbbm{E}_{\mathbbm{S}^{2d-1}}[f_j(\vec{x})] &= \int_{\mathbbm{S}^{2d-1}} f_j(\vec{x}) \frac{dA_{\mathbbm{S}^{2d-1}}}{A_{\mathbbm{S}^{2d-1}}} = \frac{1}{d}
\end{flalign}
for all $j\in\{0,\ldots,d-1\}$\@. This shows that the average bases overlaps over the uniform spherical measure are all equal to the bases overlaps of MUBs.\\[-3mm]

Next, we apply L\'{e}vy's lemma which upper bounds the likelihood that the overlaps between a randomly chosen basis and the fixed basis or between two randomly chosen bases deviate from $\frac{1}{d}$ by $\epsilon> 0$\@.
\begin{lemma}[L\'{e}vy's Lemma~\cite{MilmanSchechtman1986,Ledoux2001,HaydenLeungWinter2006}]
    Let $f:\mathbbm{S}^{d'-1}\mapsto \mathbbm{R}$ be a Lipschitz continuous function such that $\exists\;K\geq0$,
    \begin{equation}
        |f(\vec{x})-f(\vec{y})|\leq K||\vec{x}-\vec{y}||_2 \quad\forall\;\vec{x},\vec{y}\in\mathbbm{S}^{d'-1}.\label{eq:Lipschitz}
    \end{equation}
    For a point $\vec{x}\in\mathbbm{S}^{d'-1}$ chosen uniformly at random,
    \begin{equation}
        \text{Pr}\{|f(\vec{x})-\mathbbm{E}[f]|>\epsilon\} \leq 2 \exp\left(-\frac{Cd'\epsilon^2}{K^2}\right),
    \end{equation}
    where $\mathbbm{E}[f]$ is the mean value of $f$ and the constant $C>0$ may take the value $C=(9\pi^3\ln2)^{-1}$\@.
\end{lemma}

In order to use L\'{e}vy's lemma, we must show that $f_j(\vec{x})$ is Lipschitz continuous. To show this, we return to the Cartesian coordinate and define $P_{i,j}=\vec{e}_i\cdot\vec{e}_i^T + \vec{e}_j\cdot\vec{e}_j^T$ for $i\neq j$, $X=\vec{x}\cdot\vec{x}^T$, and $Y=\vec{y}\cdot\vec{y}^T$\@. For any $\vec{x},\vec{y}\in\mathbbm{S}^{2d-1}$,
\begin{flalign}
    |f_j(\vec{x})-f_j(\vec{y})| &= |x^2_{2j+1} + x^2_{2j+2} - y^2_{2j+1} - y^2_{2j+2}|\nonumber\\
    &= |\tr[(X-Y)P_{2j+1,2j+2}]|\\
    &\leq ||X-Y||_\text{HS}\nr||P_{2j+1,2j+2}||_\text{HS},
    \nonumber
\end{flalign}
where we use the Cauchy{\textendash}Schwarz inequality and the Hilbert-Schmidt norm, $||\cdot||_\text{HS}$\@. Since $||P_{i,j}||_\text{HS}=\sqrt{2}$,
\begin{flalign}
    |f_j(\vec{x})-f_j(\vec{y})| &\leq \sqrt{2\;\tr[(X-Y)^\dagger(X-Y)]}\nonumber\\
    &=\sqrt{2[(\vec{x}^T\vec{x})^2+(\vec{y}^T\vec{y})^2-2|\vec{x}^T\vec{y}|^2]}\\
    &= 2\sqrt{1-|\vec{x}^T\vec{y}|^2},\nonumber
\end{flalign}
where we use the fact that $\vec{x},\vec{y}\in\mathbbm{S}^{2d-1}$\@. Then, knowing that $\vec{x}^T\vec{y}=\sum_{j=1}^{2d}x_jy_j\leq 1$, we have
\begin{flalign}
    |f_j(\vec{x})-f_j(\vec{y})| &\leq\, 2\sqrt{\Bigl(1-\sum_{i=1}^{2d}x_iy_i\Bigr)\Bigl(1+\sum_{j=1}^{2d}x_jy_j\Bigr)}\nonumber\\
    &\leq\, 2\sqrt{2\Bigl(1-\sum_i x_iy_i\Bigr)}\,=\,2 ||\vec{x}-\vec{y}||_2\,.
\end{flalign}

Hence, the Lipschitz constant in Eq.\;\eqref{eq:Lipschitz} is $K=2$ and by setting $d'=2d$, L\'{e}vy's lemma tells us that
\begin{equation}
    \text{Pr}\left\{\left|f_{j}(\vec{x})-\frac{1}{d}\right|\,>\,\epsilon\right\} \leq 2 \exp\left(-\frac{d\epsilon^2}{18\pi^3\ln2}\right).\label{eq:bases_COM}
\end{equation}
This implies that in large dimensions $d$, it is very likely to find two randomly chosen bases to be close to mutually unbiased. Since we can still certify entanglement in high dimensions with measurement bases that are not MUBs with our witness, Ineq.\;\eqref{eq:bases_COM} suggests that it may not worth the effort to precisely control the relative phases of the measurement bases to ensure that they are mutually unbiased in large dimensions.


\subsection{Maximal number of orthonormal bases with overlap constraints}\label{app:onbNumBound}\vspace*{-1.5mm}

In this appendix, we will prove the following corollary which relates the maximal number of orthonormal bases that obey certain overlap constraints to the function $\lambda(\mathcal{C})$ defined in Theorem~\ref{thm:SchmidtRkBound}.
\begin{corollary}\label{corollary:MaxBasesNum}
    Suppose that there are $m$ orthonormal bases with overlaps $\mathcal{C}=\{|\langle e^z_a|e^{z'}_{a'}\rangle|^2\}_{a,a',z\neq z'}$\@. Then,
    \begin{equation}
        m \leq \frac{d+1}{2}\,\left(1+\sqrt{1+\frac{8\lambda(\mathcal{C})\bigl(\lambda(\mathcal{C})-1\bigr)}{d^2-1}}\,\right) \,\eqcolon\, \overline{m}_d
    \end{equation}
    for all $d\geq2$, where $\lambda(\mathcal{C})$ is defined in Theorem~\ref{thm:SchmidtRkBound}.
\end{corollary}

We first give the formal statement of the Welch bounds before proving the corollary. These bounds limit the maximal number of unit vectors compatible with certain overlaps between them.

\begin{lemma}[Welch bounds~\cite{Welch1974}]\label{lemma:WelchBounds}
    Let $\{\ket{\psi_i}\}_{i=1}^M\subset \mathbbm{C}^d$ be a set of $M$ unit vectors. Then, for all $k\geq 1$,
    \begin{equation}\label{ineq:Welch}
        \sum_{i,j=1}^M |\langle\psi_i|\psi_j\rangle|^{2k} \geq \frac{M^2}{\binom{d+k-1}{k}}.
    \end{equation}
\end{lemma}

\begin{proof}[Proof of Corollary~\ref{corollary:MaxBasesNum}]
    Our goal is to upper bound the number of orthonormal bases, $m$, for a given value of $\lambda(\mathcal{C})$ that is defined in Theorem~\ref{thm:SchmidtRkBound}. Let us apply Lemma~\ref{lemma:WelchBounds} by setting $M=md$ and $\{\ket{\psi_i}\}_{i=1}^M = \{|e^{z}_{a}\rangle\}_{a,z}$\@. The summation term in Ineq.\;\eqref{ineq:Welch} can be split into two terms
    \begin{align}
        \sum_{i,j=1}^M |\langle\psi_i|\psi_j\rangle|^{2k}= \sum_{z}\sum_{a,a'}|\langle e^z_a|e^{z}_{a'}\rangle|^{2k} + \sum_{z\neq z'}\sum_{a,a'}|\langle e^z_a|e^{z'}_{a'}\rangle|^{2k}\,.
    \end{align}
    Now we set $k=2$, then the first term evaluates to $\sum_{z}\sum_{a,a'}|\delta_{a,a'}|^4 = md$ and the second term is related to $\lambda(\mathcal{C})$ by 
    \begin{align}
        \sum_{\substack{z\neq z'\\ a,a'}}|\langle e^z_a|e^{z'}_{a'}\rangle|^4 = 2\lambda(\mathcal{C})\bigl(\lambda(\mathcal{C})-1\bigr) + d\sum_{z\neq z'}[(d+1)c^{z,z'}_\text{min}-1]\,.
    \end{align}
    Using the fact that $c^{z,z'}_\text{min}\leq \frac{1}{d}\;\forall\;z,z'$ and after some algebra, we obtain
    \begin{equation}
        m^2-(d+1)m-\frac{2(d+1)}{d-1}\lambda(\mathcal{C})\bigl(\lambda(\mathcal{C})-1\bigr)\,\leq\, 0,
    \end{equation}
    which implies the inequality
    \begin{equation}
        m\,\leq\, \frac{d+1}{2}\,\left(1+\sqrt{1+\frac{8\lambda(\mathcal{C})\bigl(\lambda(\mathcal{C})-1\bigr)}{d^2-1}}\,\right) \,\eqcolon\, \overline{m}_d
    \end{equation}
    for all $d\geq2$, as stated in Corollary~\ref{corollary:MaxBasesNum}.
\end{proof}

\vspace*{-1.5mm}

\subsection{Proof of Lemma~\ref{lemma:3MUBsAnyD}}\label{app:3MUBsAnyD}\vspace*{-1.5mm}

In this appendix, we present the proof of Lemma 2 from the main article, which we restate here for clarity.

\begin{replemma}{lemma:3MUBsAnyD}
    For any $d\in \mathbbm{N}$, the three orthonormal bases $\{\ket{e^1_a}=\ket{a}\}_{a=0}^{d-1}$, $\{\ket{e^2_a}\}_{a=0}^{d-1}$, and $\{\ket{e^3_a}\}_{a=0}^{d-1}$, with
    \begin{subequations}
    \begin{flalign}
        \ket{e^2_a} & =\frac{1}{\sqrt{d}}\sum_{j=0}^{d-1} e^{i2\pi\left[\frac{aj}{d}+f(j)\right]}\ket{j}, \\
        \ket{e^3_a} & =\frac{1}{\sqrt{d}}\sum_{j=0}^{d-1} e^{i2\pi\left[\frac{(d-p^r) j^2}{2d}+\frac{aj}{d}+f(j)\right]}\ket{j},
    \end{flalign}
    \end{subequations}
    where $f$ is any real-valued function, $r\in\mathbbm{N}\cup\{0\}$, and $p$ is any odd prime such that $gcd(d,p)=1$ and $d>p^r$, are mutually unbiased. The simplest example would be having $p^r=1$.
\end{replemma}
We remark that the freedom in defining the relative phases $\{e^{i2\pi f(j)}\}_j$ in both bases can be thought of as the global phase freedom in defining the computational basis. We also note that the quadratic phases in our construction and our bases unbiasedness can be related to properties of stabilizer states in odd dimensions \cite{Gross2006,KuengGross2015}.

To prove this, we will need the following propositions and observations, where Propositions~\ref{prop:GaussSum} and~\ref{prop:ReciprocityThm_GaussSum} were proven in Refs.\;\cite{Herrera-ValenciaSrivastavPivoluskaHuberFriisMcCutcheonMalik2020} and~\cite{BerndtEvansWilliams1998}, respectively.

\begin{proposition}[Lemma 1 in Ref.\;\cite{Herrera-ValenciaSrivastavPivoluskaHuberFriisMcCutcheonMalik2020}] \label{prop:GaussSum}
Let $a,b\in\mathbbm{Z}$ and $c$ be a positive odd integer such that $gcd(a,c)=1$\@. Then,
\begin{equation}
    G(a,b,c) \coloneqq \sum_{n=0}^{c-1} e^{2\pi i\frac{an^2+bn}{c}} = \varepsilon_c\sqrt{c}\left(\frac{a}{c}\right) e^{-2\pi i\frac{\psi(a)b^2}{c}},
\end{equation}
where $\left(\frac{a}{c}\right)\in\{-1,0,1\}$ is the Jacobi symbol, 
\begin{align}
    \varepsilon_c \,=\, 
    \begin{cases}
        1 &\text{if}\ \ c\equiv 1 \mod 4,\\ i &\text{if}\ \ c\equiv 3 \mod 4,
    \end{cases}
\end{align}
and $\psi(a)\in\mathbbm{Z}$ such that $4\psi(a)a\equiv 1\mod c$\@. Therefore, $|G(a,b,c)|=\sqrt{c}$\@.
\end{proposition}

\begin{proposition}[Reciprocity theorem for generalized Gauss sums {[Theorem 1.2.2 in Ref.\;\cite{BerndtEvansWilliams1998}]}]\label{prop:ReciprocityThm_GaussSum}
Let $a,b,c\in\mathbbm{Z}$ such that $ac\neq0$ and $ac+b$ is even. Then,
\begin{equation}
    S(a,b,c)\coloneqq \sum_{n=0}^{|c|-1} e^{i\pi\frac{an^2+bn}{c}} = \left|\frac{c}{a}\right|^{\frac{1}{2}} e^{i\pi\frac{|ac|-b^2}{4ac}} S(-c,-b,a).
\end{equation}
\end{proposition}

\begin{observation}\label{obs:gcd_oddD}
    For all odd $d\in\mathbbm{N}$, odd prime $p$ such that $gcd(d,p)=1$ and $r\in\mathbbm{N}\cup\{0\}$, $gcd\left(d,\frac{d-p^r}{2}\right)=1$\@.
\end{observation}
\begin{proof}
    For odd $d\in\mathbbm{N}$, $\frac{d-p^r}{2}\in\mathbbm{Z}$\@. Let $gcd\left(d,\frac{d-p^r}{2}\right)=k$\@. Then, there exist $n,m\in\mathbbm{N}$ such that $d=nk$ and $\frac{d-p^r}{2}=mk$\@. Since $k\neq p^q$ for any $q\in\mathbbm{N}$ due to $gcd(d,p)=1$, it follows that the largest $k\in\mathbbm{N}$ that satisfies $(n-2m)k = d-(d-p^r) = p^r$ is $k=1$\@.
\end{proof}

\begin{observation}\label{obs:gcd_evenD}
    For all even $d\in\mathbbm{N}$, odd prime $p$ such that $gcd(d,p)=1$ and $r\in\mathbbm{N}\cup\{0\}$, $gcd\left(d-p^r,\frac{d}{2}\right)=1$\@.
\end{observation}
\begin{proof}
    Let $gcd\left(d-p^r,\frac{d}{2}\right)=k$\@. Then, there exist $n,m\in\mathbbm{N}$ such that $d-p^r=nk$ and $\frac{d}{2}=mk$\@. Since $k\neq p^q$ for any $q\in\mathbbm{N}$ due to $gcd(d,p)=1$, the largest $k\in\mathbbm{N}$ that satisfies $(2m-n)k = d-(d-p^r) = p^r$ is $k=1$\@.
\end{proof}

\begin{proposition}\label{prop:OurGaussSum}
    Let $d\in\mathbbm{N}$, $k\in\mathbbm{Z}$, $r\in\mathbbm{N}\cup\{0\}$, and $p$ be an odd prime such that $gcd(d,p)=1$ and $d>p^r$\@. Then, it holds that
    \begin{equation}
        \mathcal{G}_{d}^{k} \coloneqq \left|\sum_{j=0}^{d-1} e^{2\pi i\left(\frac{(d-p^r)j^2}{2d}+\frac{kj}{d}\right)}\right| = \sqrt{d}.
    \end{equation}
\end{proposition}
\begin{proof}
    For all odd $d\in\mathbbm{N}$, $d-p^r$ is even and $gcd\left(d,\frac{d-p^r}{2}\right)=1$ by Observation~\ref{obs:gcd_oddD}, so we can apply Proposition~\ref{prop:GaussSum} to get
    \begin{flalign}
        \mathcal{G}_{d}^{k} = \left|\sum_{j=0}^{d-1} e^{i2\pi\frac{\left(\frac{d-p^r}{2}\right)j^2+kj}{d}}\right|= \left|G\left(\frac{d-p^r}{2},k,d\right)\right|=\sqrt{d}.
    \end{flalign}

    We now consider the remaining cases with even $d\in\mathbbm{N}$. The reciprocity theorem (Proposition~\ref{prop:ReciprocityThm_GaussSum}) can be applied in our case for all even $d\in\mathbbm{N}_{\geq2}$ because if we let $a=d-p^r$, $b=2k$, and $c=d$, then $ac\neq 0\;\forall$ even $d>1$ and $ac+b=d(d-p^r)+2k$ is even for all $d,k\in\mathbbm{Z}$, odd $p$, and $r\in\mathbbm{N}\cup\{0\}$\@.\\[-3mm]

    For all even $d\in\mathbbm{N}$, we apply Proposition~\ref{prop:ReciprocityThm_GaussSum} to get
    \begin{flalign}
        \mathcal{G}_{d}^{k} = |S(d-p^r,2k,d)|= \sqrt{\frac{d}{d-p^r}}|S(-d,-2k,d-p^r)|, \label{eq:OurGaussSumEven1}
    \end{flalign}
    where we observe that
    \begin{flalign}
        S(-d,-2k,d-p^r) = \sum_{j=0}^{d-p^r-1} e^{-2\pi i\frac{\left(\frac{d}{2}\right)j^2+kj}{d-p^r}} = G\left(-\frac{d}{2},-k,d-p^r\right).\label{eq:OurGaussSumEven2}
    \end{flalign}
    Since $d-p^r$ is odd and $gcd\left(d-p^r,\frac{d}{2}\right)=1$ by Observation~\ref{obs:gcd_evenD}, we can apply Proposition~\ref{prop:GaussSum} to get $\left|G\left(-\frac{d}{2},-k,d-p^r\right)\right|=\sqrt{d-p^r}$\@. Combining this with Eqs.\;\eqref{eq:OurGaussSumEven1} and \eqref{eq:OurGaussSumEven2}, we obtain $\mathcal{G}_{d}^{k} = \sqrt{d}$ for all even $d\geq2$ as well.
\end{proof}

Finally, we can state the proof of Lemma~\ref{lemma:3MUBsAnyD} having proven Proposition~\ref{prop:OurGaussSum}.
\begin{proof}[Proof of Lemma~\ref{lemma:3MUBsAnyD}]
    First of all, all three bases are clearly orthonormal as we know that $\frac{1}{d}\sum_{j=0}^{d-1} e^{i2\pi\frac{(b-a)j}{d}} = \delta_{a,b}$ for all $a,b\in\mathbbm{Z}$ and $d\in\mathbbm{N}_{\geq2}$\@.
    It is obvious that $|\langle e^1_a|e^z_b\rangle| = \frac{1}{\sqrt{d}}$ for all $z\in\{2,3\}$ and $a,b\in[d]\coloneqq\{0,\ldots,d-1\}$\@. To prove that $|\langle e^2_a|e^3_b\rangle| = \frac{1}{\sqrt{d}}$ for all $a,b\in[d]$, it follows that
    \begin{flalign}
        |\langle e^2_a|e^3_b\rangle| = \frac{1}{d}\left|\sum_{j=0}^{d-1} e^{i2\pi\left(\frac{(d-p^r) j^2}{2d}+\frac{(b-a)j}{d}\right)}\right| = \frac{1}{\sqrt{d}},
    \end{flalign}
    where the last equality follows from Proposition~\ref{prop:OurGaussSum} if we set the integer $k=b-a$\@.
\end{proof}

\vspace*{-3mm}

\subsection{Application of AMUBs}\label{app:AMUBs}\vspace*{-1.5mm}

In this appendix, we investigate whether using approximately MUBs (AMUBs) as the local measurement bases can provide any practical advantage in witnessing Schmidt number in local dimensions where the maximum number of MUBs is unknown.\\[-3mm]

Suppose that we can control an arbitrary number of measurement bases with extreme precision. As Corollary~\ref{cor:MUBoptimal} suggests, we should aim at locally measuring in as many MUBs as possible to get the best performance of our witness. Sadly, the total number of MUBs is unknown for dimensions that are not prime powers (e.g., $d=6,10,12$). 
In fact, given the prime factorization of the dimension $d=\prod_j p_j^{n_j}$ with $p_j^{n_j}<p_{j+1}^{n_{j+1}}\;\forall\;j$, (tensor products of) the Wootters{\textendash}Fields construction only guarantees $p_1^{n_1}+1$ MUBs to exist~\cite{ColomerMortimerFrerotFarkasAcin2022}. However, even with such construction, calculating all the relative phases for each basis is non-trivial for large prime powers~\cite{WoottersFields1989}.\\[-3mm]

Alternatively, if we allow our measurement bases to be nearly mutually unbiased, then one can construct $d+1$ AMUBs for any dimension $d$~\cite{ShparlinskiWinterhof2006}. The AMUBs construction in Ref.\;\cite{ShparlinskiWinterhof2006} has a simpler description than the one by Wootters{\textendash}Fields and takes the following form,
\begin{equation}
    \ket{M^{z}_a} = \frac{1}{\sqrt{d}}\sum_{j=1}^d \exp\left[i2\pi\left(\frac{zj^2}{p}+\frac{aj}{d}\right)\right]\ket{j-1}, \label{eq:AMUBsConstruction}
\end{equation}
where $z\in\{1,\ldots,d\}$ labels the basis, $a\in[d]$ labels the basis vector, and $p$ is the smallest prime such that $p\geq d$\@. Together with the standard basis $\{\ket{M^{0}_j}=\ket{j}\}_j$, they form a set of $d+1$ AMUBs. Using a result from Ref.\;\cite{CochranePeral2001}, it was shown that the bases overlaps satisfy~\cite{ShparlinskiWinterhof2006}
\begin{equation}
    |\langle M^{z}_a|M^{z'}_{a'}\rangle|^2 < \frac{4p}{\pi d^2}\left(\ln p+\gamma-\ln\frac{\pi}{2}+\frac{\pi-1}{2p}\right)+\frac{1}{d}
\end{equation}
for all $z\neq z'\in\{1,\ldots,d\}$ and $a,a'\in[d]$, where $\gamma\approx0.577$ is Euler's constant. Also, we have $|\langle M^{z}_a|j\rangle|^2 = \frac{1}{d}$ for all $z\in\{1,\ldots,d\}$ and $a,j\in[d]$\@. By the Bertrand{\textendash}Chebyshev theorem \cite{Bertrand1845,Chebyshev1852}, which says that there exists a prime $p$ such that $n<p\leq 2n$ for all $n\in\mathbbm{N}$, we get
\begin{flalign}
    c^{z,z'}_\text{max}<\frac{8}{\pi d}\left(\ln(2d)+\gamma-\ln\frac{\pi}{2}+\frac{\pi-1}{4d}\right)+\frac{1}{d} \label{ineq:AMUBcmax}
\end{flalign}
for $z\neq z'\in\{1,\ldots,d\}$, which scales as $\frac{8\ln d}{\pi d}+\mathcal{O}\left(\frac{1}{d}\right)$\@. Note that the upper bound in Ineq.\;\eqref{ineq:AMUBcmax} is larger than 1 for $d<9$, so the bound is only relevant for large dimensions. In large local dimensions $d$, since we are not aware of any non-trivial analytic lower bound for $c^{z,z'}_\text{min}$, we have to assume $c^{z,z'}_\text{min}=0$ and due to the loose upper bound for $c^{z,z'}_\text{max}$ in Ineq.\;\eqref{ineq:AMUBcmax}, we cannot witness any non-trivial Schmidt number of any state using Lemma~\ref{lemma:SchmidtRkLoose} as $\overline{\mathcal{T}}(\overline{\mathcal{C}})=m$\@. For small local dimensions $d$, one can easily find the smallest prime $p\geq d$~\cite{OEIS}, so the values of $c^{z,z'}_\text{max}$ and $c^{z,z'}_\text{min}$ should be computable.\\[-3mm]

To see whether using AMUBs as the local measurement bases in our Schmidt-number witness construction can provide any advantage when the maximal set of MUBs is not known, we consider the smallest non-prime-powered dimensions with the least known MUBs: $d=6,10,14$, and 22 \footnote{They are all products of 2 and another odd prime, so the tensor products of the Wootters{\textendash}Fields constructions will only guarantee 3 MUBs in $d=6,10,14$, and 22.}. The smallest primes $p\geq d$ that enter into Eq.\;\eqref{eq:AMUBsConstruction} are 7, 11, 17, and 23, respectively. Our goal is to see if using $4\leq m\leq d+1$ of these AMUBs can witness higher Schmidt numbers than using only the 3 known MUBs. For this purpose, we consider the isotropic state (see Appendix~\ref{app:isotropic}) again and compare the noise tolerances for witnessing each Schmidt number $2\leq k\leq d$ using $4\leq m\leq d+1$ AMUBs with those using only 3 MUBs. Unfortunately, we verify numerically that measuring in any $m\in\{4,\ldots,d+1\}$ AMUBs described by Eq.\;\eqref{eq:AMUBsConstruction} gives zero noise tolerance. Interestingly, we find that by replacing $p$ in Eq.\;\eqref{eq:AMUBsConstruction} with a non-prime value can sometimes improve the noise tolerance slightly. For example, in $d=6$, if $p=7$ is replaced by 7.2, then measuring in a particular subset of 4 ``modified" AMUBs will give strictly positive noise tolerances for all $k\in\{2,\ldots,6\}$, but still not enough to surpass the noise tolerance when using 3 MUBs. In fact, by replacing $p=7$ with 7.2, three of the bases become mutually unbiased. In Appendix~\ref{app:3MUBsAnyD}, we generalize this observation, enabling us to construct 3 MUBs with a simple analytic description in any dimension $d\in\mathbbm{N}$\@.
This observation also suggests that there could be other constructions of AMUBs that are more suited for witnessing Schmidt numbers than the one from Ref.\;\cite{ShparlinskiWinterhof2006}.

\end{document}